%% file: main.tex
\documentclass[10pt,aps,prx,superscriptaddress,nofootinbib,twocolumn]{revtex4-2}

\input{macros}

\def\ketbra#1#2{\vert{#1}\rangle\!\langle{#2}\vert}

\begin{document}

\title{Optimal unitary trajectories under commuting target and cost observables; applications to cooling}

\author{Ralph Silva}
\affiliation{Institute for Theoretical Physics, ETH Zurich, Wolfgang-Pauli-Str. 27, 8093 Z\"urich, Switzerland}
\author{Pharnam Bakhshinezhad}
\affiliation{Vienna Center for Quantum Science and Technology, Atominstitut, TU Wien, 1020 Vienna, Austria}
\author{Fabien Clivaz}
\email{fabien.clivaz@technikum-wien.at}
\affiliation{Vienna Center for Quantum Science and Technology, Atominstitut, TU Wien, 1020 Vienna, Austria}
\affiliation{Fachhochschule Technikum Wien, H\"ochst\"adtplatz 6, 1200 Wien}

\begin{abstract}

The preparation of quantum states, especially cooling, is a fundamental technology for nanoscale devices. The past decade has seen important results related to both the limits of state transformation and the limits to their efficiency --- the quantum versions of the third and second law of thermodynamics. The limiting cases always involve an infinite resource cost, typically machine complexity or time. Realistic state preparation takes into account both a finite size of the machine and constraints on the operations we can perform. In this work, we determine in full generality the optimal operation for a predominant quantum paradigm: state transformation under a single unitary operation upon a finite system, in the case where the observables corresponding to the target (such as ground state probability) and cost (such as dissipation) commute. We then extend this result to the case of having a third, commuting, globally conserved quantity (such as total energy). The results are demonstrated with the paradigmatic example of ground state cooling, for both arbitrary and energy-preserving unitary operations.

\end{abstract}

\maketitle 

\section*{Introduction}

The control of quantum states, particularly the preparation of pure/cold states is of indisputable practical interest. Pure states are necessary to perform accurate measurements \cite{Guryanova_2020}. They are the backbone of quantum computation \cite{nielsen2010quantum}, and play a key role in precise time-keeping \cite{Xuereb_2023,SchwarzhansLockErkerFriisHuber2021} and entanglement generation \cite{Bohr-NJP-2015,Diotallevi-NJP-2024}. In addition, understanding the limitations to cooling gives us insight into thermodynamic laws that govern quantum systems \cite{Masanes_2017_Temperature}.

The first limitation to cooling is \textit{attainability}, which concerns determining the states that can be reached within a given framework. Significant progress has been made toward a more concrete understanding of Nernst’s unattainability principle \cite{Nernst-1906}—also known as the third law—by exploring the coldest achievable temperatures under reasonable assumptions \cite{Allahverdyan-2011, Clivaz-2019, Alhambra-2019, Boykin-2002, Park-2016, Rodriguez-Briones-2017, Rodriguez-Briones-2016, Schulman-1999,Taranto_2020ExponentialImprovement}. While calculating the exact attainable temperature is often overly ambitious, progress is typically illustrated through lower bounds on temperature \cite{Masanes-2017, Wilming-2017, Scharlau-2018,Raeisi-2015}, which serve to refine the third law, and upper bounds, which are realized through specific protocols \cite{Boykin-2002,Park-2016,Raeisi-2015,Rodriguez-Briones-2017,Schulman-1999,soldati2024}. In some instances, a protocol achieves the lower bound, demonstrating its tightness and providing a tangible method for reaching the lowest possible temperature \cite{Rodriguez-Briones-2016, Clivaz-2019, Alhambra-2019}.

Complementary to attainability is the concept of \textit{efficiency}. Cooling requires interactions with an auxiliary system, and inevitably results in heat dissipation into the environment. This issue is critical: even with advancements enabling the control of larger systems -- such as those composed of many qubits -- ineffective heat management can introduce significant noise, severely degrading the performance of the intended quantum task.

Taking the energy cost into account allows for a sharper understanding of the second law of thermodynamics \cite{Brandao-2015, Reeb-2014, Mueller-2018,Cwiklinski-2015}, which itself provides a lower bound on the energy consumption for cooling a quantum system. Protocols have been derived that saturate said lower bound \cite{Landauer_1961}, however all of these suffer from a significant drawback: they assume access to infinite resources of some kind, either an infinitely sized auxiliary system, or an infinite time allowing for an unbounded number of operations \cite{Reeb-2014, Skrzypczyk-2014, Taranto-2021, Clivaz-2019bis}.

In realistic scenarios, such idealized resources are unavailable, prompting growing interest in approaches constrained by finite resources. Recent studies have addressed this by limiting dimensional control and exploring various operations to establish upper bounds, such as the Reeb-Wolf bound or methods based on thermodynamic length \cite{taranto_2024_FiniteResources, lipkabartosik2024MInDissipationinInteracting, Rolandi_2023}. While these provide optimal protocols, they apply in the regime of sufficiently large time steps\footnote{`time' here denotes the number of operations upon the system} or numbers of qudits.

Thus a crucial gap in this direction is to provide optimally efficient protocols for finite-size machines under a single or small number of time steps. In this work we take an important step in filling this gap for the case of a single-shot unitary operation. While motivated by cooling we in fact prove 
the following, more general result. Take a finite system, a target observable for which we have a desired final expectation value, and a cost observable whose desired value we wish to minimise, and which commutes with the target observable. In the case of cooling, the target and cost are respectively the projector upon the ground state subspace of the system (target) and the Hamiltonian of system and machine (cost). We construct the unitary that minimises the cost value given the target value. Moreover, we apply this process to every possible value of the target, producing a trajectory of states and unitaries that spans the entire set of attainable target values. We refer to this as the optimal trajectory.

In many cases, such protocols are subject to additional constraints imposed by global conservation laws. For example, in the context of cooling, the total energy operator is a common conserved quantity, requiring any applied unitary operation to commute with it. After having established our main result, we proceed to extend said scenario to account for these conservation laws, provided the conserved quantity also commutes with both the target and cost observables. We finally solve the above generalised scenario.

\begin{table}
\centering
	\begin{tabular}{ |c||c|c| } 
		\hline
		&\multicolumn{2}{|c|}{Result type}\\
		\hline
		& No-go & Constructive\\
		\hline \hline
		Asymptotically optimal & \cite{Brandao-2015, Mueller-2018, Cwiklinski-2015} & \cite{Reeb-2014, Skrzypczyk-2014, Taranto-2021, Clivaz-2019bis}\\
		\hline
		Optimal for every dimension & \cite{Reeb-2014} & Our contribution\\
		\hline
	\end{tabular}
\caption{A selection of literature upon the efficiency of cooling quantum systems: some containing no-go theorems bounding the efficiency, others with constructive protocols, albeit only optimal in an asymptotic sense. We provide the first constructive protocol optimal for every dimension.}
\label{tablegap}
\end{table}

The remainder of the paper is organised as follows. In Section I we set our notation as well as precisely define our setting. In Section II we expose our main result in intuitive terms. In Section III we provide a proof for the protocol of our main result to be optimal in terms of energy expenditure. In Section IV we generalize our scenario to when an extra commuting conserved quantity is taken into account before concluding in Section V.

\section{The setting and working example} \label{sec:settingjointU}

\subsection{Main problem: optimal unitary transformations}

In this paper we look at a quantum system and two commuting observables that act on that system. We are investigate how to unitarily manipulate said quantum system such that the resulting state achieves a desired (average) value when measuring the first observable while at the same time minimising the (average) value when measuring the second. Answering this question in particular answers the following thermodynamic special case: consider a quantum system that you want to cool with the help of another quantum system. Both systems may interact arbitrarily. what is the minimal energy cost to cooling said system to a given temperature? And what interactions achieve this?

We now set the stage and notation of the problem. To this end, let $R$ be a finite-dimensional quantum system with associated Hilbert space $\mathcal{H}_R$ that is initially in state $\rho \in \mathcal{L}(\mathcal{H}_R)$, and let $\mathcal{O}_A$ and $\mathcal{O}_E$ be a pair of commuting (Hermitian) observables on $\mathcal{H}_R$, labelled `target' and `cost' observables respectively. The expectation values of these observables are real-valued functions on the space of density matrices that we denote by $\mathcal{A}[\sigma] = \Tr \left[ \mathcal{O}_A \sigma \right]$ and $\mathcal{E}[\sigma] = \Tr \left[ \mathcal{O}_E \sigma \right]$. We respectively label them `target' and `cost' functions. 

As $[\mathcal{O}_A,\mathcal{O}_E]=0$, there exists a basis $(\ket{0}, \dots, \ket{d-1})$ of $\mathcal{H}_R$, where $d$ is the dimension of $R$, that is a common eigenbasis of $\mathcal{O}_A$ and $\mathcal{O}_E$, and which we call the \textit{preferred basis}. It follows that both the target and cost functions are linear with respect to the diagonal elements of any state $\sigma$ in this basis. More specifically, we may write the target and cost functions acting on a general state $\sigma$ of our system as
\begin{align}
    \mathcal{A}[\sigma] &= \Tr \big[ \sum_n a_n \ketbra{n}{n} \,\sigma \big] = \sum_n a_n \,p_n = {\bf a} \cdot {\bf p}\\
    \mathcal{E}[\sigma] &= \Tr \big[ \sum_n E_n \ketbra{n}{n}\, \sigma \big] =\sum_n E_n\, p_n = {\bf E} \cdot {\bf p},
\end{align}
where $p_n=\Tr\left[\ketbra{n}{n} \,\sigma\right]$ are referred to as the \textit{populations} of $\sigma$ and ${\bf p}=(p_0,\dots, p_{d-1})$ as the population vector. We have also introduced ${\bf a}=(a_0,\dots, a_{d-1})$, and ${\bf E}=(E_0,\dots, E_{d-1})$, the vectors of eigenvalues of $\mathcal{O}_A$ and $\mathcal{O}_E$, which we will dub the target and cost vectors respectively. Note that we will at times denote the target and cost functions by $\mathcal{A}({\bf p})$ resp. $\mathcal{E}({\bf p})$ when we will want to emphasize their dependency on ${\bf p}$. With this, our question can be formulated as follows: for an arbitrary given value of $\alpha$, what is
\begin{align} \label{equ:main_question_pop}
     \min_U \; \mathcal{E}[\sigma], \quad \text{s.t. } \mathcal{A}[\sigma]=\alpha,
\end{align}
where $\sigma=U \rho \, U^{\dagger}$ and $U$ is a unitary operation? 

\subsection{Optimal trajectories}

Rather than minimising the cost for a single value $\alpha$ of the target, one can do so for the set of all possible values of the target. The answer gives us a function of $\alpha$ which we label the \textit{minimal cost function} 
\begin{equation} \label{equ:MinimalCostFunction}
      \omega_{\text{opt}}(\alpha):=  \min_{\substack{U \\ \text{s.t. } \mathcal{A}[\sigma]= \alpha}} \mathcal{E}[\sigma].
\end{equation}
To gain a better understanding of the situation, it is essential to identify the states that correspond to the minimum values of the cost function. This can be achieved by examining the arguments, specifically the unitaries, that yield these minimal values. Doing so may however seem like a daunting task: the space of unitaries is a large and non trivial one. Luckily enough, thanks to the Schur-Horn theorem \cite{Marshall-2011}, we can get rid of the unitary dependency in Eq.~\ref{equ:main_question_pop} and reformulate our problem solely in terms of doubly stochastic matrices. This can then furthermore be simplified to a problem depending solely on the populations, see Figure~\ref{fig:differentlevels} for a depiction of the above reformulations. We will mostly work in the population vector reformulation. From the population vector $\bf p$ we can then reconstruct the doubly stochastic matrices, unitaries and states $\{ \sigma(\alpha)\}_{\alpha}$ that achieve the minimal cost. Note however that due to potential degeneracies in the cost and target vectors ${\bf a }$ and ${\bf E} $ there might be multiple states $\sigma$ that achieve the minimal cost for a given $\alpha$.
And so, while the minimal cost function is by definition unique, there may be many choices of state trajectories associated to it. In this work we will make a specific choice of such a state trajectory, associating for every achievable $\alpha$ a unique state $\sigma(\alpha)$. That trajectory will be constructed explicitly in Section~\ref{sec:result} and can be understood as our main result. We will refer to it as \textit{the optimal trajectory}. However, our method of constructing the optimal trajectory gives us a solid understanding of how to construct the other state trajectories associated to the minimal cost function. 

\begin{figure}[t]
\centering
\includegraphics[width=0.9\columnwidth]{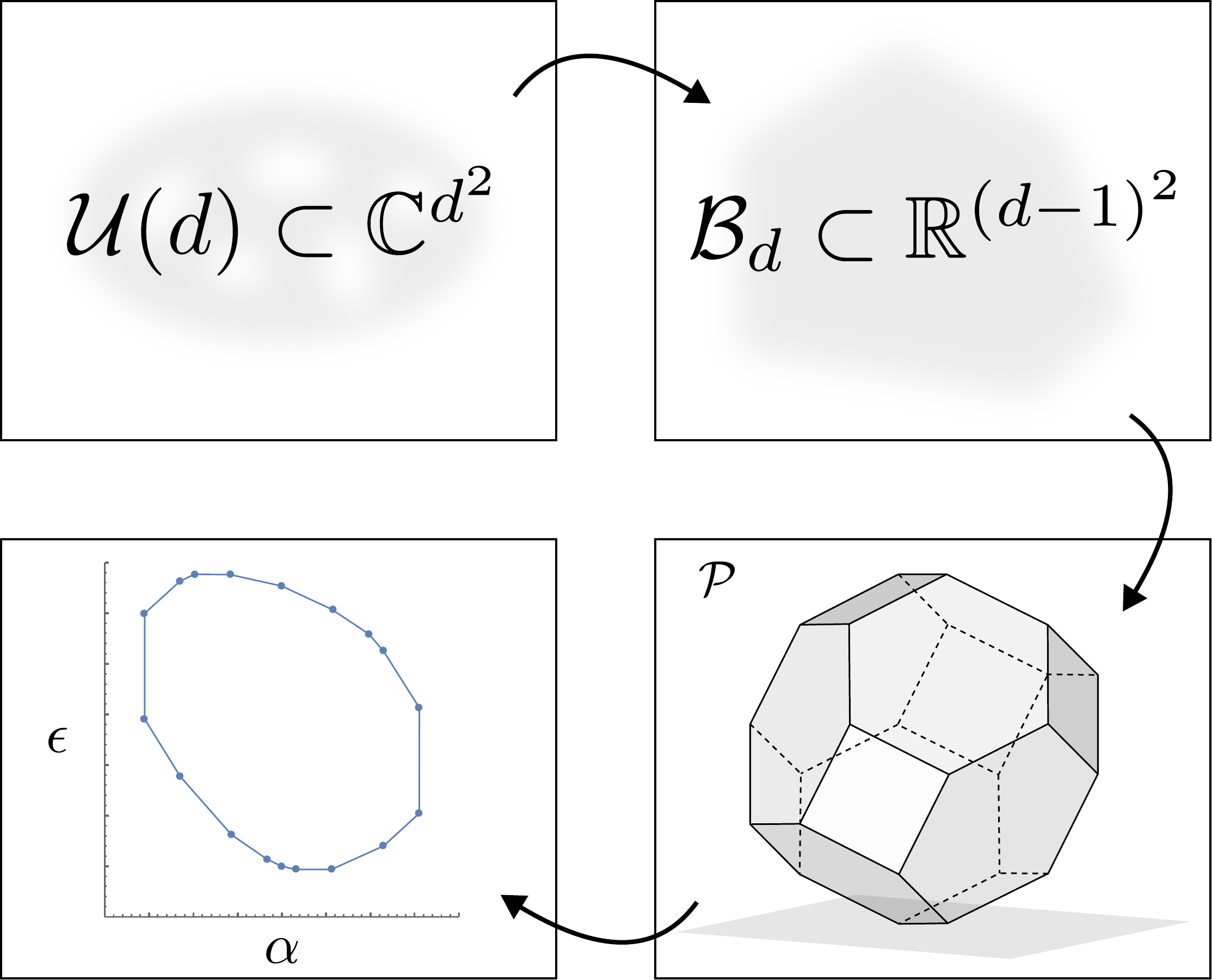}
\caption{A display of the optimisation spaces of the four reformulations of our problem. \textbf{Top left:} The space of all $d$-dimensional unitaries $\mathcal{U}(d)$. This relates to the original version of our problem of Eq.~\ref{equ:main_question_pop}. Note that $\mathcal{U}(d)$ is not simply connected which in intuitively means that $\mathcal{U}(d)$ has (topological) holes and hints at the difficulty of optimising over it. \textbf{Top right:} Thanks to the Schur-Horn and Birkhoff-von Neumann theorems we can reduce our problem to the Birkhoff polytope of doubly-stochastic transformations of dimension $d$. This is a $(d-1)^2$-dimensional subset of the vector space of $d \times d$ real matrices. \textbf{Bottom right:} We can further reduce the problem by considering how the eigenvalues of $\rho$ transform under the Birkhoff polytope, giving us the population polytope $\mathcal{P}$. We here display the population polytope of the four-dimensional vector $\textbf{p}=\frac{1}{10} (1,2,3,4)$, it has dimension 3. For a general $d$, $\mathcal{P}$ has dimension $d-1$ unless $p_0 = \dots = p_{d-1}$. \textbf{Bottom left:} The induced polytope of target vs cost, depicted here for our working example, see also Fig.~\ref{fig:polytopeworkingexample}. Each simplification is a many to one mapping (black arrows). We construct our solution primarily in the population polytope, demonstrating the optimality in the induced polytope and constructing corresponding trajectories in $\mathcal{B}_d$ and $\mathcal{U}(d)$.
}
\label{fig:differentlevels}
\end{figure}

\section{Main result}\label{sec:result}

In this section we state our main result, namely the construction of the optimal trajectory, the proof of which is detailed in the section~\ref{sec:proofmain}. Section~\ref{sec:GeneralisedScenario} will be dedicated to looking at a generalisation of our scenario for when an additional commuting conserved quantity enters the picture. 

\textbf{Ordering the basis states w.r.t. target and cost}. The result is simplest described when the preferred basis states $\{\ket{i}\}_i$ are ordered so that
\begin{enumerate}
    \item the coefficients of the target function $\mathcal{A}[\sigma] = \sum_i a_i\, p_i$ are in increasing order
    \begin{align} \label{eq:coefford}
        a_0 \leq a_1 \leq ... \leq a_{d-1},
    \end{align}
    \item for every set of degenerate (equal) target coefficients, i.e. $a_i = ... = a_{i+m}$, the coefficients of the cost function are in increasing order,
    \begin{align}
        E_i \leq E_{i+1} \leq ... \leq E_{i+m}.
    \end{align}
\end{enumerate}
In other words, the basis is ordered primarily w.r.t. increasing values of the target, and secondarily w.r.t. increasing values of the cost.

\subsection{The optimal trajectory}

\textbf{Minimal point of the optimal trajectory.} A final state $\sigma$ that is at the minimal point of the trajectory --- i.e. $\mathcal{A}[\sigma] = \alpha_{min}, \mathcal{E}[\sigma] = \omega_{\text{opt}}(\alpha_{min})$ --- is diagonal w.r.t. the preferred basis, with the populations arranged in decreasing order w.r.t. the ordering chosen above. As the eigenvalues $\{\lambda_n\}_n$ of the state and the diagonal elements are the same for diagonal states, this choice ensures that the largest eigenvalue contributes to the smallest target value $a_0$, the second largest to the smallest and so on. Furthermore, the secondary ordering w.r.t. the cost function ensures that between states of equal target value, the largest eigenvalue is associated with the state of minimal cost, and so on.

The state $\sigma$ at the minimal point is unique if and only if one of the following statements holds.
\begin{enumerate}
    \item The $a_i$'s are all non-degenerate, i.e. $a_i \neq a_j$.
    \item  Point 1. fails and among degenerate blocks of $a_i$'s the $E_i$'s are non-degenerate, i.e.
    \begin{align}
        a_i = a_j \quad\Longrightarrow\quad E_i \neq E_j.
    \end{align}
    
    \item Point 1. and point 2. fail and the arrangement of eigenvalues as described above leads to degenerate $\lambda_i$'s among blocks degenerate w.r.t. both $a_i$'s and $E_i$'s, i.e.
    \begin{align}
        \left( a_i = a_j \quad\text{\small AND}\quad E_i = E_j \right) \quad\Longrightarrow\quad \lambda_i = \lambda_j.
    \end{align}
\end{enumerate}

\textbf{Single step in the optimal trajectory.} Starting from the minimal point, the optimal trajectory is composed of a sequence of two-level swaps, each of which is selected as follows.
\begin{enumerate}
    \item List all of the two-level swaps of \textit{adjacently-valued} populations that would increase the target value, i.e. of a pair of states $\{\ket{k},\ket{l}\}$ with populations $p_k < p_l$ s.t. the following two conditions are satisfied:
    \begin{itemize}
        \item there is no state $\ket{j}$ with $p_k < p_j < p_l$,
        \item $a_k > a_l$.
    \end{itemize}
    \item Among these pick the pair (or any among equivalent pairs) for which the following quantity is minimal:
    \begin{align}\label{equ:mingradient}
       \frac{E_k - E_{l}}{a_k - a_l}.
    \end{align}
    \item Swap the selected pair, then repeat.
\end{enumerate}

A few things are worth noting upon. Firstly, \eqref{equ:mingradient} represents the gradient of cost vs target for the two-level swap in question, picking the minimum actually corresponds to minimising the gradient of the cost versus target among all the attainable points that increase $\alpha$. In other words the minimal value of Eq.~\ref{equ:mingradient} is the gradient of the minimal cost function $\omega_{\text{opt}}(\alpha)$.

However, a significant simplification of our proof is that not all unitaries, not even all two-level swaps, need be considered, only those between adjacently-valued populations; this follows from what we prove in Sec. \ref{sec:proofmain} --- that such swaps precisely generate the edges of the polytope of reachable population vectors.

\textbf{Maximal point of the optimal trajectory.} The above algorithm ends when there is no swap that increases the value of the target $\alpha$. This corresponds to when $\alpha = \alpha_{\max}$. One can characterize the states that achieve the maximal point of an optimal trajectory in an analogously as those that achieve the minimal point of the optimal trajectory. For the maximal point the eigenvalues are now ordered primarily increasingly w.r.t. increasing target values. Within subspaces with equal target value they are, as for the minimal point, ordered decreasingly w.r.t. the cost.

\textbf{Form of the optimal trajectory and minimal cost function}. As we just saw, for given target and cost observables $\mathcal{O}_A$, resp. $\mathcal{O}_E$, as well as a fixed initial state $\rho$, the optimal trajectory can be constructed via the above selected fixed sequence of 2-level swaps. By appropriate parametrisation of said 2-level swaps one can therefore continuously change the target value $\alpha$ from its initial value $\alpha_{\text{in}}=\mathcal{A}[\rho] \in [\alpha_{\min}, \alpha_{\max}]$ to its desired final value $\alpha_{\text{fin}}$. In doing so one needs only implement one 2-level swap partially at a time, implementing the next 2-level swap only once the current 2-level swap is fully implemented. While implementing a given 2-level swap, the gradient of the minimal cost function $\omega_{\text{opt}}$ is a constant \eqref{equ:mingradient}. The entire function is thus continuous and piecewise linear, as well as convex as we prove later.

Furthermore, at the level of doubly stochastic matrices, the above partial 2-level swap corresponds to a T-transform while at the level of unitaries it translates to a partial rotation in the corresponding two-dimensional subspace.

\textbf{If the initial state is non-optimal.} Depending on the initial state of the system it is possible that the initial point $\{\alpha_0,\omega_0\}$ is not part of the optimal trajectory. This just means that there are states $\rho'$ with the same target value but a lower cost value. In this case a natural choice of protocol from the initial state is to first apply the unitary that takes us to the point on the optimal trajectory that has the initial $\alpha$ value, and then move along the optimal trajectory as desired. There is a freedom in choosing the initial dynamic that depends upon the properties one would like the trajectory from the initial state to the optimal trajectory to have.

\subsection{Working example }
\subsubsection{Ground state cooling}
A paramount example of the above scenario is found in the thermodynamic task of cooling a quantum system $S$ via a joint unitary operation with a second --- usual thermal --- system $M$ that we refer to as machine for obvious reasons. $R$ is here the non-interacting composition of the system $S$ and machine $M$ i.e. the Hilbert space $\mathcal{H}_R = \mathcal{H}_S \otimes \mathcal{H}_M$ and its Hamiltonian is $H_{SM} = H_S \otimes \mathds{1}_M + \mathds{1}_S \otimes H_M$. The initial state of the composite system a product state $\rho_S \otimes \tau_M$, where $\tau_M$ is a thermal state w.r.t. the environment temperature $\beta$.\footnote{A `free' temperature, i.e. one that any system can thermalise to at no cost. We employ the inverse temperature $\beta = (kT)^{-1}$.} A typical target function is the ground state population of $S$, corresponding to the target observable being the projector $\Pi^{(0)}_S$ upon the ground state subspace of $S$, i.e. $\mathcal{O}_A = \Pi^{(0)}_S \otimes \mathds{1}_M$, and the cost function is the average energy of the joint state, corresponding to the observable being the total energy, i.e., $ \mathcal{O}_E=H_{SM}$.\footnote{The cost is in fact the \textit{change} in average energy, but this is independent of the unitary transformation, and can thus be subtracted as a constant at the end.} The target and cost observables commute, and the preferred eigenbasis can be taken to be the product energy eigenbasis. Here the target is simply the final value $\alpha \in [0,1]$ of the population of ground state subspace of $S$ after the unitary transformation. If the ground state is non-degenerate, then this is also referred to as \textit{erasure}, the concentration of population into a single pure state.

We will use this task to illustrate our results and methods throughout the paper. For concreteness, we consider the case where $S$ is a non-degenerate qubit initially in the maximally mixed state, and $M$ is a four-dimensional machine with distinct energy eigenvalues. The initial state of what we will refer to as our working example can thus be expressed as
\begin{align}
    \rho_S \otimes \tau_M &= \frac{\mathds{1}_S}{2} \otimes \sum_{m=0}^3 \tau_m \ket{m}\!\bra{m}_M, \\
    &= \sum_{i=0}^1 \sum_{m=0}^3 q_m \ket{im}_{SM}\!\bra{im} \label{eq:workingexampleinitial}
\end{align}
where the set of initial populations of the machine are $\{\tau_0,\tau_1,\tau_2,\tau_3\}$, arranged in strict decreasing order, i.e., $\tau_i > \tau_{i+1} $. To simplify the notation used later in the paper, we denote $\tau_m/2$ as $q_m$.

The target value of the initial state is $\alpha_{\text{in}} = 1/2$, our goal will be to find $\alpha_{min},\,\alpha_{max}$ and construct the optimal trajectory for $\alpha \in \left[ \alpha_{min}, \alpha_{max} \right]$.

\subsubsection{Optimal trajectory of the working example}

We apply our result to the working example of ground state cooling of a qubit, with initial state as in \eqref{eq:workingexampleinitial}. We introduce here a tabular visualisation of the state (Table \ref{tab:exampleinitial}) that is useful for thermodynamic tasks with joint states: we arrange the initial populations in a table, each row/column corresponding to an eigenstate of the system/machine.
\begin{figure}[t]
    \centering
    \includegraphics[width=0.5\linewidth]{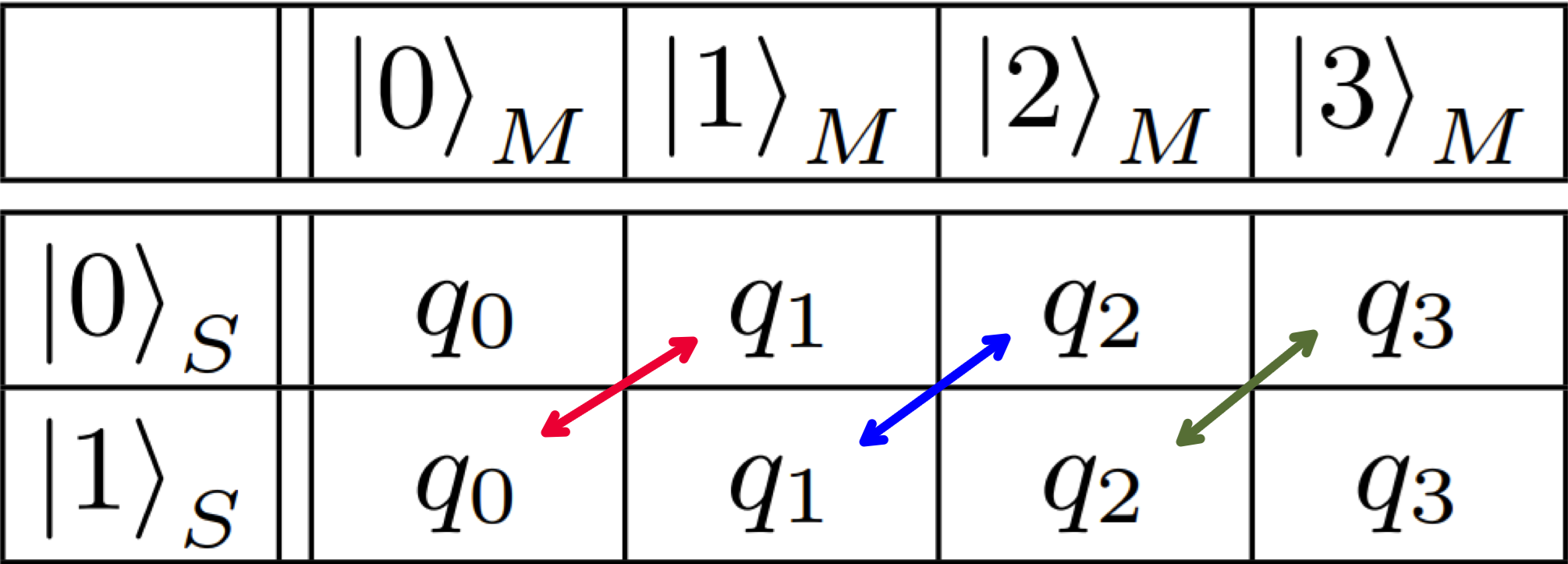}
    \caption{Visualisation of the initial state of system and machine as in \eqref{eq:workingexampleinitial}, with the swaps of adjacently-valued populations that increase the ground state population marked by colored arrows. The initial ground state population of the system is $2\,(\sum_i q_i) =1 $, with $q_i > q_{i+1} $. Also $q_i=\frac{\tau_i}{2}$ for $i=0,\dots,3$, where $\tau_i$ is the $i^{\text{th}}$ population of the machine.}
\label{tab:exampleinitial}
\end{figure}

The initial state does lie on the optimal trajectory, as will be seen in section \ref{sec:proofmain}, so rather than construct the entire trajectory we do so from the initial state to the maximal point. The part from the minimum to the initial state is symmetric.

If we label the energy eigenstates of the joint Hamiltonian by $\ket{i}_S \otimes \ket{j}_M$, then the coefficients of the cost and target functions are
\begin{align}
    a_{i,j} &= \begin{cases}
        1   & \text{if $i=0$}, \\
        0   & \text{if $i=1$},
    \end{cases} \\
    E_{i,j} &= E_i^{(S)} + E_j^{(M)}.
\end{align}

As stated there are two properties that an optimal swap must possess: that it increases the value of the target, and that it involves adjacent-valued populations. For our example this implies that 1) the swap must be between the ground and excited subspaces of the system, i.e. between the first and second row of the table, and 2) the swap must be between $q_i$ and $q_{i+1}$, where $i\in \{0,1,2\}$, with the former being in the excited (second row) and latter in the ground (first row) subspace.

Upon examining our initial state, we identify three swaps that satisfy the above conditions. For clarity, we denote the system's energy gap as $E_S$ and the machine's energy gaps as $E_{mn} = E_n - E_m$. The swaps are as follows:
\begin{enumerate}
    \item $\ket{01} \leftrightarrow \ket{10}$, with $\Delta \alpha =  q_0 - q_1$ and cost gradient $\frac{\Delta \epsilon}{\Delta \alpha}= E_{01}-E_S$,
    \item $\ket{02} \leftrightarrow \ket{11}$, with $\Delta \alpha = q_1 - q_2$ and cost gradient $\frac{\Delta \epsilon}{\Delta \alpha}=E_{12}-E_S$,
    \item $\ket{03} \leftrightarrow \ket{12}$, with $\Delta \alpha = q_2 - q_3$ and cost gradient $\frac{\Delta \epsilon}{\Delta \alpha}=E_{23}-E_S$.
\end{enumerate}

Here, $\Delta \alpha = \mathcal{A}(\sigma_{\text{after}}) - \mathcal{A} (\sigma_{\text{before}})$ and $\Delta \epsilon = \mathcal{E}(\sigma_{\text{after}}) - \mathcal{E} (\sigma_{\text{before}})$, where $\sigma_{\text{before}}$ and $\sigma_{\text{after}}$ refer to the states before and after the swap, respectively. To connect this with Eq.~\ref{equ:mingradient}, observe that if a swap were to occur between levels $\ket{k}$ and $\ket{l}$ — where we use the single index notation $\ket{n}$ for $\ket{i_n j_n}$, with $n=4 i_n + j_n$ — then
\begin{equation}
    \frac{\Delta \epsilon}{\Delta \alpha}= \frac{E_k-E_l}{a_k-a_l}.
\end{equation}
Thus, we seek the swap that minimizes $\frac{\Delta \epsilon}{\Delta \alpha}$. Since $E_S$ appears in each gradient, the optimal swap is determined solely by the smallest machine energy gap among $\{E_{01},E_{12},E_{23}\}$.

The smallest energy gap among those listed will depend on the choice of machine $M$. Once we perform the initial optimal swap—whichever one that may be—we can repeat the process. In this example, the next step will involve selecting from the two swaps not chosen initially. After performing the optimal swap among these two, the third step will be the remaining swap of the three, concluding with a final swap, $\ket{03} \leftrightarrow \ket{10}$, to achieve maximal cooling.

The entire graph of possible trajectories and associated conditions on the energy gaps of the machine is depicted in Fig. \ref{fig:erasuretrajectory}.

\onecolumngrid
\newpage

\begin{figure}[h]
    \centering
    \includegraphics[width=\linewidth]{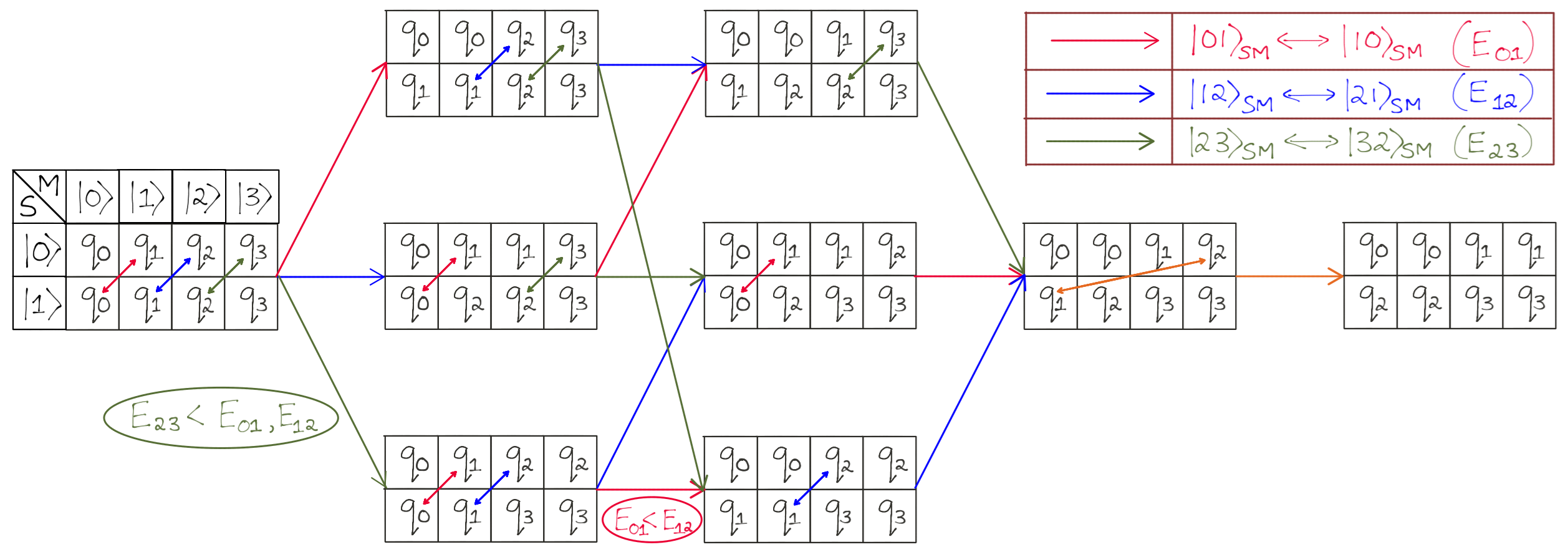}
    \caption{Possible optimal trajectories for the erasure of a maximally mixed qubit using a 4-dimensional machine with energy gaps and populations defined by $E_{mn} = E_n - E_m$. Each table represents a joint state of the system and machine, formatted as in Figure~\ref{tab:exampleinitial} but with simplified notation (row and column labels omitted). Swaps of adjacent population values that enhance the target state are indicated by colored arrows, with each color signifying a specific energy gap: red, blue, and green correspond to $E_{01}, E_{12}$, and $E_{23}$ as the smallest energy gap among candidates, respectively. For example, if $E_{23}<E_{01}<E_{12}$, the optimal path follows the bottom trajectory, starting with the green arrow for $E_{23}$, the smallest in $\{E_{23},E_{12},E_{01}\}$. The next step selects the red arrow for $E_{01}$, the smallest in $\{E_{01},E_{12}\}$, leaving the blue arrow, corresponding to $E_{12}$, as the final selection. To achieve maximal cooling, the last swap, indicated by an orange arrow, is always $\ket{03} \leftrightarrow \ket{10}$, regardless of the energy gap ordering.}
    \label{fig:erasuretrajectory}
\end{figure}

\twocolumngrid

\section{Proof of main result} \label{sec:proofmain}

In this section, we prove our main result: the construction of the optimal trajectory. We structure the proof into subsections that highlight the key conceptual and technical steps. In Sec.~\ref{proof:populationpolytope}, we demonstrate that the set of population vectors ${\bf p}$ attainable through unitary operations forms a polytope $\mathcal{P}$, with vertices corresponding to permutations of the eigenvalue vector of the system's initial state. Next, in Sec.~\ref{proof:edgesgeneric} and \ref{proof:edgesadjacent}, we show that the edges of this polytope consist of two-level swaps between adjacent eigenvalues. In Sec.~\ref{proof:inducedpolytope}, we focus on the values of the target $\mathcal{A}$ and cost $\mathcal{E}$ of the unitarily transformed state, denoted by $(\alpha, \epsilon) \in \mathds{R}^2$. We find that the set of attainable points $(\alpha, \epsilon)$ forms a two-dimensional polytope $\mathcal{Q}_2$, which can be viewed as a projection of $\mathcal{P}$ into two dimensions. The optimal trajectory is then identified with the lower boundary of $\mathcal{Q}_2$.\footnote{For each value of $\alpha$, there are two corresponding values of $\epsilon$ on the boundary of $\mathcal{Q}_2$, with points on the optimal trajectory delivering the smaller $\epsilon$ value.} In Sec.~\ref{proof:extrematrajectory}, we show that the minimal and maximal points of the optimal trajectory can be selected as vertices of $\mathcal{P}$.

In Sec.~\ref{proof:gradientedge}, we reach the core of the proof. We demonstrate that one can progress along the optimal trajectory from vertex to vertex via an edge by selecting, among adjacent-valued swaps, the one that minimizes the gradient of cost versus target. Finally, in Sec.~\ref{proof:conclusion}, we synthesize all previous results to construct the complete optimal trajectory within the population polytope, lifting it to the space of unitary operations and density matrices in Sec.~\ref{proof:liftedtrajectory}.

\subsection{The population polytope}\label{proof:populationpolytope}

We start by examining the set of population vectors that can be achieved through unitary transformations. Although the population vector $\bf p$ depends non-linearly on the unitary transformation—the variable we aim to minimize over—we find that the achievable set of populations forms a well-defined polytope, which we refer to as the \textit{population polytope} $\mathcal{P}$. By working within this polytope, we can avoid the topological complexity of the full unitary space $\mathcal{U}(d)$, focusing instead on the more structured, geometrically accessible space of attainable population vectors. This simplification is possible because both the target and cost functions are diagonal in a shared basis, or equivalently, because the target and cost observables commute

To connect our problem to polytope theory~\cite{Brondsted-1983, Ziegler-1995, Gruenbaum-2003} we utilize the Schur and Horn theorems from majorisation theory\cite{Marshall-2011}.
Consider the vector of eigenvalues $ \boldsymbol{\lambda} = \{\lambda_k\}_k$ of the initial system state $\rho$. Since unitary operation preserves eigenvalues, $\boldsymbol{\lambda}$ remains the vector of eigenvalues for the final state $U \rho U^{\dagger}$. According to Schur's theorem~\cite[Chapter~9.B]{Marshall-2011}, the vector of eigenvalues of any density matrix always majorises its vector of diagonal elements. This applies to our unitarily transformed state, such that
\begin{align}
    \boldsymbol{\lambda} \succ {\bf p},
\end{align}
where ${\bf p}$ denotes the vector of populations of $U \rho U^{\dagger}$. 

Moreover, Horn's theorem~\cite[Chapter~9.B]{Marshall-2011} states that for a given vector of eigenvalues $\boldsymbol{\lambda}$, any vector ${\bf v}$ that is majorized by $\boldsymbol{\lambda}$ can be achieved by some unitary transformation. Therefore, the set of unitarily attainable populations ${\bf p}$ is precisely defined as
\begin{equation}
    \{ {\bf p} \mid {\bf p} \prec \boldsymbol{ \lambda}\}.
\end{equation}
Now, ${\bf p} \prec \boldsymbol{\lambda}$ is equivalent to the statement that ${\bf p} = D \boldsymbol{\lambda}$, for some doubly stochastic matrix $D$. Therefore, the set of all attainable vectors $\bf{p}$ is exactly the set of all doubly stochastic transformations $D$ of $\boldsymbol{\lambda}$.

According to the Birkhoff-von Neumann theorem~\cite[Chapter~2.A]{Marshall-2011}, the set of doubly stochastic matrices is the convex hull of permutations. In particular, given the set of all $d$-dimensional permutations $\{P_n\}_n$, every doubly stochastic transformation $D$ can be expressed as
\begin{align} \label{eq:doublystochastic}
    D &= \sum_{n=0}^{d!-1} r_n P_n,
\end{align}
where ${\bf r}=(r_0,\dots,r_{d!-1})$ is a vector of probabilities ($r_n \geq 0$, $\sum_n r_n = 1$). Consequently, the set of attainable population vectors ${\bf p}$ can also be represented as the convex hull of the permutations of $\boldsymbol{\lambda}$:
\begin{align} \label{eq:pconvexhull}
    {\bf p} = D \boldsymbol{\lambda} = \sum_{n=0}^{d!-1} 
    r_n \; P_n \boldsymbol{\lambda}.
\end{align}

Thus, the set of all possible final population vectors forms a (bounded, finite and convex) polytope\footnote{We consider a polytope to be the convex hull of finitely many points. This coincides with what is sometimes denoted by bounded, finite and convex polytopes in the literature.} $\mathcal{P}$. Furthermore, the vertices of this polytope are precisely the permutations of the eigenvalue vector $\boldsymbol{\lambda}$. For further details, see Lemma~\ref{lem:permutationisvertex} in Appendix~\ref{app:grouptheorypolytope}. We refer to $\mathcal{P}$ as the population polytope of $\boldsymbol{\lambda}$. As an example, Fig. \ref{fig:population_polytope}  illustrates the population polytope for the four-dimensional eigenvalue vector $\boldsymbol{\lambda}=\frac{1}{10} (1,2,3,4)$.
\begin{figure}[h]
\centering
\includegraphics[width=0.9\columnwidth]{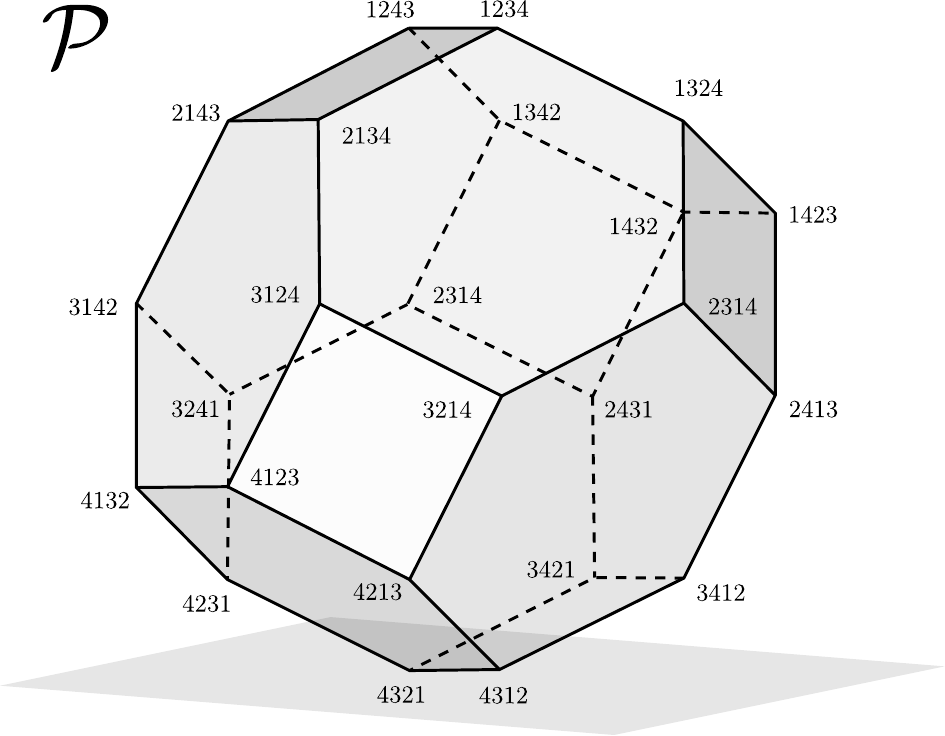}
\caption{The three dimensional population polytope for $\boldsymbol{\lambda}=\frac{1}{10} (1,2,3,4)$. Each vertex corresponds to a permutation of the elements of $\boldsymbol{\lambda}$.
}
\label{fig:population_polytope}
\end{figure}

\subsection{The edges of a generic polytope}\label{proof:edgesgeneric}

In this section we discuss a few properties of a polytope that we require for our result, these apply to polytopes in general and not only to the population polytope. To this end we begin by denoting the polytope as $\mathcal{Q}$ and we let $V$ be a vertex of $\mathcal{Q}$.

\begin{definition}[Vertex vector at $V$]
A {\bf vertex vector at $V$} is a vector of the form $V_i-V$, where $V_i$ and $V$ are distinct vertices of $\mathcal{Q}$.
\end{definition}
Thus, a vertex vector at $V$ represents a vector from $V$ to another vertex of $\mathcal{Q}$. An {\bf edge} (of $\mathcal{Q}$) is a one dimensional face of $\mathcal{Q}$ - see e.g.~\cite[Chapter~2.1]{Ziegler-1995}. An { \bf edge at $V$} is an edge that contains $V$. We say that a vertex vector (at $V$) $V_i-V$  generates an edge $\mathcal{D}$ (at $V$) if, for any ${\bf p} \in \mathcal{D}$,
\begin{equation}\label{equ:VertexVectorGeneratingEdge}
    {\bf p} = r (V_i-V)+V
\end{equation}
for some $r \in [0,1]$. Every edge at $V$ is generated by a (unique) vertex vector, specifically $V_i - V$, where $V_i$ and $V$ are the endpoints of the edge. For more details, see Lemma~\ref{lemma:generatingedge} in Appendix~\ref{subsection:generatingedge}.
We refer to this vertex vector as the \textbf{vertex vector of the edge}. Note, however, that not every vertex vector generates an edge. Instead, the vertex vectors of edges at $V$ generate the remaining vertex vectors (at $V$) in the following sense (Proof in Appendix \ref{app:edgesminimalset}).
\begin{lemma}\label{lemma:edgesminimal}
For any point ${\bf p} \in \mathcal{Q}$ and vertex $V$, the vector ${\bf p} - V$ lies within $\mathcal{K}_V$, the cone generated by the vertex vectors of edges at $V$. Specifically,
\begin{align}
    {\bf p} - V &= \sum_i r_i \; \left(Z_i - V \right),
\end{align}
where $r_i \geq 0$ and $\{Z_i\}_i$ is the set of vertices of $\mathcal{Q}$ such that $\{Z_i-V\}_i$ are the vertex vectors of edges at $V$. 
\end{lemma}

\begin{figure}[h]
\centering
\includegraphics[width=0.8\columnwidth]{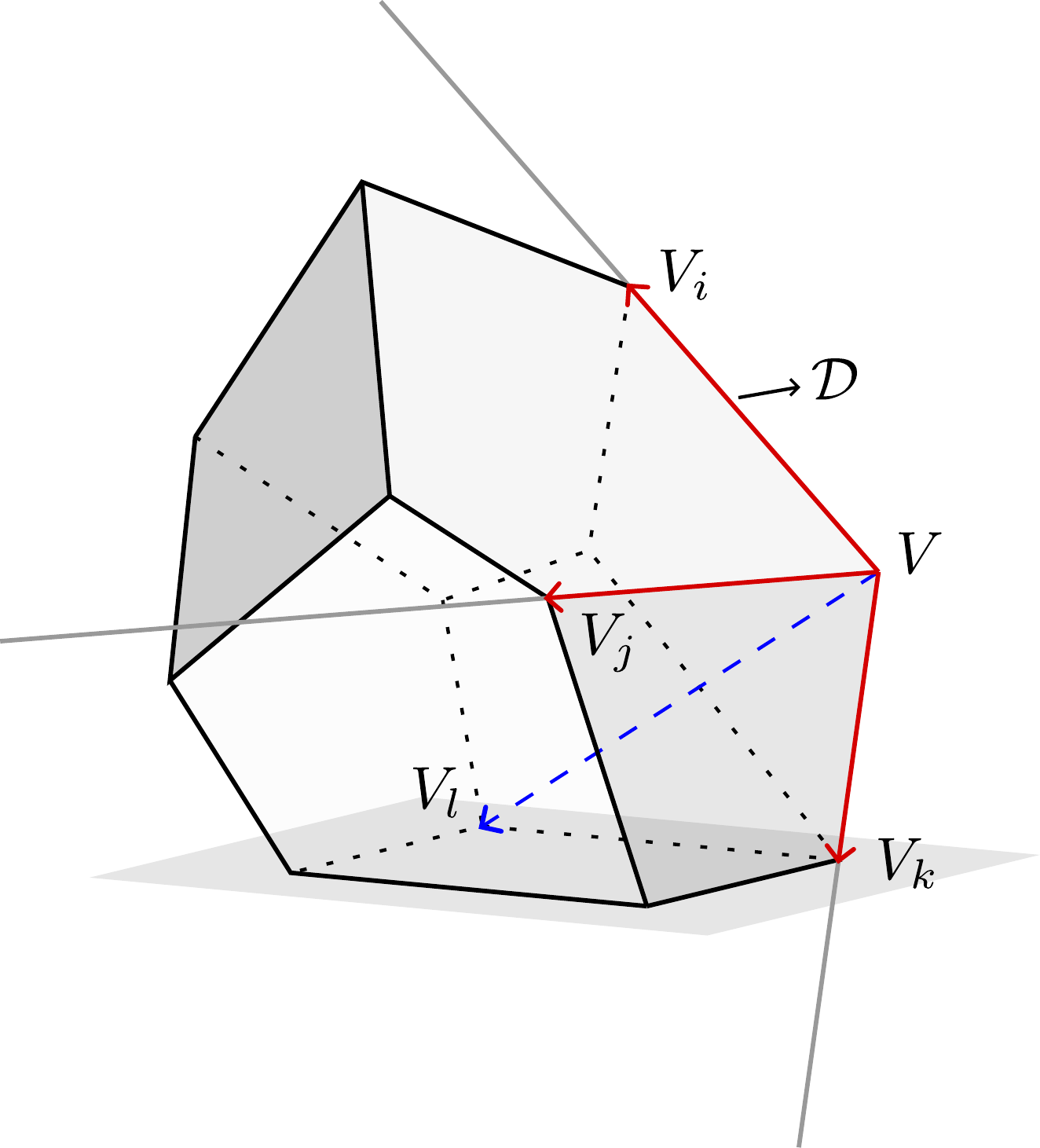}
\caption{A sketch of an arbitrary polytope $\mathcal{Q}$ with a highlighted vertex $V$. The vector $V_i-V$ is a vertex vector at $V$ and generates a specific edge at $V$, defined as $\mathcal{D}=\{ {\bf p} \in \mathcal{Q} \mid {\bf p} = r( V_i -V) + V, r\in [0,1]\}$. In contrast, $V_l - V$ (displayed in dashed blue) is a vertex vector at $V$ that does not generate an edge at $V$. There are 3 vertex vectors at $V$ (displayed in red) that generate edges: $V_i-V$, $V_j-V$ and $V_k-V$. By Lemma~\ref{lemma:edgesminimal}, any point ${\bf p} \in \mathcal{Q}$ satifies ${\bf p} - V \in \mathcal{K}_V$, which is the cone generated by $V_i-V$, $V_j-V$ and $V_k-V$. This can be visualized by extending the edges at $V$ (depicted in light gray) and observing that the whole polytope lies within the cone (centered at $V$) obtained from these extended edges.}\label{fig:polytopedrawing}
\end{figure}

Each vertex vector of an edge at $V$ lies on what is known as an extreme ray of $\mathcal{K}_V$. This means, in particular, that all the vertex vector of edges at $V$ are required to generate $\mathcal{K}_V$. More formally, we have the following result (proof in Appendix~\ref{app:extremerays}).
\begin{lemma} \label{lemma:extremerays}
A vertex vector of an edge at $V$ cannot be expressed as a positive linear combination of other vertex vectors at $V$.
\end{lemma}

An illustration of the vertex vectors and edges of a generic polytope is given in Fig. \ref{fig:polytopedrawing}.

\subsection{The edges of the population polytope}\label{proof:edgesadjacent}

Let us return to the specific case of a population polytope. Here, each vertex corresponds to a permutation of elements, meaning that each vertex vector is induced by a particular permutation. Specifically, we say that a permutation {\bf $P$ generates the edge $\mathcal{D}$ }(at $V$) if the vertex vector $P(V) -V$ generates $\mathcal{D}$. As we will show, the permutations that generate edges are, in fact, highly specific: they are the following type of 2-cycle\footnote{a k-cycle is a cyclic permutation of $k$ elements.}:
\begin{definition}[av-swap]
For a vector ${\bf p}=(p_0,\dots,p_{d-1}) \in \mathbb{R}^d$, we call two of its components $p_i$ and $p_j$ \textit{ adjacent-valued} if either $p_i < p_j$ or $p_j < p_i$ holds and
\begin{align}
    \nexists \;  p_k \quad s.t. \begin{cases}
                                p_i < p_k < p_j & \text{if $p_i < p_j$}, \\
                                p_j < p_k < p_i & \text{if $p_j < p_i$},
                            \end{cases}
\end{align}
that is, there is no other component of ${\bf p}$ between $p_i$ and $p_j$. We define a 2-cycle that swaps adjacent-valued elements an adjacent-valued swap, or \textit{av-swap} for short.
\end{definition}

Note that there are at most $d-1$ possible av-swaps for a vector in $\mathbb{R}^d$,  as opposed to $d(d-1)/2$ possible 2-cycles. We are now ready to state the main result of this section.

\begin{theorem}\label{thm:edgesadjacentvalued}
The edges of  $\mathcal{P}$ are generated by av-swaps.
\end{theorem}

The complete proof can be found in Appendix \ref{app:barvinokextension}. It is a modified version of a proof by Barvinok (\cite{Barvinok-2002}, Chapter VI, Proposition 2.2). In Barvinok’s original proof, the structure of all faces of the population polytope is derived under the assumption that the population vector has all distinct elements. Here, we extend this proof to account for possible degeneracies in the population vector.

\subsection{The induced polytope of cost vs target \& the optimal cost function}\label{proof:inducedpolytope}

Next, we turn our attention to the attainable values of the target and cost functions. By varying over all unitarily attainable populations ${\bf p}$ we can collect the points $(\alpha,\epsilon) \in \mathbb{R}^2$, where $\alpha =\mathcal{A}({\bf p})$ and $\epsilon = \mathcal{E}({\bf p})$. As both the target $\mathcal{A}$ and cost $\mathcal{E}$ are linear functions of the population vector and the latter forms a polytope, it follows that the set of possible values of $(\alpha,\epsilon )$ also forms a polytope. The vertices of this polytope are a subset of the values of $(\alpha, \epsilon)$ obtained from vertices of the population polytope. Indeed, consider a point $(\alpha, \epsilon)$ as above. Then $(\alpha, \epsilon) = \left(\mathcal{A}({\bf p}), \mathcal{E}({\bf p}) \right)$ for some point ${\bf p}$ in the population polytope and as ${\bf p}$ may be expressed as $ {\bf p} = \sum_n q_n V_n$, where $V_n= P_n \lambda$ are the vertices of $\mathcal{P}$ and $(q_n)_n$ a probability distribution, by linearity of $\mathcal{A}$ and $\mathcal{E}$,
\begin{equation} \label{equ:alphapolytope}
\begin{aligned} 
   (\alpha, \epsilon)&=\left( \mathcal{A}({\bf p}), \mathcal{E}({\bf p}) \right)\\
   &= \left( \sum_n q_n \mathcal{A}(V_n), \sum_n q_n \mathcal{E}(V_n) \right)\\
   &= \sum_n q_n \left( \mathcal{A}(V_n), \mathcal{E}(V_n) \right).
\end{aligned}
\end{equation}
This shows that any point $(\alpha, \epsilon)$ lies in the convex hull of $\{(\mathcal{A}(V_n), \mathcal{E}(V_n)\}_n$, thus proving the above assertion.
We refer to the newly obtained polytope as the induced polytope $\mathcal{Q}_2$, which, as a subset of $\mathbb{R}^2$, is in fact a polygon. Fig. \ref{fig:polytopeworkingexample} illustrates the induced polytope for our working example of erasure.

\begin{figure}[h]
\centering
\includegraphics[width=0.9\columnwidth]{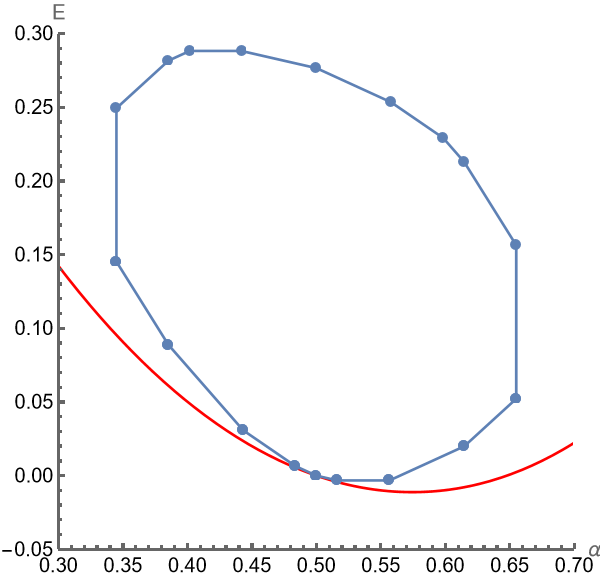}
\caption{
Induced polytope for our working example: a maximally mixed qubit coupled to a 4-dimensional machine with parameters: $\beta = 1$, $E_S = 0.3, E_1 = 0.1, E_2 = 0.4, E_3 = 1.1$. In red is the lower bound on work cost given by the free energy change of the system\cite{Reeb-2014,Taranto-2021}. 
}
\label{fig:polytopeworkingexample}
\end{figure}

By definition, the optimal cost function is the lower boundary of the induced polytope. A consequence of the convexity of the population polytope is that the optimal cost function $\omega_{\text{opt}}(\alpha)$ is also convex:
\begin{lemma}\label{lem:optconvex}
    $\omega_{\text{opt}}(\alpha)$ is convex.
\end{lemma}

\begin{proof}
(by contradiction) Assume that $\omega_{\text{opt}}(\alpha)$ is not convex. Then there exist $\alpha_1, \alpha_2 \in [\alpha_{\min}, \alpha_{\max}]$ as well as some $p \in [0,1]$ such that 
\begin{align}\label{eq:woptconcave?}
    \omega_{\text{opt}} \left( p \alpha_1 + (1-p) \alpha_2 \right) &> p \,\omega_{\text{opt}}(\alpha_1) + (1-p) \,\omega_{\text{opt}}(\alpha_2).
\end{align}

Corresponding to $\alpha_1$ and $\alpha_2$ there must exist population vectors ${\bf p}_1,{\bf p}_2 \in \mathcal{P}$ for which
\begin{align}
    \mathcal{A} ({\bf p}_1) &= \alpha_1, \quad \mathcal{E} ({\bf p}_1) = \omega_{\text{opt}}(\alpha_1), \\
    \mathcal{A} ({\bf p}_2) &= \alpha_2, \quad \mathcal{E} ({\bf p}_2) = \omega_{\text{opt}}(\alpha_2). 
\end{align}

Since the population polytope is convex, any convex combination of ${\bf p}_1$  and ${\bf p}_2$ is part of the polytope too. So given $p \in [0,1]$ as above, we have $ {\bf p}^\prime = p \, { \bf p}_1 + (1-p) \, { \bf p}_2 \in \mathcal{P}$. The target and cost values of this point are
\begin{align}
    \mathcal{A}( {\bf p}^{\prime} ) &:= {\bf a} \cdot {\bf p}'=p \, \alpha_1 + (1-p) \,\alpha_2 =: \alpha', \\
    \mathcal{E}({\bf p}^\prime) &:= {\bf E} \cdot {\bf p}'=p \, \omega_{\text{opt}}(\alpha_1) + (1-p) \, \omega_{\text{opt}}(\alpha_2).\label{eq:costlinepoint}
\end{align}
Comparing Eq.~\ref{eq:costlinepoint} with Eq.~\ref{eq:woptconcave?} we readily see that
\begin{equation} \label{equ:optimalityviolation}
    \omega_{\text{opt}}(\alpha') > \mathcal{E}({\bf p}').
\end{equation}
But Eq.~\ref{equ:optimalityviolation} means that ${\bf p}'$ -- which satisfies $\mathcal{A}({\bf p}') = \alpha'$ -- achieves a lower cost than $\omega_{\text{opt}}(\alpha')$, the optimal cost function at $\alpha'$ , which is impossible. Therefore, it must be that $ \omega_{\text{opt}}(\alpha)$ is convex.\footnote{Analogously: if one maximises rather than minimises the cost, the resulting function must be concave.}

\end{proof}

\subsection{Minimal point of the optimal trajectory}\label{proof:extrematrajectory}

To construct the optimal trajectory, we begin with finding its start, i.e. a population vector that minimises the target function. From Eq.~\ref{equ:alphapolytope} given any ${\bf p} \in \mathcal{P}$ we have that
\begin{equation}
    \mathcal{A}({\bf p}) = \sum_i q_i \mathcal{A}(V_i)
\end{equation}
for an appropriately chosen probability distribution $(q_n)_n$. Thus there exists a (non-empty) set $\mathcal{S}_{min}$ of vertices of $\mathcal{P}$ that attain the minimal possible value of the target, we label this value $\alpha_{min}$\footnote{The same holds for the maximal possible value of the target $\alpha_{\max}$}. By linearity of the target function, the value of the constraint on the entire polytope $\mathcal{P}_{min} \subset \mathcal{P}$ generated by $\mathcal{S}_{min}$ must be $\alpha_{min}$. Furthermore, any population $ {\bf p}$ for which $A ( {\bf p}) = \alpha_{\min}$ must lie in $\mathcal{P}_{\min}$. We therefore have 
\begin{align}
    &\min_U \mathcal{E}({\bf p}), \quad \text{s.t. } \mathcal{A}({ \bf p}) = \alpha_{\min}\\
    =& \min_{{\bf p} \in \mathcal{P}_{\min}} \mathcal{E}({\bf p}),
\end{align}
meaning that in order to answer our question of Eq.~\ref{equ:main_question_pop} for $\alpha = \alpha_{\min}$ we need only look for ${\bf p} \in \mathcal{P}_{\min}$. Now within $\mathcal{P}_{\min}$, the same reasoning implies that there exists at least one vertex that attains the minimal value of the cost $\mathcal{E}(\cdot)$ (given $\alpha_{min}$). We can pick any such vertex to be the start of the optimal trajectory.

\subsection{The gradient of cost vs target}\label{proof:gradientedge}

Next, consider an arbitrary vertex $V$ of the population polytope $\mathcal{P}$ that is on the optimal trajectory and has a target value different than $\alpha_{\max}$. We here prove that $\left. \frac{\Delta \epsilon}{\Delta \alpha} \right \rvert_V $ --- the gradient of the cost vs target at $V$ --- is bounded by the values it takes along edges at $V$, see Fig.~\ref{fig:gradients}. 
\begin{figure}[h]
\centering
\includegraphics[width=0.95\columnwidth]{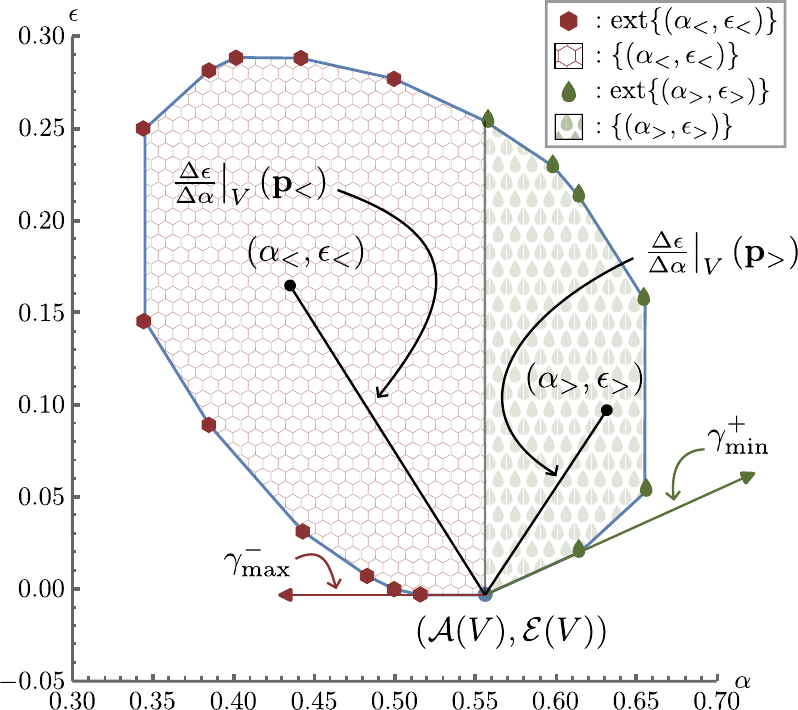}
\caption{We showcase here the behavior of $\left. \frac{\Delta \epsilon}{\Delta \alpha} \right \rvert_V $ at a specific vertex $(\mathcal{A}(V),\mathcal{E}(V))$ of the induced polytope of cost vs target $\mathcal{Q}_2$. The boundary of the induced polytope $\mathcal{Q}_2$ is displayed in blue. The set of points for which the target value is smaller than $\mathcal{A}(V)$ -- $\{(\alpha_<, \epsilon_<)\}$-- is displayed with a brown and hexagonal background and the vertices of the induced polytope belonging to that set -- $\text{ext}\{(\alpha_<, \epsilon_<)\}$-- are displayed as brown hexagonal points. The set of points for which the target value is bigger than $\mathcal{A}(V)$ -- $\{(\alpha_>, \epsilon_>)\}$-- is displayed with a green droplet background with the vertices of the induced polytope belonging to that set -- $\text{ext}\{(\alpha_>, \epsilon_>)\}$-- displayed as green droplet points.  Note that while any vertex of the induced polytope must be an image of a vertex of the population polytope $\mathcal{P}$ it may very well be that for a given vertex $V_i$ of $\mathcal{P}$ -- even for one such that $V_i-V$ is a vertex vector of an edge at $V$ -- $(\mathcal{A}(V_i), \mathcal{E}(V_i))$ is not a vertex of the induced polytope. The convexity of the optimal cost function (see Lemma~\ref{lem:optconvex}) 
 implies that $\gamma^+_{\min} \geq \gamma^-_{\max}$, middle statement of Theorem~\ref{thm:increasingslope}. One also readily sees that $\left.\frac{\Delta \epsilon}{\Delta \alpha} \right \rvert_V ({\bf p}_>) \geq \gamma^+_{\min}$ and $ \gamma^-_{\max} \geq \left.\frac{\Delta \epsilon}{\Delta \alpha} \right \rvert_V ({\bf p}_<)$ as prescribed by Theorem~\ref{thm:increasingslope}. Also note that the statements of Theorem~\ref{thm:increasingslope} hold whatever the signs of $\left.\frac{\Delta \epsilon}{\Delta \alpha} \right \rvert_V ({\bf p}_>),\gamma^+_{\min}, \gamma^-_{\max}$ and $ \left.\frac{\Delta \epsilon}{\Delta \alpha} \right \rvert_V ({\bf p}_<)$ may be. In particular, for our choice of values for the above figure we have $\gamma^-_{\max} >0$ and $\left.\frac{\Delta \epsilon}{\Delta \alpha} \right \rvert_V ({\bf p}_<)< 0$ from which we trivially get $\gamma^-_{\max}\geq \left.\frac{\Delta \epsilon}{\Delta \alpha} \right \rvert_V ({\bf p}_<)$.}
\label{fig:gradients}
\end{figure}
To illustrate this, consider a point ${\bf p}$ from the population polytope $\mathcal{P}$. We need to distinguish between two situations based on whether the target value of ${\bf p}$ is greater than or less than that of our vertex $V$. If $\mathcal{A}({\bf p}) > \mathcal{A}(V)$, we will label our point as ${\bf p}_>$ and refer to the corresponding values in $\mathcal{Q}_2$ as $(\mathcal{A}({\bf p}_>)$. Conversely, if $\mathcal{A}({\bf p}) < \mathcal{A}(V)$ we denote our point as ${\bf p}_<$ and write $(\alpha_<,\epsilon_<)$ for its associated point in the induced polytope $\mathcal{Q}_2$. Additionally, we denote the set of vertices of $\mathcal{P}$ that are different from $V$ and belong to an edge at $V$ as $\{Z_i\}_i$.

Now, if we label the gradient of the cost vs target functions at $V$ for a given point ${\bf p}$ as
\begin{equation}
    \left.\frac{\Delta \epsilon}{\Delta \alpha} \right \rvert_V ({\bf p}) = \frac{{\bf E} \cdot ({\bf p} - V)}{{\bf a} \cdot ({\bf p} - V)},
\end{equation}
and introduce the notation
\begin{align}
    \gamma^+_{\min} &= \min_{\substack{ Z_i \text{: } \mathcal{A}(Z_i) > \mathcal{A}(V) }} \left.\frac{\Delta \epsilon}{\Delta \alpha} \right \rvert_V (Z_i),\\
    \gamma^-_{\max} &= \min_{\substack{Z_i \text{: } \mathcal{A}(Z_i) < \mathcal{A}(V)}} \left.\frac{\Delta \epsilon}{\Delta \alpha} \right \rvert_V (Z_i),
\end{align}
what we see is the following.
\begin{theorem}\label{thm:increasingslope} $\left.\frac{\Delta \epsilon}{\Delta \alpha} \right \rvert_V ({\bf p}_>) \geq \gamma^+_{\min} \geq \gamma^-_{\max} \geq \left.\frac{\Delta \epsilon}{\Delta \alpha} \right \rvert_V ({\bf p}_<)$.
\end{theorem}
The statement of Theorem~\ref{thm:increasingslope} can be understood as three different inequalities. The inequality $\gamma_{\min}^+ \geq \gamma_{\max}^-$ means that the gradient of cost vs target at $V$ of any vertex of an edge at $V$ that increases the target value must be greater (or equal) than the gradient of cost vs target at $V$ of any vertex of an edge at $V$ that decrease the target value. This relationship is closely tied to the convexity of the induced polytope and that the optimal cost function lies on the boundary of that polytope, a proof of it can be found in Appendix~\ref{app:increasingslope}. The inequality $\left.\frac{\Delta \epsilon}{\Delta \alpha} \right \rvert_V ({\bf p}_>) \geq \gamma^+_{\min}$ implies that the gradient of any point increasing the target value has to at least be $\gamma_{\min}^+$. This hints at the fact that the most economic way to increasing the target value is to do so at a slope of $\gamma_{\min}^+$ -- or along an edge inducing a $\gamma_{\min}^+$ gradient to be more precise -- and will be crucial to the construction of the optimal trajectory in Sec.~\ref{proof:conclusion}. The proof of that inequality heavily relies on the fact that $\gamma_{\min}^+ \geq \gamma_{\max}^-$. It starts by expressing ${\bf p}-V$ as the positive combination of vertex vector of edges at $V$ as in Lemma~\ref{lemma:edgesminimal} and then uses $\gamma_{\min}^+ \geq \gamma_{\max}^-$ to lower bound $\Delta \epsilon$. The full proof can be found in Appendix~\ref{app:increasingslope}. Finally the inequality $\gamma^-_{\max} \geq \left.\frac{\Delta \epsilon}{\Delta \alpha} \right \rvert_V ({\bf p}_<)$ 
serves as an analog of the previous one, demonstrating that the most economic way to decrease the target value is to do so at a slope of $\gamma^-_{max}$. Its proof, also in Appendix~\ref{app:increasingslope}, is derived by modifying the proof for $\left.\frac{\Delta \epsilon}{\Delta \alpha} \right \rvert_V ({\bf p}_>) \geq \gamma^+_{\min}$ to account for the fact that now $\Delta \alpha <0$, which reverses the direction of the inequality.
\subsection{The optimal trajectory}\label{proof:conclusion}
We now have the two ingredients to build up the optimal trajectory. We demonstrate --- using induction --- that it can always be built as a sequence of vertices connected by edges formed through av-swaps.

To begin, we select the starting point of the trajectory to be a vertex, as established in Sec. \ref{proof:extrematrajectory}, completing the base case. For the induction step we place ourselves at a vertex $V$ of the optimal trajectory and consider a vertex $Z_i$ of $\mathcal{P}$ such that $Z_i - V$ is a vertex vector of an edge at $V$ with $i \in S_+$ and $\gamma_i = \gamma_{\min}^+$. For any point ${\bf p} \in \mathcal{P}$ satisfying $ \mathcal{A}(V) < \mathcal{A}({\bf p}) \leq \mathcal{A}(Z_i)$ we can select a point ${\bf p^\prime}$ on the edge at $V$ generated by $Z_i - V$ that shares the same target value:
\begin{equation} \label{equ:pprime}
    {\bf p}' = r' \, (Z_i - V) + V,
\end{equation}
with $0 \leq r'=\frac{\mathcal{A}({\bf p})-\mathcal{A}(V)}{\mathcal{A}(Z_i)-\mathcal{A}(V)} \leq 1$. We now demonstrate that $\mathcal{E}({\bf p^\prime}) \leq \mathcal{E}({\bf p})$ for all ${\bf p} \in \mathcal{P}$.

Given that $\mathcal{A}({\bf p}) > \mathcal{A}(V)$, Theorem~\ref{thm:increasingslope} from Sec.~\ref{proof:gradientedge} implies that the cost function of ${\bf p}$ can be bounded from below as follows.
\begin{align}
    \mathcal{E}({\bf p}) &= \mathcal{E}(V) + \frac{\mathcal{E}({\bf p}) - \mathcal{E}(V)}{A({\bf p}) - A(V)} \left( A({\bf p}) - A(V) \right) \\
    &\geq \mathcal{E}(V) + \gamma^+_{min} \, \Delta \alpha, \label{equ:dynamical_lower_bound}
\end{align}
where $\Delta \alpha := \mathcal{A}({\bf p}) - \mathcal{A}(V)$. For ${\bf p^\prime}$, using Eq.~\ref{equ:pprime}, the cost is given by:
\begin{align}
    \mathcal{E}({\bf p}')&=  r' \frac{\mathcal{E}(Z_i) - \mathcal{E}(V)}{\mathcal{A}(Z_i)-\mathcal{A}(V)} [\mathcal{A}(Z_i)-\mathcal{A}(V)]+\mathcal{E}(V).
\end{align}
Simplifying, we get
\begin{align}
    \mathcal{E}({\bf p}')=  \gamma_{\min}^+ \, \Delta \alpha +\mathcal{E}(V)\leq \mathcal{E}({\bf p}).
\end{align}
This shows that ${\bf p}'$ is indeed a solution to our optimisation problem for $\alpha = \mathcal{A}({\bf p}')$. Since the above holds for any $\mathcal{A}(V) < \alpha \leq \mathcal{A}(Z_i)$, we have shown that
\begin{equation}
    {\bf p}(r) = r \, (Z_i - V) + V, \quad r \in (0,1],
\end{equation}
is a solution of the optimisation Eq.~\ref{equ:main_question_pop} for $\alpha \in (\mathcal{A}(V), \mathcal{A}(Z_i)]$. Therefore, given the optimal trajectory up to $V$, we can extend the trajectory from $V$ to another vertex $Z_i$. Moreover, since $Z_i-V$ is a vertex vector of an edge at $V$, Theorem~\ref{thm:edgesadjacentvalued} ensures that $Z_i$ is obtained from $V$ via an av-swap. This completes the induction step and, consequently, proves that the entire optimal trajectory is a sequence of vertices connected by edges generated by av-swaps.

\subsection{The optimal trajectory in the density matrix and unitary spaces}\label{proof:liftedtrajectory}
Our goal now is to extend the optimal trajectory to its corresponding representations in the spaces of density matrices and unitary transformations. Given that the mapping from density matrices to population vectors, as well as the mapping from unitary transformations to doubly-stochastic transformations, are both many-to-one, we can anticipate multiple ways to reproduce the optimal trajectory in these higher-dimensional spaces. Here, we opt for what is arguably the simplest choice.

To start, recall that the population vector represents the diagonal elements of $\rho \in \mathcal{H}_R$ w.r.t. the preferred basis of $\mathcal{H}_R$. This means that a given population vector determines the diagonal elements of $\rho$. Additionally, since the eigenvalues of $\rho$ remain unchanged under unitary transformations, in the specific case where the population vector is a permutation of the eigenvalue vector, $\rho$ must necessarily be diagonal. 

To see why this is true, suppose the population vector of $\rho$ was as above and that $\rho$ had a non-zero off-diagonal element, $\rho_{ij}$. In that case, a unitary transformation could rotate $\rho_{ij}$ into the diagonal elements $\bra{i} \rho \ket{i}$ and $\bra{j} \rho \ket{j}$, leaving all other diagonal elements unchanged. This operation would increase the value of the larger of $\bra{i} \rho \ket{i}$ and $\bra{j} \rho \ket{j}$, violating Schur's theorem. Thus given every vertex $V \in \mathcal{P}$ in the optimal trajectory, the corresponding density matrix $\rho_V$ is fixed: a diagonal state with the elements of $V$  upon its diagonal.

To complete the trajectory we need to determine the density matrix corresponding to an edge between two optimal vertices. Since an edge corresponds to an av-swap, which is a 2-cycle, the points along the edge are represented by partial swaps. At the level of doubly stochastic matrix, this is enacted by the following:
\begin{align}
    D &= t \ket{i}\!\bra{i} + (1-t) \ket{i}\!\bra{j} + \nonumber\\
     &\quad (1-t) \ket{j}\!\bra{i} + t \ket{j}\!\bra{i} + \sum_{n \notin{i,j}} \ket{n}\!\bra{n},
\end{align}
where $i$ and $j$ are the indices of the elements being swapped. This is in fact also a unistochastic matrix, i.e., $D_{mn} = \left| u_{mn} \right|^2$ for some unitary $U=(u_{mn})_{mn}$. By choosing the simplest phase configuration for the unitary elements, one possible form of the corresponding unitary matrix is
\begin{align}\label{equ:2dimunitary}
    U &= \cos\theta \ket{i}\!\bra{i} + \sin\theta \ket{i}\!\bra{j} + \nonumber\\
     &\quad -\sin\theta \ket{j}\!\bra{i} + \cos\theta \ket{j}\!\bra{i} + \sum_{n \notin{i,j}} \ket{n}\!\bra{n},
\end{align}
where $\cos^2\theta = t$,  which represents a rotation within the 2-dimensional subspace of $\mathcal{H}_R$ spanned by $\ket{i}$ and $\ket{j}$. The sought-after density matrix is given by 
\begin{equation}
U \rho_V U^{\dagger}.
\end{equation}
On the other hand, the trajectory in the unitary space is simply given by the appropriate succession of unitaries as in Eq~\ref{equ:2dimunitary}, concatenated via a suitable parametrisation.  
Finally note that all the above expressions are with respect to the preferred basis, i.e. the common eigenbasis of target and cost observables.

In conclusion, constructing the optimal trajectory in both state and unitary spaces is as straightforward as in the population space. We begin by finding the diagonal state $\rho_{\min}$, as well as associated permutation $U_{\min}$, that minimizes the target while satisfying the minimal target cost constraint. Then, we perform a series of 2-level unitary rotations, each swapping the optimal pair of elements, until the target reaches its maximal value.

\subsection{Features of the optimal trajectory for qubit cooling} \label{proof:qubitfeatures}
In this section we apply the results of Sec.~\ref{sec:result} and Sec.~\ref{sec:proofmain} to the specific case of ground state cooling of a qubit system. This allows us to uncover insightful properties of the optimal trajectory in this scenario.

There are two properties of the optimal swap that at first glance are unrelated: 1) that it minimises the gradient of the cost function and 2) that it switches adjacent-valued populations. We can illustrate the connection between them for the case of ground state cooling of a qubit system; in this case we can prove that 1) implies 2). As we will see below, if the machine has non-degenerate eigenvalues this is necessarily true. In the presence of degenerate machine energy, assuming 1) we can always choose an optimal swap satisfying 2).

Take a joint state of the system and machine that is at a vertex of the optimal trajectory. A key property of the state is that it is \textit{passive w.r.t. the system subspaces}, by which we mean: the eigenvalues in the ground/excited subspace of the system are ordered non-increasingly w.r.t. energy,
\begin{align}
    p_{xm} \leq p_{xn} \; \forall \; m > n, \; x \in \{0,1\}.
\end{align}
If this were not the case then one could swap two populations within the same subspace and decrease the energy of the state; this maintains the ground state population of the system while decreasing the energy, contradicting the assumption that the state was optimal.

We can now prove that if the machine has distinct energies, the property of being adjacent-valued is necessary to maintain this ``subspace-passivity'' of the optimal trajectory. Consider the swap of states $\ket{0i}_{SM}$ and $\ket{1j}_{SM}$ with populations $p_{0i} < p_{1j}$, so that cooling is possible via this swap. If these are not adjacently-valued, then either $p_{0,i-1}$ or $p_{1,j+1}$ lies in-between the pair of values. Consider the first case. Then the swap of $p_{0,i-1}$ and $p_{1j}$ would also lead to cooling, but with a gradient
\begin{align} \label{eq:optimalgradientcooling}
    E^{(M)}_{i-1} - E^{(M)}_j-E_S < E^{(M)}_i - E^{(M)}_j-E_S.
\end{align}
This beats the original swap, which is thus not optimal. The second case works analogously. In short, we must have 
\begin{equation}
    p_{1, j+1} \leq p_{0,i} < p_{1,j} \leq p_{0 ,i-1}.
    \end{equation}
    The argument is visualised in Fig. \ref{tab:adjacentvaluedpassive}.

\begin{table}[h]
    \renewcommand{\arraystretch}{1.2}
    \centering
    \begin{tabular}{ |c||c|c|c|c|c|c|c|c| }
        \hline
         & & ... & $\ket{j}_M$ & ... & ... & $\ket{i}_M$ & ... & \\
        \hline\hline
        $\ket{0}_S$ & & & & ... & $p_{0,i-1}$ & $p_{0,i}$ & $p_{0,i+1}$ & ... \\
        \hline
        $\ket{1}_S$ & ... & $p_{1,j-1}$ & $p_{1,j}$ & $p_{1,j+1}$ & ... & & &\\
        \hline 
    \end{tabular}
    \caption{State of system and machine prior to swap of $p_{0,i}$ and $p_{1,j}$. In order for each subspace to remain passive, it must be that the populations are adjacent-valued.}
    \label{tab:adjacentvaluedpassive}
\end{table}

When the machines has degenerate energies, since the inequality of Eq.~\ref{eq:optimalgradientcooling} can become an equality, 1) does not necessarily imply 2) anymore. Nevertheless, assuming 1) one can always choose an optimal swap such that 2) is fulfilled. Note that the above argument only works for binary-valued target functions. For more complicated targets one can prove that a non-adjacent-valued swap can be broken down into a sequence of av-swaps, one of which outperforms the original, this follows from Theorem \ref{thm:edgesadjacentvalued} and Lemma \ref{lemma:edgesminimal}.

Finally, note thta if the system to be cooled is a qubit, then there is only a single system energy gap $E_S$. Every swap that can change the value of the ground state population must involve a pair of states $\ket{0\,i}_{SM},\ket{1\,j}_{SM}$, with populations $p_{0i}<p_{1j}$. The gradient of work cost vs target for this swap is thus
\begin{align}
    \frac{W}{\Delta \alpha} &= E_i^{(M)} - E_j^{(M)} - E_S.
\end{align}
When we compare the gradient for different swaps in order to pick the minimum, the energy gap of the system appears as the same constant everywhere --- it can thus be ignored in constructing the optimal trajectory.

\section{Generalised scenario}
\label{sec:GeneralisedScenario} 

When transforming the state of a system, physical symmetries often impose significant constraints. Examples of such symmetries include temporal, spatial, and rotational symmetries, which correspond to the conservation of the well known physical quantities of energy, momentum, and angular momentum, respectively. In quantum theory, these constraints restrict the allowed transformations of a state. Specifically, if there is a conserved physical observable $\mathcal{O}_C$, the only permissible unitary transformations are those that commute with $\mathcal{O}_C$.

In this section, we extend our main result by introducing a generalized framework that includes an additional conserved observable, $\mathcal{O}_C$, which commutes with both the target and cost observables. This generalization addresses a broader class of problems, as it applies to any such  $\mathcal{O}_C$, with our original result emerging as a special case when $\mathcal{O}_C=\mathds{1}$. We begin by presenting the generalized result and then demonstrate its application to the specific context of incoherent cooling, a central paradigm within quantum thermodynamics.

As in the original version of our problem, let $R$ be a finite-dimensional quantum system initially in a state $\rho \in \mathcal{L} (\mathcal{H}_R)$, and let there be two commuting observables, $\mathcal{O}_{\mathcal{A}}$ resp. $\mathcal{O}_{\mathcal{E}}$ that we refer to as the target and cost observable. We now additionally bring into the picture a third observable, $\mathcal{O}_C$, that commutes with both $\mathcal{O}_{\mathcal{A}}$ and $\mathcal{O}_{\mathcal{E}}$, and we demand that $\mathcal{O}_C$ is conserved upon our unitary state transformation. So in short $\mathcal{O}_{\mathcal{A}}, \mathcal{O}_{\mathcal{E}}$ and $\mathcal{O}_{C}$ form a set of commuting observables and we restrict our unitary transformations to those commuting with $
\mathcal{O}_C$. Our modified question can hence be formulated as follows. For any achievable value of $\alpha$, what is
\begin{equation}
\begin{aligned}
    \min_{\{ U \; : \; [U, \mathcal{O}_C]=0\} } \text{Tr} \left( \mathcal{O}_{\mathcal{E}} \,  \sigma \right), \quad \text{s.t. }& \text{Tr}\left( \mathcal{O}_{\mathcal{A}} \, \sigma \right) = \alpha,
    \end{aligned}
    \end{equation}
with $\sigma = U \rho U^{\dagger}$. As we can readily see, the original problem of Eq.~\ref{equ:main_question_pop} is recovered by  taking $\mathcal{O}_C = \mathds{1}$ and noting that $\mathcal{X}[\sigma] = \text{Tr} \left( \mathcal{O}_X \sigma \right)$, where $X=\mathcal{A}, \mathcal{E}.$

The analysis in the case of a conserved quantity runs analogously to the general case. The key revision is that the new set of attainable populations forms a direct product of population polytopes rather than a single one. In the following sections we spell this out in greater details.

\subsection{Conserved subspaces and incoherent states}
To begin with we decompose the Hilbert space $\mathcal{H}$ into a direct sum of eigenspaces of $\mathcal{O}_C$ (we omit the subscript $R$ for simplicity)
\begin{align}
    \mathcal{H} &= \bigoplus_{i=1}^{N_C} \mathcal{H}^{(i)},
\end{align}
where $N_C$ is the number of distinct eigenvalues of $\mathcal{O}_C$. If we label the projector onto $\mathcal{H}^{(i)}$ as $\Pi^{(i)}$, then by inserting twice the identity $\mathds{1}=\sum_i \Pi^{(i)}$,  we can split the initial state into two parts:
\begin{align}
    \rho &= \sum_{i=1}^{N_C} \rho^{ii} + \sum_{i \neq j} \rho^{ij}\\
    &= \rho^{dephased} + \rho^{phase},
\end{align} 
where $\rho^{ij} := \Pi^{(i)} \rho \; \Pi^{(j)} $. Under energy-preserving unitaries, $\rho^{dephased}$ is independent of $\rho^{phase}$, i.e. its elements are not affected by the elements of the latter~\cite{lipkabartosik-2023}. Indeed, by direct calculation
\begin{equation}
    \sigma^{ij} = U^{ii} \rho^{ij} \left( U^{jj} \right)^\dagger,
\end{equation}
where $X^{ij}=\Pi^{(i)} X \Pi^{(j)}$, $X=\sigma, U, \rho$. Thus we need not keep track of any of the phases, and can instead simply consider the dephased state. Crucially, this also means that for the purpose of the Schur-Horn theorem which we will use shortly, we can consider the eigenvalues of the state to be that of the dephased state. Going forward we will stop keeping track of $\rho^{phase}$, effectively working with the dephased version, that we call an \textit{incoherent} state.

As the state (or rather the effective dephased version) is itself block-diagonal w.r.t. the eigenbasis of $\mathcal{O}_C$, we can express it as
\begin{align}\label{eq:rhodirectsum}
    \rho &= \bigoplus_{i=1}^{N_C} \rho^{(i)},
\end{align}
where $\rho^{(i)} = \rho^{ii}$ from above. This structure has a critical implication: the eigenvalues of the full state can be determined independently from each subspace. Specifically, we have:
\begin{align}
    \boldsymbol{\lambda}\left[ \rho \right] &= \bigoplus_{i=1}^{N_E} \boldsymbol{\lambda} \left[ \rho^{(i)} \right],
\end{align}
where $\boldsymbol{\lambda}[\sigma]$ stands for the vector of eigenvalues of $\sigma$, with $\sigma$ representing either $\rho$ or $\rho^{(i)}$. For ease of notation we label the eigenvalue vector corresponding to the $i^{th}$ block, $\boldsymbol{\lambda}\left[ \rho^{(i)} \right]$, as $\boldsymbol{\lambda}^{(i)}$.

As the unitary operation commutes with $\mathcal{O}_C$, it also inherits a block-diagonal structure with respect to the eingenbasis of $\mathcal{O}_C$:
\begin{align}
    U &= \bigoplus_{i=1}^{N_C} U^{(i)},
\end{align}
where each $U^{(i)}=U^{ii}$ is a unitary operator acting on $\mathbb{C}^{d_i}$. Here $d_i$ represents the dimension of the $i^{th}$ subspace, corresponding to the multiplicity of the $i^{th}$ distinct eigenvalue of $\mathcal{O}_C$.

To finalize our setup, and in direct analogy to the setup of our original problem, we select a basis—referred to as {\it the preferred basis}—in which all three operators $\mathcal{O}_C$, $\mathcal{O}_A$ and $\mathcal{O}_E$ are diagonal. This choice is possible because the operators pairwise commute. Notably, in this basis, both the target and cost functions are linear with respect to the diagonal elements of the state $\rho$ -- referred to as the {\it poppulations} of $\rho$.

\subsection{Generalized population polytope}
\label{subsec:PG}
We now turn our attention to investigating the set of achievable population vectors, $\bf{p}$, within the generalised scenario. We find that this set also forms a polytope, which we refer to as the { \it generalized population polytope} $\mathcal{P}_G$.

First observe that within each block $\mathcal{H}^{(i)}$, the Schur-Horn theorem applies just as in the original problem. Specifically, the diagonal elements with respect to the preferred basis in each block form a vector, $\textbf{p}^{(i)}$, which is majorized by the eigenvalues of that block:
\begin{align}
    \boldsymbol{\lambda}^{(i)} \succ \textbf{p}^{(i)}.
\end{align}
This relationship implies the existence of a doubly stochastic matrix $D^{(i)}$ such that $\textbf{p}^{(i)} = D^{(i)} \boldsymbol{\lambda}^{(i)}$. By varying over all possible unitaries within each block independently, we obtain the entire set of doubly-stochastic transformations of $\boldsymbol{\lambda}^{(i)}$ for each block.

Building on the results of the original problem, we conclude that the set of attainable populations within each block $\mathcal{H}^{(i)}$ is the population polytope whose vertices are permutations of $\boldsymbol{\lambda}^{(i)}$. Extending this to the overall system, the population vector ${\bf p}$ can be expressed, thanks to Eq.~\ref{eq:rhodirectsum}, as:
\begin{align}
    \textbf{p} &= \bigoplus_{i=1}^{N_C} \textbf{p}^{(i)},
\end{align}
where $\textbf{p}^{(i)} \in \mathbb{R}^{d_i}$ is the population of the $i^{\text{th}}$ subspace. Consequently, the set of allowed population vectors may be expressed as the direct product of the population polytopes of each subspace:
\begin{align}
    \mathcal{P}_G &= \mathcal{P}^{(1)} \times ... \times \mathcal{P}^{(N_C)},
\end{align}
where each $\mathcal{P}^{(i)} \in \mathbb{R}^{d_i}$ is the convex hull of permutations of $\boldsymbol{\lambda}^{(i)}$.

As a direct product, the structure of $\mathcal{P}_G$ follows straightforwardly from its constituent polytopes. In Appendix \ref{app:facesdirectproduct} we show that the faces of $\mathcal{P}_G$ are the direct product of faces from each constituent polytopes. In particular, the edges of $\mathcal{P}$ arise from the direct product of an edge from any one $\mathcal{P}^{(i)}$ with vertices from the remaining polytopes. Furthermore, we already know how to generate the edges of $\mathcal{P}^{(i)}$: namely via an adjacently-valued swap within that subspace, where adjacency is defined solely with respect to the eigenvalues $\boldsymbol{\lambda}^{(i)}$.

\subsection{Optimal trajectory of the generalized scenario}

We may now construct the optimal trajectory of the generalized scenario. As for the original scenario, we begin by picking a point corresponding to the minimum of the target function, i.e. $\alpha_{min}$. As we can pick points within each polytope $\mathcal{P}^{(i)}$ independently, picking the minimal point in $\mathcal{P}_G$ corresponds to picking a minimal point w.r.t. each $\mathcal{P}^{(i)}$. More precisely, we can express the cost and target vector of coefficients as direct sums w.r.t. the subspaces $\mathcal{H}^{(i)}$,
\begin{align}
    \textbf{a} &= \textbf{a}^{(1)} \oplus ... \oplus \textbf{a}^{(N_C)}, \\
    \textbf{E} &= \textbf{E}^{(1)} \oplus ... \oplus \textbf{E}^{(N_C)}.
\end{align}
We may then express the cost and target functions as
\begin{align}
    \mathcal{A}(\textbf{p}) &= \sum_{i=1}^{N_C} \textbf{a}^{(i)} \cdot \textbf{p}^{(i)}, \\ 
    \mathcal{E}(\textbf{p}) &= \sum_{i=1}^{N_C} \textbf{E}^{(i)} \cdot \textbf{p}^{(i)}.
\end{align}
From each $\mathcal{P}^{(i)}$, we select a vertex $V^{(i)}$ that minimizes the value of $\textbf{a}^{(i)} \cdot \textbf{p}^{(i)}$. If multiple vertices achieve this minimum, we choose among them the one that has the minimal $\textbf{E}^{(i)} \cdot \textbf{p}^{(i)}$. The direct sum $V^{(1)} \oplus \cdots \oplus V^{(N_C)}$ then defines the minimal point along the optimal trajectory.

Continuing, recall that we proved in Sec. \ref{proof:gradientedge} that one can move along the optimal trajectory from any vertex $V$ by selecting an edge (at $V$) of the population polytope with the minimal cost versus target gradient. As the edges of the generalized polytope correspond to av-swaps within a single subspace, we simply pick the av-swap having the minimal gradient among all subspaces to continue the trajectory until the next vertex. Proceeding in this fashion, we eventually reach the maximal point of the trajectory.

Note that an equivalent description of the optimal trajectory is as follows. One can treat each pair of scalar products $ \{\textbf{a}^{(i)} \cdot \textbf{p}^{(i)}, \textbf{E}^{(i)} \cdot \textbf{p}^{(i)}\}$ as cost and target functions for the $i^{th}$ subspace. Using these, one can independently construct the optimal trajectory of cost versus target for each polytope $\mathcal{P}^{(i)}$. Each trajectory corresponds to a sequence of av-swaps progressing from a minimal to a maximal vertex, as established in our treatment of the original problem. The globally optimal trajectory is then constructed by starting from a direct sum of the minimal vertices and, at each step, choosing the most efficient av-swap among the individual trajectories. This process continues until reaching the direct sum of the maximal vertices.

\subsection{Application: incoherent cooling}

A thermodynamic application to the generalised scenario is incoherent cooling, the motivation for which stems from the need to spell out more explicitly the resources used to perform the transformation from the coherent cooling scenario. Indeed, coherent cooling assumes the ability to perform any joint unitary on the system and machine. This freedom requires two implicit resources:
\begin{enumerate}
    \item a \textit{battery}, i.e. a system that can exchange energy freely with the system and machine,
    \item a \textit{clock}, i.e. a phase reference that allows one to keep track and manipulate evolving phases in the system and machine.
\end{enumerate}

Alternatively, one may work without the above resources by shifting to an \textit{incoherent} scenario. Rather than an implicit battery one takes a (explicit) heat bath $B$ as the source of energy, and the unitaries are constrained to be energy-preserving on the joint $SMB$ system. The question of state transformation is then that of transforming a general state of the system $\rho_S$ via an energy-preserving unitary on $SMB$,
\begin{align}
    \rho_S^\prime &= \Tr_{MB} \left[ \tilde{U}_{SMB} \left( \rho_S \otimes \tau_M \otimes \tau_B \right) \tilde{U}_{SMB}^\dagger \right],
\end{align}
where $\tilde{U}_{SMB}$ commutes with the total Hamiltonian --- which consists of the sum of local Hamiltonians --- and $\tau_X = e^{-\beta_X H_X}/\Tr \left[e^{-\beta_X H_X}\right]$ is a thermal state of system $X$ at inverse temperature $\beta_X$. The heat source is hotter than the temperature of the environment, i.e. $\beta_B < \beta_M$ so that the process is mediated by the heat flow from $B$ to $M$, which increases entropy.

The target function is typically the population of the ground state of the system, corresponding to the observable $\Pi^{(0)}_S \otimes \openone_M \otimes \openone_B$, the cost function is the average energy of the hot bath,\footnote{Here too the cost is the \textit{change} in average energy of $B$, but the initial energy is independent of the unitary transformation, so we may ignore it for the optimisation and subtract it at the end.} corresponding to the observable $\openone_{SM} \otimes H_B$. These two together with the total energy observable form a trio of pairwise commuting observables, and the preferred basis is the tensor product of the energy eigenbases of each system.

\subsection{Incoherent cooling example}

As an example to the above incoherent cooling scenario, consider a qubit system $S$, qutrit machine $M$, and qubit heat source $B$, all with an energy spacing of $\Delta$. The joint Hamiltonian has twelve energy eigenstates, falling into five degenerate spaces as displayed in Fig. \ref{fig:incoherent}.
\begin{figure}[h]
    \centering
    \includegraphics[width=0.9\columnwidth]{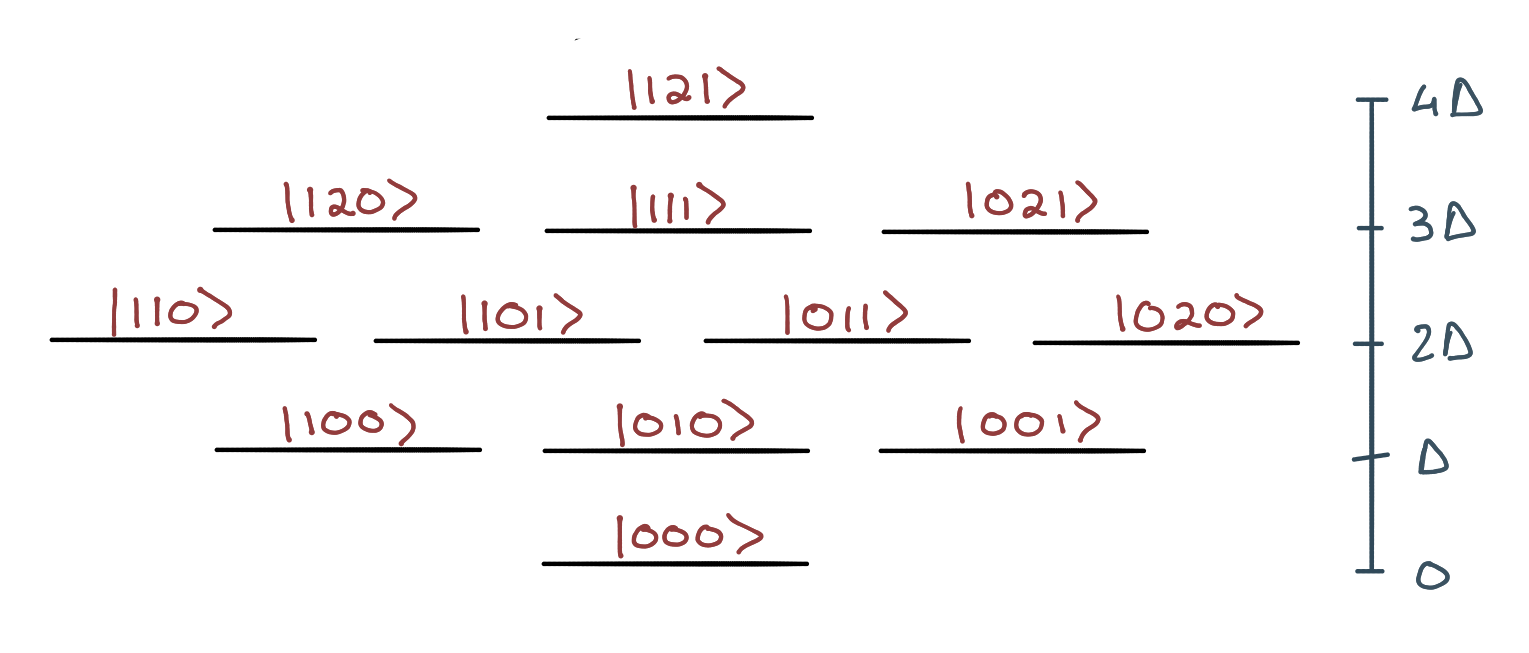}
    \caption{Degenerate energy subspaces for qubit $S$, qutrit $M$ and qubit $B$ each with level spacing $\Delta$. Within each subspace every doubly stochastic transformation of the populations is possible, the resulting polytope is the convex hull of the direct sums of permutations in each block.}
    \label{fig:incoherent}
\end{figure}

The subspaces $\mathcal{H}^{(i)}$ therefore correspond to the eigenspaces $\text{Eig}(0)$, $\text{Eig}(\Delta)$, $\text{Eig}(2 \Delta)$, $\text{Eig}(3 \Delta)$ and $\text{Eig}(4 \Delta)$. And indeed, as should be,
\begin{equation}
    \mathcal{H} = \oplus_{i=0}^4 \text{Eig}( i \Delta).
\end{equation}
We next decompose $\boldsymbol{\lambda}$, the vector of eigenvalues of $\rho_S \otimes \tau_M \otimes \tau_B$, in the above subspaces, delivering
\begin{equation}
    \boldsymbol{\lambda}=\boldsymbol{\lambda}^{(0)} \oplus \boldsymbol{\lambda}^{(1)} \oplus \boldsymbol{\lambda}^{(2)} \oplus \boldsymbol{\lambda}^{(3)} \oplus \boldsymbol{\lambda}^{(4)}.
\end{equation}
If we next denote the population polytope associated to $\text{Eig}( j \Delta)$ by $\mathcal{P}^{(j)}$, applying what we learned from Sec.~\ref{subsec:PG}, the generalized population polytope $\mathcal{P}_G$ is thus 
\begin{equation}
    \mathcal{P}_G= \mathcal{P}^{(0)} \times \mathcal{P}^{(1)} \times \mathcal{P}^{(2)} \times \mathcal{P}^{(3)} \times \mathcal{P}^{(4)}.
\end{equation}
Furthermore, the vertices of $\mathcal{P}_G$ are given by the direct sum of permutations of the $\boldsymbol{\lambda}^{(i)}$. To this end, note that as $\text{Eig}(0)$ and $\text{Eig}(4 \Delta)$ are singular points, $ \boldsymbol{\lambda}^{(0)}$ and $\boldsymbol{\lambda}^{(4)}$ are in fact real numbers. In contrast, $\boldsymbol{\lambda}^{(1)}$ and $\boldsymbol{\lambda}^{(3)}$ are 3-dimensional vectors, and $\boldsymbol{\lambda}^{(2)}$ is a 4-dimensional vector. As the number of vertices of $\mathcal{P}_G$ is the product of the number of vertices of the respective $\mathcal{P}^{(i)}$, it has
\begin{equation}
    6 \times 24 \times 6 = 864
\end{equation}
vertices. Note, however, that for most practical purposes, most of the swaps do not lead to cooling. For example, suppose that $S$ is initially thermal at the environment temperature and that the temperature of $B$ is infinite, then the only energy preserving swap leading to cooling of $S$ is the $\ket{101} \leftrightarrow \ket{020}$ swap.

\section{Conclusion and outlook}

\textbf{Summary.} In this work, we looked at the unitary transformation of a quantum system with two relevant commuting observables. We were interested in minimising the average value of the first observable (the cost observable) while reaching a desired average value on the second observable (the target observable). Solving this problem provides a general framework for addressing a broad class of thermodynamic questions, including determining the minimal energy cost required to cool a quantum system to a specified temperature and identifying the optimal unitaries for such cooling processes.

The crux of the above depicted problem is that the set of unitaries is a large and topologically non-trivial space that makes it challenging to optimize over. To address this issue, we systematically reformulated the problem through a sequence of mathematical reductions. Starting from the space of 
$d$-dimensional unitary operations, we utilized the Schur-Horn and Birkhoff-von Neumann theorems to simplify the problem into the Birkhoff polytope of doubly stochastic matrices, and ultimately to the population polytope $\mathcal{P}$. This allowed us to define the minimal cost function $\omega_{opt}(\alpha)$ and identify the corresponding state trajectories, including the explicit construction of an optimal trajectory. 

Due to the convexity of the population polytope, the optimal cost function is also convex. This property, along with properties of the edges of the population polytope, enabled us to develop a framework based on two-level swaps to derive the trajectory that transitions from an optimal state of minimal target value $\alpha_{min}$ to an optimal state of maximal target value $\alpha_{max}$. For any given target value $\alpha$, the cost function is minimized along this trajectory, which we refer to as the optimal trajectory. Importantly, the entire trajectory can be composed as a sequence of 2-level swaps. 

We then proceeded to show how this result solves the question of determining the minimal energy cost as well as optimal state transformation required for cooling a finite-dimensional quantum system, referred to as the {\it system}, with the help of another finite dimensional quantum system, referred to as the {\it machine}, under joint unitary operations. We furthermore explicitly worked out the minimal energy cost required for ground state cooling a qubit system with a given 4-dimensional machine in this scenario. 

We finally solved a generalized version of the above optimization problem to the case where an additional commuting conserved observable was taken into account. This in turn solves another thermodynamic scenario, namely that of determining the optimal transformation when cooling a finite dimensional quantum system with the help of a finite dimensional machine and a finite dimensional resourceful bath, the allowed transformation being a global energy conserving unitary.

\textbf{Virtual qubits and temperatures.} Our work emphasizes the central role of 'virtual qubits' in the context of quantum thermodynamics and state transformation. The concepts of `virtual qubits' and `virtual temperatures' as a means of understanding thermal machines with discrete energy levels were first introduced in \cite{Brunner2012}, showing how the smallest possible machines, such as a 2-qubit engine or a 3-qubit fridge, could be understood as coupling the system of interest to a suitable 2-level energy gap of the machine.

Since then, a number of results have demonstrated the utility of virtual qubits in finding temperature limits in quantum thermal machines. Put simply, finding the lowest/highest possible temperature amounts to determining the coldest/hottest/most inverted virtual qubit that can be made in the machine, regardless of machine size or structure \cite{Silva2016,Clivaz-2019}, including in the presence of coherence between energy levels \cite{soldati2024}.

This paper extends the above argument to the problem of efficiency: here we not only care about the target variable (e.g. temperature), but also the cost for achieving said temperature (e.g. energy). Once again, the problem is reduced to finding the optimal 2-level operation, i.e., virtual qubit. In particular, for qubit cooling the problem reduces to finding the smallest energy gap in the machine that is `colder' --- i.e. has a higher population ratio in ground vs excited --- than the system.

\textbf{Complexity of the optimal trajectory.} Now that we have reduced the question of optimal unitary transformation to that of qubit gates, a natural follow-up is to investigate the `complexity' of the optimal trajectory. In particular:
\begin{enumerate}
    \item the `length', i.e. the number of qubit gates that form the optimal trajectory,
    \item the `complexity' of each step, i.e. how easy it is to determine the optimal qubit gate from among all possible qubit gates.
\end{enumerate}
The latter question is especially relevant for the feasibility of programming the optimal trajectory for scenarios beyond the scope of analytical results.

\textbf{Sequential unitaries.} In this work we focused exclusively on a single unitary application. A natural extension to both thermodynamic scenarios we looked at is however to allow for repeated interaction between the system and machine, where the machine typically thermalises back to its initial state in between the interactions. The question of energetically efficient transformation is also highly relevant in that realm. We hope that our work helps shed some light into the relevant tools to tackle energetic efficiency in that regime.

\textbf{Optimal pair of machine and unitary.} Finally, while we have here assumed that the machine that was brought into the picture to help cool the system of interest was given, we might assume some flexibility in engineering said machine. The agent preparing the machine might for example be constrained to deliver a machine of a fixed dimension but could be able to adjust its energy levels but for example adjusting external fields accordingly. In that context, the question of which machine, as well as unitary, allows cooling the system most efficiently naturally comes to mind. 

\section{Acknowledgments}

F.C thanks Ludovico Lami for useful discussions on polytope theory. R.S. acknowledges funding from the Swiss National Science Foundation via an Ambizione grant PZ00P2\_185986. P.B. and F.C. acknowledges funding from the European Research Council (Consolidator grant ‘Cocoquest’ 101043705) and financial support from the Austrian Federal Ministry of Education, Science and Research via the Austrian Research Promotion Agency (FFG) through the project FO999914030 (MUSIQ) funded by the European Union – NextGenerationEU.

F.C. and R.S. conceptualized the problem and developed the major part of the analytical results. F.C. and P.B. coordinated the project and writing. All authors contributed to ideas, discussions and the writing of the paper. F.C. designed Fig. \ref{fig:differentlevels} and all the polytope visualisations.

\bibliographystyle{apsrev4-2}
\bibliography{optimalcoolingbibliography}

\clearpage
\onecolumngrid
\appendix

\newpage
\onecolumngrid
\appendix

\section*{Appendices}
\setcounter{section}{0}
\tableofcontents

\input{appendix}
\end{document}

%% file: macros.tex
\usepackage{amsmath}
\usepackage[english]{babel}
\usepackage[T1]{fontenc}
\usepackage[colorlinks,citecolor=blue,linkcolor=blue]{hyperref}
\usepackage[utf8]{inputenc}
\usepackage{bbm}
\usepackage[dvipsnames]{xcolor}
\usepackage{comment}
\usepackage{minitoc}
\usepackage{ulem}
\usepackage{tabularx}

\usepackage{amsfonts}
\usepackage{amssymb}
\usepackage{amsthm}
\usepackage{braket}
\usepackage{dsfont}
\usepackage[dvipsnames]{xcolor}
\definecolor{nice}{RGB}{230,0,230}
\definecolor{darkgreen}{RGB}{50,190,50}
\usepackage{tikz}

\theoremstyle{plain}
\newtheorem{theorem}{Theorem}
\newtheorem{lemma}{Lemma}
\newtheorem{thmcorollary}{Corollary}[theorem]
\newtheorem{lemcorollary}{Corollary}[lemma]

\theoremstyle{definition}
\newtheorem{definition}{Definition}

\theoremstyle{remark}
\newtheorem{remark}{Remark}

\DeclareMathOperator{\Tr}{Tr}



%% file: appendix.tex
\section{Proofs regarding general polytopes}

In this section we list results applicable to general polytopes that we use in our result.

\subsection{An edge is the convex hull of two vertices} \label{subsection:generatingedge}

\begin{lemma} \label{lemma:generatingedge}
   Given an edge $\mathcal{D}$ of a polytope $\mathcal{P}$, there exists exactly two vertices $V_i$, $V_j$ of $\mathcal{P}$ such that
   \begin{equation}
       \mathcal{D}=\{ r (V_i-V_j) +V_j \mid r \in [0,1]\}.
   \end{equation}
\end{lemma}

\begin{proof}
    Let $\mathcal{D}$ be an edge of $\mathcal{P}$. So $\mathcal{D}$ is a non-empty face of $\mathcal{P}$. Hence by Theorem 7.3, p. 45 of \cite{Brondsted-1983}, $\mathcal{D}$ is a polytope with vertices being a subset of those of $\mathcal{P}$. Let us denote these vertices by $V_1, \dots, V_n$. So $V_1, \dots, V_n \in \text{ext}(\mathcal{P})$ and $\mathcal{D}=\text{conv}\{V_1, \dots, V_n\}$. We want to prove that $n=2$ must hold. First, $n \geq 2$ must hold for if $n=0$, $\mathcal{D}= \emptyset$ and so $\text{dim}(\mathcal{D})=-1$, which is a contradiction to $\text{dim}(\mathcal{D})=1$. Similarly for $n=1$, we get $\mathcal{D}= \text{conv}(V_1) = \{V_1\}$ and $\text{dim}(\mathcal{D})=0$, also a contradiction to $\text{dim}(\mathcal{D})=1$. This establishes $n \geq 2$. To prove $n \leq 2$, suppose to the contrary that $n > 2$. So $ \mathcal{D}$ contains at least three distinct vertices of $\mathcal{P}$, say $V_1, V_2, V_3$. Since $V_1 \neq V_2$, $V_1$ and $V_2$ are affinely independent and therefore constitute an affine basis of $\text{aff}(\mathcal{D})$. So in particular since $V_3 \in \text{aff}(\mathcal{D}) $, there is a $\lambda \in \mathds{R}$ such that
    \begin{equation}
    V_3 = \lambda V_1 + (1-\lambda) V_2.
    \end{equation}
    If $\lambda=0$ resp. $\lambda=1$ we get $V_3=V_2$ resp. $V_3=V_1$, a contradiction to $V_1, V_2, V_3$ being distinct. If $\lambda \in\, ]0,1[$, then we get that $]V_1, V_2[\, \cap \{V_3\} \neq \emptyset$, where $]V_1,V_2[$ denotes the open segment between $V_1$ and $V_2$, i.e.,
    \begin{equation}
        ]V_1,V_2[\,=\{\mu V_1 + (1-\mu) V_2 \mid \mu \in ]0,1[\,\}.
    \end{equation}
    As $\{V_3\}$ is a face of $\mathcal{P}$, this implies $V_1,V_2 \in \{V_3\}$, again a contradiction to $V_1,V_2,V_3$ being distinct. If $\lambda < 0$ then with $0<\mu =\frac{1}{1-\lambda} < 1$ we get
    \begin{equation}
        V_2 = \mu V_3 + (1-\mu) V_1,
    \end{equation}
    which leads to $]V_3,V_1[\, \cap\,  \{V_2\} \neq \emptyset$. Now using that $\{V_2\}$ is a face of $\mathcal{P}$, this implies $V_3, V_1 \in \{V_2\}$, once more a contradiction to $V_1, V_2,V_3$ being distinct. So $\lambda > 1$ must hold. But then with $0<\mu=\frac{1}{\lambda}<1$,
    \begin{equation}
        V_1 = \mu V_3 + (1-\mu) V_2,
    \end{equation}
    and we get $]V_3, V_2[\, \cap \{V_1\} \neq \emptyset$. Now using that $\{V_1\}$ is a face of $\mathcal{P}$ we get that $V_3, V_2 \in \{V_1\}$, again a contradiction to $V_1, V_2, V_3$ being distinct. So for all possible values of $\lambda$ we get a contradiction, which shows that $n \leq 2$ must hold. And so, as desired, $n=2$. The unicity of the two vertices follows directly from the fact that if there were vertices $V_k, V_l$ such that $\mathcal{D} = \text{conv}\{V_k, V_l\}$, then as $\mathcal{D} = \text{conv}\{V_1,V_2\}$ we get $V_k, V_l \in \text{conv}\{V_1,V_2\}$, which using that $\{V_1\}, \{V_2\}, \{V_k\} $ and $\{V_l\}$ are faces of $\mathcal{P}$ similarly as above implies $\{V_k,V_l\} = \{V_1, V_2\}$.
\end{proof}
\subsection{The vertices of a translated polytope are the translated vertices}

\begin{lemma}
    Let $\mathcal{P}$ be a polytope in $\mathds{R}^d$, $V \in \text{ext}(\mathcal{P})$. Then $\mathcal{P}-V$ is also a polytope and $\text{ext}(\mathcal{P}-V) = \{V_i-V \mid V_i \in \text{ext}(\mathcal{P})\}$.
\end{lemma}
\begin{proof}
   $\mathcal{P}$ is a polytope in $\mathds{R}^d$ so there exists a $N \in \mathds{N}$ and $V_1, \dots, V_N \in \mathds{R}^d$ such that
   \begin{equation}
       \mathcal{P}=\text{conv}\{V_1,\dots,V_N\}.
   \end{equation}
   W.l.o.g. $V=V_1$. Let $\mathcal{P}_1 = \mathcal{P}-V_1$. We now claim that
   \begin{equation}
   \mathcal{P}_1= \text{conv}\{0, V_2-V_1, \dots, V_N-V_1\}.
   \end{equation}
   We start by showing $\mathcal{P}_1 
   \subset \text{conv}\{0, V_2-V_1, \dots, V_N-V_1\}$. To that end let $a \in \mathcal{P}_1$. So there are $\lambda_1, \dots, \lambda_N \in [0,1]$, with $\sum_{i=1}^N \lambda_i=1$ such that
   \begin{align}
       a &= \left(\sum_{i=1}^N \lambda_i V_i \right)-V_1= \sum_{i=1}^N \lambda_i V_i - \sum_{i=1}^N \lambda_i V_1\\
       &= \sum_{i=1}^N \lambda_i (V_i-V_1) \in \text{conv}\{0, V_2-V_1, \dots, V_N-V_1\}.
   \end{align}
   To show $ \text{conv}\{0, V_2-V_1, \dots, V_N-V_1\} \subset \mathcal{P}_1$ pick $a \in \text{conv}\{0, V_2-V_1, \dots, V_N-V_1\}$. So there exists $\lambda_1, \dots, \lambda_N \in [0,1]$, $\sum_{i=1}^N \lambda_i=1$ with
   \begin{align}
       a&= \sum_{i=1}^N \lambda_i (V_i -V_1)= \sum_{i=1}^N \lambda_i V_i - \sum_{i=1}^N \lambda_i V_1\\
       &= \left(\sum_{i=1}^N \lambda_i V_i \right)-  V_1 \in \text{conv}\{V_1, \dots, V_N\} -V = \mathcal{P}-V=\mathcal{P}_1.
   \end{align}
   So $\mathcal{P}_1$ is a polytope and with Theorem 7.2 of \cite{Brondsted-1983}
   \begin{equation}
       \text{ext}(\mathcal{P}_1) \subset \{0,V_2-V_1, \dots, V_N-V_1\}.
   \end{equation}
   We claim that in fact $\text{ext}(\mathcal{P}_1) =\{0,V_2-V_1, \dots, V_N-V_1\}$. To prove our claim let $i \in \{1, \dots, N\}$ and suppose to the contrary that $V_i-V_1$ is not an extreme point of $\mathcal{P}_1$. Then by definition there exist $y$ and $z$, two distinct points of $\mathcal{P}_1$, such that 
   \begin{equation}
       ]y,z[ \,\cap \{V_i-V\} \neq \emptyset.
   \end{equation}
   This means there is a $\lambda \in ]0,1[$ such that
   \begin{equation}
       V_i-V_1= \lambda y + (1-\lambda) z.
   \end{equation}
   Now, as $y, z \in \mathcal{P}_1= \mathcal{P}-V_1$ there are $a,b \in \mathcal{P}$ such that $y= a-V_1$ and $z=b-V_1$. Note that since $y \neq z$, $a \neq b$. So
   \begin{equation}
       V_i-V_1 = \lambda (a-V_1) + (1-\lambda) (b-V_1) = \left( \lambda a + (1-\lambda) b \right)-V_1.
   \end{equation}
   And therefore $V_i= \lambda a + (1-\lambda) b$ meaning there exist distinct $a,b \in \mathcal{P}$ such that
   \begin{equation}
       ]a,b[ \cap \{V_i\} \neq \emptyset,
   \end{equation}
   a contradiction to $V_i \in \text{ext}(\mathcal{P})$. And so, as wished $V_i-V_1 \in \text{ext}(\mathcal{P}_1)$ for any $i \in \{1, \dots, N\}$.
\end{proof}

\subsection{A polytope lies within the cone generated by vertex vectors at any vertex $V$ (Lemma \ref{lemma:edgesminimal})}\label{app:edgesminimalset}

We here prove Lemma~\ref{lemma:edgesminimal} of the main text.

\begin{proof}[Proof of Lemma~\ref{lemma:edgesminimal}]
Let $V_1, \dots, V_N$ be the vertices of $\mathcal{P}$ and let w.l.o.g. $V=V_1$. Let us also denote the vertex vectors at $V_1$ by $V_{i1} = V_i-V_1$. 
Let $ {\bf p} \in \mathcal{P}$ be a point of our polytope. Then there are $\lambda_i, \dots, \lambda_N \in [0,1]$ such that $\sum_{i=1}^N \lambda_i=1$ for which
\begin{equation}
    {\bf p} = \sum_{i=1}^N \lambda_i V_i.
\end{equation}
Also
\begin{align}
    {\bf p}-V_1 &= \left( \sum_{i=1}^N \lambda_i V_i \right) - V_1= \left(\sum_{i=1}^N \lambda_i V_i \right) - \left( \sum_{i=1}^N \lambda_i V_1 \right)= \sum_{i=2}^N \lambda_i V_{i1}. \label{equ:conevertexvectorsatV}
\end{align}
This shows that ${\bf p}- V_1$ lies in the cone of all the vertex vectors at $V$. We have left to show that this is still true when we only sum over the vertex vectors of edges at $V$ on the right hand side of Eq.~\ref{equ:conevertexvectorsatV} (with possibly differing coefficients) . Or in other words that the vertex vectors at $V$ that do not belong to an edge at $V$ are within the cone of vertex vectors of edges at $V$. To this end, let $i_1, \dots, i_k$ be the indices of $\{2, \cdots, N\}$ for which  $V_{i_11}, \dots, V_{i_k1}$ generate an edge at $V_1$. We would like to show that 
\begin{lemma} \label{lemma:edgereducing}
    for all $l \in \{2,\dots, N\} \setminus \{i_1, \dots, i_k\}$ there exist $c_{l1},\dots,c_{lk} \geq 0$ such that
\begin{equation}
    V_{l1}= \sum_{j=1}^k c_{lj} V_{i_j1}.
\end{equation}
\end{lemma}
\begin{proof}
We construct a vertex figure of $\mathcal{P}$ at the vertex $V$ as done in \cite[Chapter~2.1]{Ziegler-1995}. To this end let $ {\bf c} \in \mathds{R}^d, c_1 \in \mathds{R}$ and let ${\bf c \cdot x} \leq c_1$ be a valid inequality for $\mathcal{P}$\footnote{This is an equality that is satisfied for all points of $\mathcal{P}$, i.e., whenever $ {\bf x} \in \mathcal{P}$.} with
\begin{equation}
    \{V_1\}= \mathcal{P} \cap \{ {\bf x} \mid {\bf c \cdot x}=c_1\}.
\end{equation}
Let also $\alpha \in \mathds{R}$ with $\alpha < c_1$ and ${\bf c} \cdot V_n < \alpha$ for all $n = 2,\dots,N$. We are now interested in
\begin{equation}
    \mathcal{P}'= \mathcal{P} \cap \{{\bf x} \mid {\bf c \cdot x} = \alpha\},
\end{equation}
which is called a vertex figure of $\mathcal{P}$ at $V_1$ in~\cite[Chapter~2.1]{Ziegler-1995}. First note that $\mathcal{P}'$ is a polytope. This is because 
\begin{equation}
    \tilde{\mathcal{P}} = \mathcal{P} \cap \{ {\bf x} \mid {\bf c \cdot x} \leq \alpha\}
\end{equation}
being the intersection of finitely many closed halfspaces and bounded (since $\mathcal{P}$ is bounded) is a ($\mathcal{H}$-)polytope. And since $\mathcal{P}'= \mathcal{P}' \cap \{{\bf x} \mid {\bf c \cdot x} \leq \alpha\}=\tilde{\mathcal{P}} \cap \{{\bf x} \mid {\bf c \cdot x} = \alpha\}$ is a face of $\tilde{\mathcal{P}}$ it is, according to Proposition~2.3 of \cite{Ziegler-1995}, a polytope too. Now, according to Propostion~2.4 of \cite{Ziegler-1995} the vertices of $\mathcal{P}'$ are given by
\begin{equation}
    F\cap \{{\bf x} \mid {\bf c \cdot x}=\alpha\},
\end{equation}
where $F$ is the set of edges of $\mathcal{P}$ containing $V_1$. Let $V_{n1}$, $n \in \{1, \dots, N\}$ be a vertex vector at $V_1$ and let
\begin{equation}
    \{w_n\} = \mathcal{P}' \cap \{r V_{n1}+V_1 \mid r\in[0,1]\}.
\end{equation}
The last equality is well-defined since on the one hand the set on the right hand side is non-empty because the function $f(r)= {\bf c} \cdot (r V_{n1}+V_1), r \in [0,1]$ is continuous with $f(0)=c_1 > \alpha > c_n := f(1)$ and so by the intermediate value theorem there exists a $r_n \in [0,1]$ with $f(r_n)=\alpha$. On the other hand $f'(r)=c_n-c_1 <0$, implying by the mean value theorem that $f$ is strictly monotonic and that $r_n$ is the only $ r \in [0,1]$ for which $f(r) = \alpha$. So $\mathcal{P}' \cap \{r V_{n1}+V_1 \mid r\in[0,1]\}$ is indeed a singleton which confirms the geometric intuition that every vertex vector at $V$ crosses with the plane $\{ {\bf x} \mid { \bf c \cdot x} =\alpha\}$ at exactly one point. As a very useful byproduct we also have
\begin{equation}
    w_n=r_n V_{n1} +V_1,
\end{equation}
meaning
\begin{equation}
    w_n -V_1 = r_n V_{n1}.
\end{equation}
Also as $V_1 \notin \mathcal{P}'$, $w_n \neq V_1$ and $r_n > 0$. Notice in particular that we can now explicitly write the set of vertices of $\mathcal{P}'$ as $F \cap \{ {\bf x} \mid {\bf c \cdot x} =c_1\} = \{w_{i_1} ,\dots w_{i_k}\}$. Let us now pick $l$ such that, as in the statement, the vertex vector $V_{l1}$ does not generate an edge at $V_1$. Since $w_l \in \mathcal{P}'$ there are positive $\beta_1, \dots, \beta_k$, $\sum_{j=1}^k \beta_j=1$, with
\begin{equation}
    w_l=\sum_{j=1}^k \beta_j w_{i_j}.
\end{equation}
Putting things together we get
\begin{align}
   V_{l1} &= \frac{1}{r_l} (w_l-V_1)= \frac{1}{r_l} \sum_{j=1}^k \beta_j (w_{i_j}-V_1)=  \sum_{j=1}^k \underbrace{\frac{\beta_j}{r_l} r_{i_j}}_{=:c_{lj}\geq 0}V_{i_j1},
\end{align}
as desired.
\end{proof}
With this we get to the desired result as the following shows.
\begin{align}
    {\bf p}-V_1 &= \sum_{i=2}^N \lambda_i V_{i1}\\
    &=\sum_{i=i_1,\dots,i_k} \lambda_i V_{i1} + \sum_{l \in \{2,\dots,N\} \setminus \{i_1, \dots, i_k\}} \lambda_l V_{l1}\\
    &= \sum_{i=i_1,\dots,i_k} \lambda_i V_{i1} + \sum_{l \in \{2,\dots,N\} \setminus \{i_1, \dots, i_k\}} \lambda_l \left(\sum_{j=1}^k c_{lj} V_{i_j1} \right)\\
    &= \sum_{j=1}^k \lambda_{i_j} V_{i_j1} + \sum_{j=1}^k \left( \sum_{l \in \{2,\dots,N\} \setminus \{i_1, \dots, i_k\}} \lambda_l  c_{lj} \right) V_{i_j1} \\
    &= \sum_{j=1}^k \underbrace{\left( \lambda_{i_j} +  \sum_{l \in \{2,\dots,N\} \setminus \{i_1, \dots, i_k\}} \lambda_l  c_{lj} \right)}_{\geq 0}  V_{i_j1},
\end{align}
as desired.
\end{proof}

\subsection{Vertex vectors of edges are extremal (Extreme rays of $\mathcal{K}_V$)} \label{app:extremerays}

Our goal in this Appendix is to prove Lemma~\ref{lemma:extremerays}. We will actually prove something a bit stronger than that, namely that the vertex vectors of edges at $V$ of $\mathcal{P}$ are the extreme rays of $\mathcal{K}_V$. From this Lemma~\ref{lemma:extremerays} will follow. We first need to set the stage, which we will start doing by showing that shifting $\mathcal{P}$ by $-V$, where $V$ is a vertex of $\mathcal{P}$, again delivers a polytope.
We next want to construct a vertex figure of $\mathcal{P}_1=\mathcal{P}-V_1$ at $0$ as done in the proof of Lemma~\ref{lemma:edgereducing} in Appendix~\ref{app:edgesminimalset}, see also for example \cite[Chapter~2.1]{Ziegler-1995}. To this end let ${\bf c} \in \mathds{R}^d, c_1 \in \mathds{R}$ and let $ {\bf c \cdot x} \leq c_1$ be a valid inequality for $\mathcal{P}_1$ with
\begin{equation}
    \mathcal{P}_1 \cap \{{\bf x} \mid {\bf c \cdot x}=c_1\}= \{0\}.
\end{equation}
Note that since ${\bf 0} \in \{ {\bf x} \mid {\bf c \cdot x}=c_1\}$, $c_1={\bf c \cdot 0}=0$. Let also $\alpha \in \mathds{R}$ with $\alpha < c_1 =0$ and ${\bf c \cdot} (V_i-V_1) < \alpha$ for all $i=2, \dots, N$. We are now interested in
\begin{equation}
    \mathcal{P}'=\mathcal{P}_1 \cap \{{\bf x} \mid {\bf c \cdot x} = \alpha\},
\end{equation}
the vertex figure of $\mathcal{P}_1$ at ${\bf 0}$. According to Proposition 2.4 of \cite{Ziegler-1995}, the vertices of $\mathcal{P}'$ are given by
\begin{equation}
    F \cap \{{\bf x} \mid {\bf c \cdot x} = \alpha\},
\end{equation}
where $F$ is the set of edges of $\mathcal{P}_1$ containing ${\bf 0}$. In fact, again according to Proposition 2.4 of \cite{Ziegler-1995} there is a bijection between the vertices of $\mathcal{P}'$ and the edges of $\mathcal{P}_1$. Denoting the vertices of $\mathcal{P}'$ by $w_1, \dots, w_l$, this means that to each vertex $w_i$ of $\mathcal{P}'$ there is a unique vertex vector of an edge of $\mathcal{P}_1$ at ${\bf 0}$, $V_i-V_1$, and $0<r_i=\frac{\alpha}{{\bf c \cdot} (V_i-V_1)}<1$ such that 
\begin{equation}
    w_i=r_i (V_i-V_1).
\end{equation}
Now let
\begin{equation}
    \mathcal{K}_{V}=\left\{ \sum_{i=2}^N \alpha_i (V_i-V_1) \mid \alpha_i \geq 0, \forall i \in \{2, \dots,N\} \right\}.
\end{equation}
From the proof of Lemma~\ref{lemma:edgesminimal} we know that it suffices to consider the $i \in \{2, \dots, N\}$ for which $V_i-V_1$ is a vertex vector of an edge of  $\mathcal{P}_1$ at ${\bf 0}$. One also readily sees that $\mathcal{K}_V$ is a cone since ${\bf 0} \in \mathcal{K}_V$ and for any $u \in \mathcal{K}_V$, $\lambda u \in \mathcal{K}_V$ for any $\lambda > 0$. We next would like to use Lemma 8.4 on page 66 of~\cite{Barvinok-2002}. To this end we need to show that $\mathcal{P}'$ is a convex base of $\mathcal{K}_V$. This is the content of the following
\begin{lemma}
    $\mathcal{P}'$ is a convex base of $\mathcal{K}_V$.
\end{lemma}
\begin{proof}
    We need to show that
    \begin{enumerate}
        \item $\mathcal{P}' \subset \mathcal{K}_V$,
        \item ${\bf 0} \notin \mathcal{P}'$,
        \item $\forall {\bf u} \in \mathcal{K}_V, \exists! \lambda > 0, {\bf v} \in \mathcal{P}' \text{ s.t. } {\bf u} = \lambda {\bf v}$,
        \item $\mathcal{P}'$ is convex.
    \end{enumerate}

    1.,2., and 3. together show that $\mathcal{P}'$ is a base of $\mathcal{K}_V$ according to Definition 8.3 of~\cite{Barvinok-2002}. The addition of 4. makes $\mathcal{P}'$ a convex base of $\mathcal{K}_V$. 4. is in fact immediate since from the proof of Lemma~\ref{lemma:edgereducing}, we know that $\mathcal{P}'$ is a polytope. And polytopes are by construction convex. To show 1. we recall that from Lemma~\ref{lemma:edgesminimal}, $\mathcal{P}_1 \subset \mathcal{K}_V$, and so
    \begin{align}
        \mathcal{P}' &= \mathcal{P}_1 \cap \{{\bf x} \mid {\bf c \cdot x} = \alpha\} \subset \mathcal{K}_V \cap \{{\bf x} \mid {\bf c \cdot x} = \alpha\} \subset \mathcal{K}_V.
    \end{align}
    Showing 2. is also straightforward since as $\alpha < c_1=0$ we immediately have ${\bf 0} \notin \{ {\bf x} \mid {\bf c \cdot x} = \alpha\}$ from which follows ${\bf 0} \notin  \mathcal{P}' =\mathcal{P}_1 \cap \{{\bf x} \mid {\bf c \cdot x}=\alpha\}$. We have left to show 3. To do so, let ${\bf u} \in \mathcal{K}_V$, ${\bf u} \neq 0$. We need to find a ${\bf v} \in \mathcal{P}'$ and a $\lambda >0$ such that ${\bf u} = \lambda {\bf v}$. And then show that those are in fact unique. As ${\bf u} \in \mathcal{K}_V$, there are $\alpha_i \geq 0, i= 2, \dots,N $ such that
    \begin{equation}
        {\bf u} = \sum_{i=2}^N \alpha_i (V_i-V_1).
    \end{equation}
    Now ${\bf c \cdot u} < 0$, for  ${\bf c} \cdot (V_i-V_1) < \alpha <0$ for any $i \in \{2,\dots, N\}$ and ${\bf u \neq 0}$ implies $\alpha_i > 0$ for at least one $i \in \{2,\dots, N\}$. So let $\beta = {\bf c \cdot u} <0$ and let ${\bf b} = \frac{\alpha}{\beta} {\bf u}$. ${\bf b}$ is our candidate vor ${\bf v}$ and $\frac{\beta}{\alpha} >0$ our $\lambda$. We already have ${\bf u} = \frac{\beta}{\alpha} {\bf b}$. We have left to show that ${\bf b} \in \mathcal{P}'$. To do so, first note that 
    \begin{equation}
        {\bf c \cdot b} = \frac{\alpha}{\beta} {\bf c \cdot u} = \alpha.
    \end{equation}
    This already shows ${\bf b} \in \{{\bf x} \mid {\bf c \cdot x}=\alpha\}$. To show that ${\bf b} \in \mathcal{P}_1$ notice that
    \begin{equation}
        {\bf b} = \sum_{i=2}^N \alpha_i \frac{\alpha}{\beta} (V_i-V_1).
    \end{equation}
    Also
    \begin{align}
        \beta = {\bf c \cdot u}= {\bf c} \cdot \sum_{i=2}^N \alpha_i (V_i-V_1)= \sum_{i=2}^N \alpha_i \underbrace{{\bf c}\cdot (V_i-V_1)}_{< \alpha}< \alpha \sum_{i=2}^N \alpha_i.
    \end{align}
    So as $\alpha <0$, we get
    \begin{equation}
        \sum_{i=2}^N \alpha_i < \frac{\beta}{\alpha}.
    \end{equation}
    And since $\frac{\beta}{\alpha} > 0$, we have
    \begin{equation}
        0< \sum_{i=2}^N \alpha_i \frac{\alpha}{\beta} < 1.
    \end{equation}
    So with $1>\gamma = 1-\sum_{i=2}^N \alpha_i \frac{\alpha}{\beta}>0$ we get that $\gamma + \sum_{i=2} \alpha_1 \frac{\alpha}{\beta} =1$ and
    \begin{equation}
    {\bf b}= \gamma (V_1-V_1)+ \sum_{i=2}^N \alpha_i \frac{\alpha}{\beta} (V_i-V_1),
    \end{equation}
    which, as desired, shows that 
    \begin{equation}
        {\bf b} \in \text{conv}\{0, V_2-V_1, \dots, V_N-V_1\}=\mathcal{P}_1.
    \end{equation}
    With the above we therefore have ${\bf b} \in \mathcal{P}'= \mathcal{P}_1 \cap \{{\bf x} \mid {\bf c \cdot x }=\alpha\}$. So we found a ${\bf v} \in \mathcal{P}'$ and a $\lambda >0$ such that ${\bf u}= \lambda {\bf v}$. We have left to show that they are unique. For this suppose that there is a ${\bf \tilde{b}} \in \mathcal{P}'$ and some $\tilde{\lambda} >0$ with ${\bf u}= \tilde{\lambda} {\bf \tilde{b}}$. Then as ${\bf \tilde{b}} \in 
    \{{\bf x} \mid {\bf c \cdot x} = \alpha\}$ we have
    \begin{align}
        \alpha = {\bf c \cdot \tilde{b}}= {\bf c} \cdot \frac{{\bf u}}{\tilde{\lambda}}=\frac{1}{\tilde{\lambda}} {\bf c \cdot u}= \frac{\beta}{\tilde{\lambda}}.
    \end{align}
    So $\tilde{\lambda}=\frac{\beta}{\alpha}$. And so we have
    \begin{equation}
        \frac{\beta}{\alpha} {\bf b}= {\bf u}= \tilde{\lambda} {\bf \tilde{b}}=\frac{\beta}{\alpha} {\bf \tilde{b}},
    \end{equation}
    from which follows ${\bf \tilde{b}=b}$, as desired.
\end{proof}
Now, from Lemma 8.4 of~\cite{Barvinok-2002} every vertex of $\mathcal{P}'$ spans an extreme ray of $\mathcal{K}_V$. so given a vertex vector of $\mathcal{P}_1$ at ${\bf 0}$, say $V_i-V_1$, we have that 
\begin{equation}
    w_i = r_i (V_i-V_1)
\end{equation}
with $r_i= \frac{\alpha}{{\mathbf{ c} \cdot ({V}_i-{V}_1)}}$ is a vertex of $\mathcal{P}'$ and so
\begin{align}
    \text{co}(w_i) = \{\lambda w_i \mid \lambda \geq 0\}= \{ \mu (V_i-V_1) \mid \mu \geq 0\}= \text{co}(V_i-V_1)
\end{align}
is an extreme ray of $\mathcal{K}_V$. So all vertex vectors of $\mathcal{P}_1$ at ${\bf 0}$ generate/span an extreme ray of $\mathcal{K}_V$. And in fact, again according to Lemma 8.4 of~\cite{Barvinok-2002}, there are all the extreme rays of $\mathcal{K}_V$ for if a point ${\bf u} \in \mathcal{K}_V$ is on an extreme ray of $\mathcal{K}_V$, then it must hold that ${\bf u} \in \text{co}(w_i)$, where $w_i$ is a vertex of $\mathcal{P}'$. So we have just proven the following
\begin{lemma}
    The vertex vectors of edges at $V_1$ of $\mathcal{P}$ are the extreme rays of $\mathcal{K}_V$.
\end{lemma}
From this we now want to show that Lemma~\ref{lemma:extremerays} follows.
\begin{proof}[Proof of Lemma~\ref{lemma:extremerays}]
We want to show that a given vertex vector of an edge at $V_1$ of $\mathcal{P}$, say $V_{i_1}-V_1$, cannot be written as a positive linear combination of the other vertex vectors of edges at $V_1$ of $\mathcal{P}_1$. Assume to the contrary that this is not the case. Then denoting the vertex vector of edges at $V_1$ of $\mathcal{P}_1$ by $V_{i_1}-V_1, \dots, V_{i_k}-V_1$, $i_1, \dots, i_k \in \{2, \dots, N\}$, our assumption implies there are $\lambda_2, \dots, \lambda_k \geq 0$ such that
\begin{align}
V_{i_1}-V_1= \sum_{l=2}^k \lambda_l (V_{i_l}-V_1)= \lambda_2 (V_{i_2}-V_1) + \sum_{l=3}^k \lambda_l (V_{i_l}-V_1).
\end{align}
So with $y= \lambda_2 (V_{i_2}-V_1) \in \text{co}(V_{i_2}-V_1) \subset \mathcal{K}_V$ and $z=\sum_{l=3}^k \lambda_l (V_{i_l}-V_1) \in \text{co}(V_{i_3}-V_1, \dots, V_{i_k}-V_1) \subset \mathcal{K}_V$ we get
\begin{equation}
    V_{i_1}-V_1= y+z.
\end{equation}
So since $V_{i_1}-V_1$ lies on an extreme ray of $\mathcal{K}_V$, by definition of an extreme ray,
\begin{equation}
    y,z \in \text{co}(V_{i_1}-V_1).
\end{equation}
So there exists a $\mu \geq 0$ such that 
\begin{equation}
    V_{i_1}-V_1= \mu y = \mu \lambda_2 (V_{i_2}-V_1).
\end{equation}
We have $\mu \lambda_2 \geq 0$. Now, if $\mu \lambda_2 =0$, $V_{i_1}-V_1=0$, which is a contradiction to $V_{i_1}$ and $V_1$ being distinct. If $\mu \lambda_2=1$ we get $V_{i_1}=V_{i_2}$ again a contradiction to them being distinct. If $0 < \mu \lambda_2 <1$ we get
\begin{equation}
    V_{i_1}-V_1= \mu \lambda_2 (V_{i_2}-V_1)+(1-\mu \lambda_2) \times 0,
\end{equation}
meaning
\begin{equation}
    \{V_{i_1}-V_1\} \cap ]V_{i_2}-V_1, 0[ \,\neq \emptyset.
\end{equation}
And so, as $V_{i_1}-V_1$ is an extreme point of $\mathcal{P}_1$, this implies $V_{i_2}-V_1,0 \in \{V_{i_1}-V_1\}$, a contradiction to $V_{i_1} \neq V_1$. So $\mu \lambda_2 > 1$ must hold. But then $0 < \frac{1}{\mu \lambda_2} < 1$ and
\begin{equation}
    V_{i_2}-V_1= \frac{1}{\mu \lambda_2} (V_{i_1}-V_1)+ (1-\frac{1}{\mu \lambda_2}) \times 0,
\end{equation}
meaning
\begin{equation}
    \{V_{i_2}-V_1\} \,\cap \, ]V_{i_1}-V_1,0[ \, \neq \emptyset.
\end{equation}
And so, as $V_{i_2}-V_1$ is an extreme point of $\mathcal{P}_1$,
\begin{equation}
    V_{i_1}-V_1,0 \in \{V_{i_2}-V_1\}, 
\end{equation}
a contradiction to $V_{i_2} \neq V_1$.
\end{proof}

\section{Proofs regarding population polytopes}

In this section we derive discuss and prove relevant properties of a population polytope, as it is the central mathematical object of our work, on which we construct the optimal trajectory.

Our approach is and leans upon that of the text by Barvinok\cite{Barvinok-2002} on convexity. As such, we employ the notation from Barvinok rather than that of the rest of the paper to aid the reader who can refer to \cite{Barvinok-2002} for foundational concepts related to polytopes in general.

\subsection{Group theory of a population polytope}\label{app:grouptheorypolytope}

In this section we describe for the interested reader the relationship between the permutation group and the vertices of a population polytope. Given a vector $a = (\alpha_1,...,\alpha_n) \in \mathbb{R}^n$, the population polytope $\mathcal{P}$ is the set of all points in $\mathbb{R}^n$ of the form $Da$ where $D$ is a doubly stochastic matrix. Via the Birkhoff von-Neumann theorem, we know that the vertices of $\mathcal{P}$ are permutations $P$ of $a$. The main lemma here is a one-to-one correspondence between the vertices and a group formed out of the permutation group.

To start off we note a powerful symmetry of $\mathcal{P}$.

\begin{lemma}\label{lem:polytopeinvariance}
$\mathcal{P}$ is invariant under permutations.
\end{lemma}
\begin{proof}
The geometric version of this proof is based on symmetry: since we run over all doubly-stochastic matrices, the polytope $\mathcal{P} \in \mathbb{R}^n$ is symmetric w.r.t. the axes of $\mathbb{R}^n$. A single permutation $P$ can be expressed as a relabelling of the axes, thus $P \mathcal{P} = \mathcal{P}$.

In algebraic terms: the set of doubly stochastic matrices in its entirety forms a group, thus $PD$ is also a doubly-stochastic matrix for $P,\,D$ being a permutation and a doubly stochastic matrix respectively. Consider the original polytope $\mathcal{P} = \{ x: \; x = Da\}$ and the permuted set $\mathcal{P}^\prime = \{ x: \; x = P_0Da \}$ for some fixed permutation $P_0$ and arbitrary doubly-stochastic $D$. Since $P_0 D = D^\prime$ is also doubly-stochastic, it follows that if $x = P_0 D a \in \mathcal{P}^\prime$, then $x = D^\prime a \in \mathcal{P}$ as well. Conversely, expressing an arbitrary $D = P_0 (P_0^{-1} D)$, we see that if $x = Da \in \mathcal{P}$, then $x = P_0 (P_0^{-1} D) a = P_0 D^{\prime\prime} a \in \mathcal{P}^\prime$ as well.
\end{proof}

Lemma \ref{lem:polytopeinvariance} is rather powerful, it implies in a geometric sense that the polytope looks the same under permutation, which gives it a very high degree of symmetry. The following lemma is a direct consequence.

\begin{lemma}\label{lem:permutationisvertex}
Every permutation $Pa$ is a vertex of $\mathcal{P}$.
\end{lemma}
\begin{proof}
Since $\mathcal{P}$ is invariant under permutations, then if $P_0 a$ is a vertex of $\mathcal{P}$, then for any permutation $P$, the point $P P_0 a$ must also be a vertex. But the set of $\{P P_0\}$ is the entire group of permutations, thus every permutation of $a$ is a vertex of $\mathcal{P}$.

\end{proof}

There are further implications that we need not go into, for instance that 1) the \textit{vertex figure} of any vertex of $\mathcal{P}$, i.e. the polyhedron formed by cutting the polytope close enough to the vertex so that it isolates it, is the same for every vertex, in turn implying that the number, size and structure of edges, faces, and higher dimensional facets containing each vertex are exactly the same as those at another.

The simplest case for the polytope vertices is when $a$ has distinct elements, i.e. $\alpha_i \neq \alpha_j$ for distinct $i$ and $j$. In this case every $Pa$ is distinct, and thus the polytope has $n!$ vertices, each corresponding to a single permutation $P$.

This brings us to the major part of this section: the structure of $\mathcal{P}$ when $a$ has degeneracies:
\begin{lemma}\label{lem:verticesfactorgroup}
Let $a = (\alpha_1,...,\alpha_n)$ be a list of $n$ real elements and $\mathcal{P}$ be the permutation polytope of $a$. Let $\mathcal{G}$ denote the group of permutations of a $n$-element list, and $\mathcal{S}$ denote the set of permutations that leave $a$ unchanged. Then $\mathcal{S}$ is a subgroup of $\mathcal{G}$ and the vertices of $\mathcal{P}$ are in one-to-one correspondence to the elements of the factor group $\mathcal{G} / \mathcal{S}$, formed by the left cosets of $\mathcal{S}$ w.r.t. elements of $\mathcal{G}$.
\end{lemma}
\begin{proof}
Let us label the distinct values that the elements of $a$ can take as $x_1 > x_2 ... > x_k$, where $k \leq n$. This allows us to partition the set of positions $\{1,...,n\}$ into $k$ pairwise disjoint non-empty subsets $S_1,...,S_k$, where each $S_i$ contains those positions of $a$ with the value $x_i$:
\begin{align}
    S_i = \{j: \; \alpha_j = x_i\}.
\end{align}
Let the cardinality of $S_i$ be denoted as $s_i$. As an example, if $a = (x_2,x_3,x_1,x_1,x_3)$ then $S_1 = \{3,4\}$, $S_2 = \{1\}$ and $S_3 = \{2,5\}$.

The set $\mathcal{S}$ of permutations that leave $a$ unchanged are exactly those composed of pairwise disjoint permutations, each acting on a single $S_i$:
\begin{align}
    \mathcal{S} = \{ P: \; P = P_1P_2...P_k \},
\end{align}
where $P_i$ acts only upon the positions within $S_i$. $\mathcal{S}$ is a subgroup because:
\begin{enumerate}
    \item If $P = P_1...P_k \in \mathcal{S}$ and $P^\prime = P_1^\prime...P_k^\prime \in \mathcal{S}$, then $P P^\prime = P_1 P_1^\prime...P_k P_k^\prime = P_1^{\prime\prime}...P_k^{\prime\prime}$ is also in $\mathcal{S}$,
    \item The identity operation is in $\mathcal{S}$ as it leaves a unchanged,
    \item If $P = P_1...P_k \in \mathcal{S}$ then $P^{-1} = P_1^{-1}...P_k^{-1}$ is also in $\mathcal{S}$.
\end{enumerate}
The cardinality $|\mathcal{S}|$ of $\mathcal{S}$ is the product $\prod_{i=1}^k s_i$ of the sizes of each of the degenerate sub-parts of $a$.

A left coset of $\mathcal{S}$ in $\mathcal{G}$  is the set of permutations $P Q$ where $Q$ varies over the elements of $\mathcal{S}$, labelled as $P\mathcal{S}$; the size of the coset is $|\mathcal{S}|$. As $Pa = P P_0 a$ for any $P_0 \in \mathcal{S}$, the entire left coset $P\mathcal{S}$ corresponds to a single operation on $a$. The set of distinct left cosets of $\mathcal{S}$ thus corresponds to an \textit{effective} set of distinct permutations of $a$, indeed it forms a group: Defining the product rule $P_1 \mathcal{S} \cdot P_2 \mathcal{S} = (P_1P_2) \mathcal{S}$, this is a group of cardinality $|\mathcal{G}|/|\mathcal{S}|$.

We can now complete the proof regarding the vertices. As every permutation of $a$ is a vertex of $\mathcal{P}$, the number of vertices is the number of distinct permutations of $a$ given its degeneracies, this is a standard quantity in combinatorial theory,
\begin{align}
    N &= \frac{n!}{s_1!...s_k!} = \frac{|\mathcal{G}|}{|\mathcal{S}|},
\end{align}
which is precisely the number of elements in the factor group. As each vertex must correspond to at least one of the cosets, and no two vertices can correspond to the same coset, it must be that the vertices are in one-to-one correspondence with the cosets.

An alternate way of stating the lemma is thus: given a list $a$ with degeneracies, such that $\mathcal{S}$ is the subgroup of permutations under which it is invariant, one can study its permutation polytope by using the \textit{effective permutation group} $\mathcal{G} / \mathcal{S}$ rather than the whole permutation group $\mathcal{G}$.
\end{proof}

\subsection{The facial structure of the population polytope}\label{app:barvinokextension}

In Barvinok\cite{Barvinok-2002}, Chapter VI Proposition 2.2, the complete description of the faces of the population polytope is stated and proven, albeit for a population vector with distinct values. In this section we extend the statement and proof to the case when some populations may be equal.

To aid the reader that may wish to refer to the proof in \cite{Barvinok-2002} prior to the following one, and since our proof uses the same technique, we follow the notation in \cite{Barvinok-2002} for the rest of this section rather than that of the rest of the paper.

Before the main theorem, we prove a lemma on cyclic permutations that we will find useful.

\begin{lemma}\label{lem:disjointcycle}

Every cyclic permutation (or cycle for short) of a list of elements can be expressed as a product of pairwise disjoint cycles on the same list such that each cycle has distinct elements.

\end{lemma}

\begin{proof}

We prove the statement by construction. Let the list have $n$ elements and be denoted by $a = \left(\alpha_1,...,\alpha_n\right)$. Express the cycle as the list of positions $(i_1,...,i_n)$ where $i_1 \to i_n$, $i_2 \to i_1$ and so on. If the list has distinct elements, then we already satisfy the statement of the lemma. Consider that there is at least one repeated element. We are free to pick the starting position of the cycle without loss of generality, so we pick it to be one of the degenerate elements. Thus $\alpha_{i_1} = \alpha_{i_r}$ for some $r$. Then the action of the original cycle is equivalent to the following product of two cycles $(i_1,...,i_{r-1}) * (i_r,...,i_n)$.

The above split decreases the total number of repeated values within cycles by at least $1$, as $\alpha_{i_1}$ and $\alpha_{i_r}$ are now in different cycles. If there are still repetitions one can repeat the above until every cycle has distinct elements.  

\end{proof}

\begin{theorem}\label{thm:Barvinokextension}
Let $a = \left( \alpha_1, ... , \alpha_n \right)$ be a point such that $\alpha_1 \geq \alpha_2 \geq ... \geq \alpha_n$ and let $P = P(a)$ be the corresponding permutation polytope. For a number $k \leq n$, let $\mathcal{S}$ be a partition of the set $\{1,...,n\}$ into k pairwise disjoint non-empty subsets $S_1,...,S_k$. Let $s_i = |S_i|$ be the cardinality of the $i$-th subset for $i=1,...,k$ and let $t_i = \sum_{j=1}^i s_j$ for $i=1,...,k$. 

Define $A_1,...,A_k$ as follows: the $A_i$ form a partition $\mathcal{A}$ of the set $\{\alpha_1,...,\alpha_n\}$ into $k$ pairwise disjoint subsets such that
\begin{enumerate}
    \item $|A_i| = s_i$,
    \item if $\alpha_r \in A_i$, $\alpha_s \in A_j$ and $i < j$, then $\alpha_r \geq \alpha_s$.
\end{enumerate}
Equivalently: $A_1$ is comprised of $s_1$ number of $\alpha$'s with the largest values, $A_2$ with $s_2$ number of remaining $\alpha$'s with the largest remaining values, and so on.

Let $F_\mathcal{S,A}$ be the convex hull of all of the points $b = \sigma(a)$, $b = \left( \beta_1,...,\beta_n \right)$ such that $\{\beta_j:j \in S_i\} = A_i$ for all $i = 1,...,k$. In words: we permute the elements of $A_1$ in the coordinate positions prescribed by $S_1 \subset \{1,...,n\}$, the elements of $A_2$ in the coordinate positions prescribed by $S_2$ and so forth, and take the convex full $F_\mathcal{S,A}$ of all resulting points.

Then $F_\mathcal{S,A}$ is a face of $P$, and for every face $F$ of $P$ we have $F = F_\mathcal{S,A}$ for some choice of partitions $\mathcal{S}$ and $\mathcal{A}$.

\bigskip

Furthermore, we label an $A_i$ as ``non-degenerate'' if there are at least two distinct values within it, i.e. a pair $\alpha_r, \alpha_s$ s.t. $\alpha_r \neq \alpha_s$. Let $k^\prime$ be the number of such non-degenerate subsets and $n^\prime$ be the total number of elements within non-degenerate subsets. Then $ \text{dim} \; F_{\mathcal{S},\mathcal{A}} = n^\prime - k^\prime$.

\end{theorem}

\begin{remark}

The statement of the theorem follows \cite{Barvinok-2002} closely, with the salient difference being the partition $\mathcal{A}$. Note that if $a$ is composed of distinct values then $\mathcal{A}$ simplifies to a sequential partition as in \cite{Barvinok-2002}, i.e. $A_1 = \{\alpha_j \; : \; 1\leq j \leq s_1\}$, and $A_i = \{\alpha_j \; : \; t_{i-1} < j \leq t_i\}$ for $i = 2,...,k$. Thus $\mathcal{A}$ is completely determined by the partition $\mathcal{S}$.

For us, the degeneracies within $a$ provide a freedom in choosing $\mathcal{A}$, which is why the resulting face is fixed by the pair $\{\mathcal{S},\mathcal{A}\}$ and is denoted $F_{\mathcal{S},\mathcal{A}}$ rather than simply $F_\mathcal{S}$ as in \cite{Barvinok-2002}.

The other difference is in the dimension of the face. In the original, every subset $A_i$ of cardinality $s_i \geq 2$ typically contributes $s_i - 1$ to the dimension, the sum of which is $n-k$. However if $A_i$ is entirely degenerate, i.e. has only one repeated value, then it contributes $0$ to the dimension of the face, thus the sum is taken over only non-degenerate $A_i$, leading to $n^\prime - k^\prime$.

\end{remark}

\begin{proof}

We begin by describing the faces $F$ of $P$ that contain $a$. Take any such face $F$. Then there exists $c = (\gamma_1,...,\gamma_n)$ and $\lambda$ such that $c \cdot x \leq \lambda$ for all $x \in P$ and $c \cdot x = \lambda$ if and only if $x \in F$. Since $a \in F$, $c \cdot a = \lambda$. Let the number of distinct values of the $\gamma_i$ be $k$, we label them in decreasing order as $\{x_1,...,x_k\}$. Then we form partitions $\mathcal{S}$ of $\{1,...,n\}$ and $\mathcal{A}$ of $\{\alpha_1,...,\alpha_n\}$ into $k$ pairwise disjoint non-empty sets corresponding to each $x_i$:
\begin{align}\label{eq:indexsets}
    S_i &= \left\{ j \; : \; \gamma_j = x_i \right\}, \\
    A_i &= \left\{ \alpha_j \; : \; \gamma_j = x_i \right\}.
\end{align}

First we prove that $\mathcal{A}$ has the property stated in the theorem, i.e if $\alpha_r \in A_i,\, \alpha_s \in A_j$ and $i<j$ then $\alpha_r \geq \alpha_s$. Consider that is not the case, i.e. there exists $i,\,j,\,r, \text{and}\, s$ such that $\alpha_r \in A_i,\, \alpha_s \in A_j$,\, \text{and}\, $i<j$ but $\alpha_r < \alpha_s$. Since $i<j$ one has $\gamma_i > \gamma_j$, and now the permutation $\sigma(a)$ of $a$ which swaps $\alpha_r$ and $\alpha_s$ violates the inquality of the face,
\begin{align}
    c. \sigma(a) &= c \cdot a + \gamma_j \alpha_r + \gamma_i \alpha_s - \gamma_i \alpha_r - \gamma_j \alpha_s = \lambda + \left( \gamma_i - \gamma_j \right) \left( \alpha_s - \alpha_r \right) > \lambda,
\end{align}
which is impossible.

\bigskip

Next, we determine which permutations $\sigma(a)$ belong to $F$, the reason being that any face of a polytope is the convex hull of the vertices of the polytope that are in the face, and the vertices of $P(a)$ are permutations of $a$.

Any permutation $\sigma$ of $n$ elements can always be expressed in a unique manner as a product of disjoint cycles (up to the ordering between cycles, which is irrelevant for this proof). Furthermore, from Lemma \ref{lem:disjointcycle} this can be further reduced to a product of disjoint cycles in such a way that each cycle contains distinct elements. We label each of the non-degenerate disjoint cycles in $\sigma(a)$ as $\phi_i(a)$, and the list of elements of $a$ that are within the $i^{th}$ cycle as $a_i$.
Note that each cycle is a valid permutation on its own. We  express the difference $c \cdot \sigma(a) - c\cdot a$ w.r.t. the cycles, the sum must be zero as $\sigma(a) \in F$:
\begin{align}
    0 = c \cdot \sigma(a) - c \cdot a &= \sum_{i=1}^N c \cdot \phi_i(a) - c \cdot a.
\end{align}
However the difference due to any one cycle must be non-positive, if not, then $c \cdot \phi_i(a) > \lambda$, which is impossible as $\phi_i(a) \in P$. Therefore they are all $0$. But since each cycle involves distinct $\alpha$'s and the $\gamma$'s are ordered non-increasing w.r.t. decreasing $\alpha$'s, the value of $c \cdot \phi_i(a)$ can only be smaller than or equal to $c \cdot a$, and only equal in the case that all $\gamma$'s in the $i^{th}$ cycle are equal.

We conclude that $\sigma(a) \in F$ only if it can be expressed as a product of disjoint cycles, each cycle involving distinct $\alpha$'s, and for which the values of $\gamma_i$ corresponding to the elements of a single cycle are all equal. Let us label the set of $\sigma(a)$ that are in $F$ as $B$ and the set of $\sigma(a)$ that can be expressed as the product of cycles just described above as $C$. Then we have proven that $B \subset C$.

Consider next a set $D$ that is larger than $C$: the set of all $\sigma(a)$ that are products of pairwise disjoint permutations, each acting on a subset of $a$ within which all of the corresponding $\gamma$-values are equal, i.e. each one of the $A_i$. This set is larger than $C$ as it allows arbitrary permutations rather than only cycles, and on larger sets (the values of $\gamma_i$ have to be equal, but the $\alpha$'s do not have to be distinct).

It is straightforward to show that any such permutation is also in $F$: as the elements of $a$ are only moved around within subsets where $\gamma$ does not change, the scalar product is preserved: $c \cdot \sigma(a) = c \cdot a$. Thus $D \subset B$.

But now we have $B \subset C \subset D \subset B$, implying that $B = C = D$. Thus the face $F$ is $F_{\mathcal{S},\mathcal{A}}$, where the partitions $\mathcal{S},\mathcal{A}$ are derived from the $\gamma$ as in \eqref{eq:indexsets}.

Vice versa, every vector $c = \left( \gamma_1,...,\gamma_n \right)$ that is ordered in a way to lead to a partition $\mathcal{A}$ with the stated property gives rise to a face $F_{\mathcal{S},\mathcal{A}}$ where $\mathcal{S} = S_1 \cup S_2 \cup ... \cup S_k$ is the partition of $\{1,...,n\}$ into subintervals of equal $\gamma$'s.

The remainder of the proof goes as in \cite{Barvinok-2002}. We have described the faces $F$ of $P$ containing $a$. Let $\sigma$ be a permutation of $n$ elements, then the action $x \mapsto \sigma(x)$ is an orthogonal transformation of $\mathbb{R}^n$. Hence we conclude that $F$ is a face of $P$ if and only if for some permutation $\sigma$, the set $\sigma(F)$ is a face of $P$ containing $a$. If $\sigma(F) = F_{\mathcal{S},\mathcal{A}}$ for $\mathcal{S} = S_1 \cup S_2 \cup ... \cup S_k$, then $F = F_{\mathcal{S}^\prime,\mathcal{A}}$, where $\mathcal{S}^\prime = \sigma^{-1}(S_1) \cup ... \cup \sigma^{-1} (S_k)$.

\bigskip

We now calculate the dimension of $F_{\mathcal{S},\mathcal{A}}$. Labelling the subset of elements in each $A_i$ as $a_i \in \mathbb{R}^{s_i}$, where $s_i$ is the cardinality of $A_i$, the face $F$ is the direct product (upto ordering of the $A_i$)
\begin{align}
    F = P(a_1) \times ... \times P(a_k)
\end{align}
of the permutation polytopes $P(a_i) \subset \mathbb{R}^{s_i}$, and the dimension of the face is thus the sum $\sum_{i=1}^k \text{dim}\; P(a_i)$. If $A_i$ has only one repeated value, then $\text{dim} \; P(a_i) = 0$, else $\text{dim} \; P(a_i) = s_i - 1$. We therefore sum only over non-degenerate $A_i$ to get
\begin{align}
    \text{dim} \; F_{\mathcal{S},\mathcal{A}} &= \sum_{i=1, A_i \; non-deg}^{k} s_i - 1 = n^\prime - k^\prime,
\end{align}
where $n^\prime$ is the total number of elements in non-degenerate $A_i$ and $k^\prime$ the total number of non-degenerate $A_i$.

\bigskip

\end{proof}


\begin{thmcorollary}[Edges of a population polytope]
Let $a = \left( \alpha_1,...,\alpha_n \right)$ be a point such that $\alpha_1 \geq \alpha_2 ... \geq \alpha_n$ and let $P = P(a)$ be the corresponding population polytope. Then the edges --- faces of dimension $1$ --- of $P(a)$ are formed by exactly permutations of $a$, related to each other by an adjacently-valued swap, i.e. a 2-cycle of two elements $\alpha_i > \alpha_j$ such that there does not exist an element in between --- $ \nexists r \; s.t. \; \alpha_i > \alpha_r > \alpha_j$.

\end{thmcorollary}

\begin{proof}

From Theorem \ref{thm:Barvinokextension}, a face of $P(a)$ is formed from permutations on a partition of disjoint subsets of $a$, and the dimension of said face is $n^\prime - k^\prime$, where $k^\prime$ is the total number of ``non-degenerate'' subsets, i.e. with at least two distinct values, and $n^\prime$ the total number of elements within such subsets.

However, this implies each non-degenerate subset must have at least two elements, which in turn implies that $n^\prime \geq 2 k^\prime$, or $n^\prime - k^\prime \geq k^\prime$. But an edge is a face of dimension $1$, therefore the only possibility is that $k^\prime = 1$ and $n^\prime = 2$; in words: there is only a single non-degenerate subset in the partition of $a$ that defines the face, and it contains only two elements.

There are thus only two permutations of $a$ that appear on such a face, related by the swap of the two elements in the non-degenerate subset.

Label this subset as $A_i = \{\alpha_r,\alpha_s\}$ where w.l.o.g. $\alpha_r > \alpha_s$. We prove that there does not exist $t$ such that $\alpha_r > \alpha_t > \alpha_s$ by contradiction of the defining property of the partition that defines the face, see Theorem \ref{thm:Barvinokextension}. Let $\alpha_t \in A_j$. Then either $i < j$ in which case $\alpha_t > \alpha_s$ violates the property, or $i > j$ in which case $\alpha_r > \alpha_t$ violates it. Thus $\alpha_r$ and $\alpha_s$ must be adjacently-valued.

\end{proof}

\subsection{The facial structure of a direct product of population polytopes}\label{app:facesdirectproduct}

Finally, we include some lemmas that extend the above result to the case of a direct product of permutation polytopes.

\begin{lemma}

The faces of a direct product of polytopes are direct products of the faces of the individual polytopes. More precisely, let $\{P_1,...,P_k\}$ be a set of $k$ polytopes, where each $P_i \in \mathbb{R}^{s_i}$, and let $P = P_1 \times ... \times P_k$ be the direct product, itself a polytope. Then $F$ is a face of $P$ if and only if $F = F_1 \times ... \times F_k$, where each $F_i$ is a face of $P_i$. Furthermore, $\text{dim} \; F = \sum_{i=1}^k \text{dim} \; F_i$.

\end{lemma}

\begin{lemcorollary}

The edges of a direct product of permutation polytopes is given by the direct product of an edge of one of the polytopes with vertices from the rest. More precisely, if $F$ is an edge of $P$ then $F = F_1 \times ... \times F_k$ where all but one of the $F_i$ --- label it $F_j$are vertices of their respective polytopes, and $F_j$ is an edge of $P_j$.

\end{lemcorollary}

\begin{proof}

We prove the lemma and corollary together, beginning with the lemma. Let $n = \sum_{i=1}^k s_i$ denote the dimension of the affine space of $P$, i.e. $P \in \mathbb{R}^n$. If $F$ is a face of $P$, then there exists a vector $c \in \mathbb{R}^n$ and $\lambda \in \mathbb{R}$ such that for all $x \in P$, $c \cdot x \leq \lambda$, with equality if and only if $x \in F$. Express $c = c_1 \oplus ... \oplus c_k$ w.r.t. the same partitions as the polytope, $|c_i| = |s_i|$. Also express a general point $x = x_1 \oplus ... \oplus x_k$ in the same manner. We define
\begin{align}
    \lambda_i = \max_{x_i \in P_i} c_i \cdot x_i.
\end{align}
We proceed to prove that $c\cdot x = \lambda$ only if $c_i \cdot x_i = \lambda_i$ for $i \in \{1,...,k\}$. Consider that this is not the case, i.e. there exists $i$ such that $c_i \cdot x_i = \eta < \lambda_i$. But we know that there exists an $x_i^\prime \in P_i$ such that $c \cdot x_i = \lambda_i$, and if $x = x_1 \oplus ... \oplus x_i \oplus ... \oplus x_k \in P$, then so is the point $x^\prime = x_1 \oplus ... \oplus x_i^\prime ... \oplus x_k$, which differs only in the $i$-th partition. But $c \cdot x^\prime = \lambda + (\lambda_i - \eta) > \lambda$ which is impossible. Thus it must be that $c_i \cdot x_i = \lambda_i$. But this implies that $\lambda = \sum_{i=1}^k \lambda_i$, which in turn trivially implies the converse --- if $c_i \cdot x_i = \lambda_i$ for $i \in \{1,...,k\}$, then $c \cdot x = \lambda$.

Thus given $F$, we can define faces for the individual polytopes $F_i$ defined by the vector inequality $c_i \cdot x_i \leq \lambda$ with equality if and only if $x_i \in F_i$. We have thus proven that $x \in F$ in and only if $x_i \in F_i$, thus $F = F_1 \times ... \times F_k$.

The proof of the converse is analogous: given faces $F_i$ for the individual polytopes, each defined by the pair $c_i, \lambda_i$, we obtain a face $F$ of the direct product from the pair $c_1 \oplus ... \oplus c_k,\, \lambda_1 + ... + \lambda_k$.

The dimension follows from the face that faces are themselves polytopes, so that $F$ is a direct product of polytopes, so that
\begin{align}
    \text{dim} \; F &= \sum_{i=1}^k \text{dim} \; F_i.
\end{align}

Thus follows the corollary: if $F$ is an edge, then $\text{dim} \; F = 1$, implying that in the sum above there must be a single face of dimension $1$ (an edge) and the rest of dimension $0$ (vertices).

\end{proof}

\section{Proof of main result}\label{app:proof}

\subsection{proof of Theorem~\ref{thm:increasingslope}} \label{app:increasingslope}
If the point ${\bf p}$ we picked in $\mathcal{P}$ happens to have a target value that is greater than that of $V$, i.e., if $\mathcal{A}({\bf p}) > \mathcal{A}(V)$, then we first prove our assertion for when $\mathcal{A}(V) \neq \alpha_{\min}$. To do so, take the set of all vertices $Z_i$ such that $Z_i-V$ is a vertex vector of an edge at $V$. By Theorem~\ref{thm:edgesadjacentvalued} we know that those are the set of vertices connected to $V$ by an av-swap. Next split the index set of $\{Z_i\}_i$ into three sets, depending on how the target value changes along the vertex vector under consideration:
\begin{enumerate}
    \item $S_- =\{ i  \text{ for which } \mathcal{A}(Z_i) < \mathcal{A}({V}) \}$,
    \item $S_+= \{ i \text{ for which } \mathcal{A}(Z_i) > \mathcal{A}({V}) \}$,
    \item $S_= = \{ i \text{ for which } \mathcal{A}(Z_i) = \mathcal{A}({V}) \}$.
\end{enumerate}

Note that since $\mathcal{A}(V) \in (\alpha_{\min}, \alpha_{\max})$, $S_-$ and $S_+$ are both non-empty. For any $i \in S_- \cup S_+$ we label by $\gamma_i$ the gradient of the cost vs target along the vertex vector from $V$ to $Z_i$, i.e.,
\begin{align}
    \gamma_i &= \frac{ {\bf E} \cdot (Z_i - V)}{ {\bf a} \cdot (Z_i - V)}.
\end{align}

We label by $\gamma^+_{\min}$ the minimum value of $\gamma_i$ for $i \in \mathcal{S}_+$, i.e.,
$    \gamma_{\min}^+ = \min_{i \in \mathcal{S}_+} \gamma_i$
is the minimum gradient along vertex vectors that increase the value of the target; analogously we label by $\gamma^-_{\max}$ the maximum value of $\gamma_i$ for $ i \in \mathcal{S}_-$, i.e.,
  $  \gamma_{\max}^-= \max_{i \in \mathcal{S}_-} \gamma_i.$
We now prove
\begin{lemma}
    $\gamma^+_{\min} \geq \gamma^-_{\max}$\footnote{Note the above statement is closely related to the convexity of the polytope and of $\omega_{\text{opt}}(\alpha)$, see Fig. \ref{fig:gradients}.}.
\end{lemma}
\begin{proof}
This follows from the optimality of $V$. The proof is by contradiction. To this end, assume $\gamma^+_{min}  <  \gamma^-_{max}$. Then pick vertices $Z_+ \in \{ Z_i \mid i \in \mathcal{S}_+ \}$ and $Z_- \in \{ Z_i \mid i \in \mathcal{S}_-\}$ whose vertex vectors from $V$ have the gradient $\gamma^+_{min}$ and $\gamma^-_{max}$ respectively. Next pick the point ${\bf p}_1 \in \mathcal{P}$ on the line segment between $Z_+$ and $Z_-$ for which the target is the same as that of $V$. This is done by choosing
\begin{equation}
{\bf p}_1 = t Z_{+} + (1-t) Z_{-}
\hspace{1cm}\text{with}\hspace{1cm}
    t=\frac{{\bf a} \cdot (V-Z_{-})}{ {\bf a} \cdot (Z_{+}-Z_{-})} \in (0,1).
\end{equation}
One readily checks that indeed $\mathcal{A}({\bf p}_1) = \mathcal{A}(V)$. The work cost of the transformation from $V$ to ${\bf p}_1$ is the following,
\begin{equation}
\label{equ:work_cost_diff}
\begin{aligned}
    { \bf E} \cdot ({\bf p}_1 - V) &= t \,{\bf E} \cdot (Z_+ - V) + (1-t)\, {\bf E} \cdot (Z_- - V)  \\
        &= t  \; {\bf a} \cdot (Z_+ - V)  \; \gamma^+_{min} + (1-t)  \; {\bf a} \cdot (Z_- - V) \; \gamma^-_{max} \\
        &= \eta \; \left( \gamma^+_{min} - \gamma^-_{max} \right),
\end{aligned}
\end{equation}
where
\begin{equation}
    \eta =\frac{\left[ {\bf a} \cdot (V-Z_-) \right] \left[ {\bf a} \cdot (Z_+ - V) \right] }{ {\bf a} \cdot \left(Z_+ - Z_-\right) } >0.
\end{equation}

Eq.~\ref{equ:work_cost_diff} is therefore negative for $\gamma^+_{min} < \gamma^-_{max}$. But this implies that ${\bf p}_1$ has the same value of the target as $V$ but a lower work cost, contradicting the assumption that $V$ is optimal, i.e., that it is a solution of Eq.~\ref{equ:main_question_pop}.
\end{proof}

Next, we minimise the gradient of the work cost in the transformation from $V$ to some (arbitrary) point ${\bf p
} \in \mathcal{P}$ satisfying $\mathcal{A}({\bf p}) > \mathcal{A}({V})$\footnote{Such a point ${\bf p}$ is guaranteed to exist since $\mathcal{A}(V) \neq \alpha_{\max}$.}, i.e. for any transformation increasing the value of the target function. We express the vector ${\bf p}-V$ w.r.t. the vertex vector of edges at $V$ as in Lemma~\ref{lemma:edgesminimal},
\begin{align}
    {\bf p} - V &= \sum_i r_i \; \left(Z_i - V \right),
\end{align}
with $r_i \geq 0$. 
The change in the target and cost are respectively:

\begin{align} \label{equ:Deltaalpha}
    \Delta \alpha 
        &= \sum_{i \in S_-} r_i \; {\bf a} \cdot (Z_i - V) + \sum_{i \in S_+} r_i \; {\bf a} \cdot (Z_i - V),
\end{align}
\begin{align}\label{equ:Deltaepsilon}
    \Delta \epsilon 
        &= \sum_{i \in S_-} r_i \; {\bf a} \cdot (Z_i - V) \; \gamma_i + \sum_{i \in S_+} r_i \; {\bf a} \cdot (Z_i - V) \; \gamma_i + \sum_{i \in S_=} r_i \; {\bf E} \cdot (Z_i - V).
\end{align}
The first term in the cost expression can be bounded from below using $\gamma^-_{max}$ as
\begin{equation}
    \sum_{i \in \mathcal{S}_-} r_i \; {\bf a} \cdot (Z_i - V) \; \gamma_i  \geq -b_1 \, \gamma^-_{max}
\hspace{1cm}\text{with}\hspace{1cm}
    b_1 =  \sum_{i \in \mathcal{S}_-} r_i \; {\bf a} \cdot (V-Z_i) \geq 0.
\end{equation}
The direction of the inequality follows from $r_i \geq 0$ together with ${\bf a} \cdot (Z_i - V) < 0$ and $\gamma_i \leq \gamma^-_{max}$ for $i \in \mathcal{S}_-$. Similarly, the second term in the cost is bounded from below by
\begin{equation}
    \sum_{i \in \mathcal{S}_+} r_i \; {\bf a} \cdot (Z_i - V) \; \gamma_i \geq b_2 \, \gamma^+_{min} 
    \hspace{1cm}\text{with}\hspace{1cm}
    b_2= \sum_{i \in \mathcal{S}_+} r_i \; {\bf a} \cdot (Z_i - V) \geq 0 .
\end{equation}
Finally, the third term of the cost must be non-negative since for each vertex $Z_i \in \{ Z_i \mid i \in \mathcal{S}_=\}$ by definition $\mathcal{A}(Z_i)=\mathcal{A}({\bf p})$, and so using the optimality of $V$, i.e., that it is a solution of Eq.~\ref{equ:main_question_pop} for $\alpha= \mathcal{A}({\bf p})$,
\begin{equation}
  {\bf E} \cdot Z_i \geq { \bf E} \cdot V
\end{equation}
must hold. We will refer to the third term as $\delta_+$. We therefore have
 $\delta_+ \geq 0$
and
\begin{align}\label{equ:gradientEA}
    \left.\frac{\Delta \epsilon}{\Delta \alpha} \right \rvert_V ({\bf p}_>) &\geq \frac{-b_1 \gamma^-_{max} + b_2 \gamma^+_{min} + \delta_+}{-b_1 + b_2},
\end{align}
where $b_1,b_2,\delta_+ \geq 0$ and $\Delta \alpha = b_2 - b_1 >0$. By rearranging the right-hand side of Eq.~\ref{equ:gradientEA} we get
\begin{align}
    \left.\frac{\Delta \epsilon}{\Delta \alpha} \right \rvert_V ({\bf p}_>)&\geq \gamma_{\min}^+ + \frac{b_1}{b_2-b_1} \left( \gamma_{\min}^+-\gamma_{\max}^- \right) + \frac{\delta_+}{b_2-b_1}\geq \gamma_{\min}^+,
\end{align}
where the last inequality holds since
\begin{equation}
    \frac{b_1}{b_2-b_1} \left( \gamma_{\min}^+ -\gamma_{\max}^- \right) \geq 0 \hspace{1cm }\text{and} \hspace{1cm}
    \frac{\delta_+}{b_2-b_1} \geq 0.
\end{equation}
Importantly by definition of $\gamma_{\min}^+$ there exists an $i \in S_+$ for which $\gamma_i = \gamma_{\min}^+$. In partiular, choosing
\begin{equation} \label{equ:minimalp}
    {\bf p} = r \, (Z_i - V) + V, \quad r \in (0,1],
\end{equation}
we get $\frac{\Delta \epsilon}{\Delta \alpha} = \gamma_{\min}^+$.

We now come back to the case $\mathcal{A}({\bf p})= \alpha_{\min}$. In that case $S_- = \emptyset$. This implies that $\gamma_{\max}^-$ is undefined. One can nevertheless easily lower bound $\frac{\Delta \epsilon}{\Delta \alpha}$ by noticing that the first term of $\Delta \alpha$, resp. $\Delta \epsilon$, in Eq.~\ref{equ:Deltaalpha}, resp. Eq.~\ref{equ:Deltaepsilon}, vanishes. And so Eq.~\ref{equ:gradientEA} simplifies to
\begin{equation}
   \left.\frac{\Delta \epsilon}{\Delta \alpha} \right \rvert_V ({\bf p}_>) \geq \gamma_{\min}^+ + \frac{\delta_+}{b_2} \geq \gamma_{\min}^+,
\end{equation}
where the minimum $\gamma_{\min}^+$ is again achieved for any ${\bf p}$ as in Eq.~\ref{equ:minimalp}. This concludes the proof for the lower bound.

For the upper bound part of the statement, one repeats the above process for the case $\Delta \alpha < 0$. In this case the gradient of $\epsilon$ vs $\alpha$ from the vertex $V$ is \textit{upper bounded by} $\gamma^-_{max}$ and the upper bound is attained for points
\begin{equation}
    {\bf p} = r \, (Z_i - V) + V, \quad r \in (0,1],
\end{equation}
with $i \in S_-$ such that $\gamma_i = \gamma_{\max}^-$. To do so, the analysis is identical until Eq~\ref{equ:gradientEA}. However, now the sign of Eq.~\ref{equ:gradientEA} is reverted due to $\Delta \alpha <0$ so that we have
\begin{equation}
     \left.\frac{\Delta \epsilon}{\Delta \alpha} \right \rvert_V ({\bf p}_<) \leq\frac{-b_1 \gamma^-_{max} + b_2 \gamma^+_{min} + \delta_+}{-b_1 + b_2} = \gamma_{\max}^- - \frac{b_2}{b_1-b_2} \left( \gamma_{\min}^+ - \gamma_{\max}^- \right) - \delta_+ \leq \gamma_{\max}^-.
\end{equation}

%% file: main.bbl
\begin{thebibliography}{41}%
\makeatletter
\providecommand \@ifxundefined [1]{%
 \@ifx{#1\undefined}
}%
\providecommand \@ifnum [1]{%
 \ifnum #1\expandafter \@firstoftwo
 \else \expandafter \@secondoftwo
 \fi
}%
\providecommand \@ifx [1]{%
 \ifx #1\expandafter \@firstoftwo
 \else \expandafter \@secondoftwo
 \fi
}%
\providecommand \natexlab [1]{#1}%
\providecommand \enquote  [1]{``#1''}%
\providecommand \bibnamefont  [1]{#1}%
\providecommand \bibfnamefont [1]{#1}%
\providecommand \citenamefont [1]{#1}%
\providecommand \href@noop [0]{\@secondoftwo}%
\providecommand \href [0]{\begingroup \@sanitize@url \@href}%
\providecommand \@href[1]{\@@startlink{#1}\@@href}%
\providecommand \@@href[1]{\endgroup#1\@@endlink}%
\providecommand \@sanitize@url [0]{\catcode `\\12\catcode `\$12\catcode
  `\&12\catcode `\#12\catcode `\^12\catcode `\_12\catcode `\%12\relax}%
\providecommand \@@startlink[1]{}%
\providecommand \@@endlink[0]{}%
\providecommand \url  [0]{\begingroup\@sanitize@url \@url }%
\providecommand \@url [1]{\endgroup\@href {#1}{\urlprefix }}%
\providecommand \urlprefix  [0]{URL }%
\providecommand \Eprint [0]{\href }%
\providecommand \doibase [0]{https://doi.org/}%
\providecommand \selectlanguage [0]{\@gobble}%
\providecommand \bibinfo  [0]{\@secondoftwo}%
\providecommand \bibfield  [0]{\@secondoftwo}%
\providecommand \translation [1]{[#1]}%
\providecommand \BibitemOpen [0]{}%
\providecommand \bibitemStop [0]{}%
\providecommand \bibitemNoStop [0]{.\EOS\space}%
\providecommand \EOS [0]{\spacefactor3000\relax}%
\providecommand \BibitemShut  [1]{\csname bibitem#1\endcsname}%
\let\auto@bib@innerbib\@empty
\bibitem [{\citenamefont {Guryanova}\ \emph {et~al.}(2020)\citenamefont
  {Guryanova}, \citenamefont {Friis},\ and\ \citenamefont
  {Huber}}]{Guryanova_2020}%
  \BibitemOpen
  \bibfield  {author} {\bibinfo {author} {\bibfnamefont {Y.}~\bibnamefont
  {Guryanova}}, \bibinfo {author} {\bibfnamefont {N.}~\bibnamefont {Friis}},\
  and\ \bibinfo {author} {\bibfnamefont {M.}~\bibnamefont {Huber}},\ }\href
  {https://doi.org/10.22331/q-2020-01-13-222} {\bibfield  {journal} {\bibinfo
  {journal} {{Quantum}}\ }\textbf {\bibinfo {volume} {4}},\ \bibinfo {pages}
  {222} (\bibinfo {year} {2020})},\ \Eprint {https://arxiv.org/abs/1805.11899}
  {arXiv:1805.11899} \BibitemShut {NoStop}%
\bibitem [{\citenamefont {Nielsen}\ and\ \citenamefont
  {Chuang}(2010)}]{nielsen2010quantum}%
  \BibitemOpen
  \bibfield  {author} {\bibinfo {author} {\bibfnamefont {M.~A.}\ \bibnamefont
  {Nielsen}}\ and\ \bibinfo {author} {\bibfnamefont {I.~L.}\ \bibnamefont
  {Chuang}},\ }\href@noop {} {\emph {\bibinfo {title} {Quantum Computation and
  Quantum Information}}},\ \bibinfo {edition} {10th}\ ed.\ (\bibinfo
  {publisher} {Cambridge University Press},\ \bibinfo {address} {Cambridge},\
  \bibinfo {year} {2010})\BibitemShut {NoStop}%
\bibitem [{\citenamefont {Xuereb}\ \emph {et~al.}(2023)\citenamefont {Xuereb},
  \citenamefont {Erker}, \citenamefont {Meier}, \citenamefont {Mitchison},\
  and\ \citenamefont {Huber}}]{Xuereb_2023}%
  \BibitemOpen
  \bibfield  {author} {\bibinfo {author} {\bibfnamefont {J.}~\bibnamefont
  {Xuereb}}, \bibinfo {author} {\bibfnamefont {P.}~\bibnamefont {Erker}},
  \bibinfo {author} {\bibfnamefont {F.}~\bibnamefont {Meier}}, \bibinfo
  {author} {\bibfnamefont {M.~T.}\ \bibnamefont {Mitchison}},\ and\ \bibinfo
  {author} {\bibfnamefont {M.}~\bibnamefont {Huber}},\ }\href
  {https://doi.org/10.1103/PhysRevLett.131.160204} {\bibfield  {journal}
  {\bibinfo  {journal} {Phys. Rev. Lett.}\ }\textbf {\bibinfo {volume} {131}},\
  \bibinfo {pages} {160204} (\bibinfo {year} {2023})},\ \Eprint
  {https://arxiv.org/abs/2301.10767} {arXiv:2301.10767} \BibitemShut {NoStop}%
\bibitem [{\citenamefont {Schwarzhans}\ \emph {et~al.}(2021)\citenamefont
  {Schwarzhans}, \citenamefont {Lock}, \citenamefont {Erker}, \citenamefont
  {Friis},\ and\ \citenamefont {Huber}}]{SchwarzhansLockErkerFriisHuber2021}%
  \BibitemOpen
  \bibfield  {author} {\bibinfo {author} {\bibfnamefont {E.}~\bibnamefont
  {Schwarzhans}}, \bibinfo {author} {\bibfnamefont {M.~P.~E.}\ \bibnamefont
  {Lock}}, \bibinfo {author} {\bibfnamefont {P.}~\bibnamefont {Erker}},
  \bibinfo {author} {\bibfnamefont {N.}~\bibnamefont {Friis}},\ and\ \bibinfo
  {author} {\bibfnamefont {M.}~\bibnamefont {Huber}},\ }\href
  {https://doi.org/10.1103/PhysRevX.11.011046} {\bibfield  {journal} {\bibinfo
  {journal} {Phys. Rev. X}\ }\textbf {\bibinfo {volume} {11}},\ \bibinfo
  {pages} {011046} (\bibinfo {year} {2021})},\ \Eprint
  {https://arxiv.org/abs/2007.01307} {arXiv:2007.01307} \BibitemShut {NoStop}%
\bibitem [{\citenamefont {Bohr~Brask}\ \emph {et~al.}(2015)\citenamefont
  {Bohr~Brask}, \citenamefont {Haack}, \citenamefont {Brunner},\ and\
  \citenamefont {Huber}}]{Bohr-NJP-2015}%
  \BibitemOpen
  \bibfield  {author} {\bibinfo {author} {\bibfnamefont {J.}~\bibnamefont
  {Bohr~Brask}}, \bibinfo {author} {\bibfnamefont {G.}~\bibnamefont {Haack}},
  \bibinfo {author} {\bibfnamefont {N.}~\bibnamefont {Brunner}},\ and\ \bibinfo
  {author} {\bibfnamefont {M.}~\bibnamefont {Huber}},\ }\href
  {https://doi.org/10.1088/1367-2630/17/11/113029} {\bibfield  {journal}
  {\bibinfo  {journal} {New J. Phys.}\ }\textbf {\bibinfo {volume} {17}},\
  \bibinfo {pages} {113029} (\bibinfo {year} {2015})}\BibitemShut {NoStop}%
\bibitem [{\citenamefont {Diotallevi}\ \emph {et~al.}(2024)\citenamefont
  {Diotallevi}, \citenamefont {Annby-Andersson}, \citenamefont {Samuelsson},
  \citenamefont {Tavakoli},\ and\ \citenamefont
  {Bakhshinezhad}}]{Diotallevi-NJP-2024}%
  \BibitemOpen
  \bibfield  {author} {\bibinfo {author} {\bibfnamefont {G.~F.}\ \bibnamefont
  {Diotallevi}}, \bibinfo {author} {\bibfnamefont {B.}~\bibnamefont
  {Annby-Andersson}}, \bibinfo {author} {\bibfnamefont {P.}~\bibnamefont
  {Samuelsson}}, \bibinfo {author} {\bibfnamefont {A.}~\bibnamefont
  {Tavakoli}},\ and\ \bibinfo {author} {\bibfnamefont {P.}~\bibnamefont
  {Bakhshinezhad}},\ }\href
  {https://iopscience.iop.org/article/10.1088/1367-2630/ad3f3d/meta} {\bibfield
   {journal} {\bibinfo  {journal} {New J. Phys.}\ }\textbf {\bibinfo {volume}
  {26}},\ \bibinfo {pages} {053005} (\bibinfo {year} {2024})}\BibitemShut
  {NoStop}%
\bibitem [{\citenamefont {Masanes}\ and\ \citenamefont
  {Oppenheim}(2017{\natexlab{a}})}]{Masanes_2017_Temperature}%
  \BibitemOpen
  \bibfield  {author} {\bibinfo {author} {\bibfnamefont {L.}~\bibnamefont
  {Masanes}}\ and\ \bibinfo {author} {\bibfnamefont {J.}~\bibnamefont
  {Oppenheim}},\ }\href {https://doi.org/10.1038/ncomms14538} {\bibfield
  {journal} {\bibinfo  {journal} {Nature Communications}\ }\textbf {\bibinfo
  {volume} {8}},\ \bibinfo {pages} {14538} (\bibinfo {year}
  {2017}{\natexlab{a}})}\BibitemShut {NoStop}%
\bibitem [{\citenamefont {Nernst}(1906)}]{Nernst-1906}%
  \BibitemOpen
  \bibfield  {author} {\bibinfo {author} {\bibfnamefont {W.}~\bibnamefont
  {Nernst}},\ }\href {https://archive.org/details/mobot31753002089495}
  {\bibfield  {journal} {\bibinfo  {journal} {Sitzungsberichte der
  K{\"{o}}nliglich Preussischen Akademie der Wissenschaften}\ ,\ \bibinfo
  {pages} {933}} (\bibinfo {year} {1906})}\BibitemShut {NoStop}%
\bibitem [{\citenamefont {Allahverdyan}\ \emph {et~al.}(2011)\citenamefont
  {Allahverdyan}, \citenamefont {Hovhannisyan}, \citenamefont {Janzing},\ and\
  \citenamefont {Mahler}}]{Allahverdyan-2011}%
  \BibitemOpen
  \bibfield  {author} {\bibinfo {author} {\bibfnamefont {A.~E.}\ \bibnamefont
  {Allahverdyan}}, \bibinfo {author} {\bibfnamefont {K.~V.}\ \bibnamefont
  {Hovhannisyan}}, \bibinfo {author} {\bibfnamefont {D.}~\bibnamefont
  {Janzing}},\ and\ \bibinfo {author} {\bibfnamefont {G.}~\bibnamefont
  {Mahler}},\ }\href {https://dx.doi.org/10.1103/PhysRevE.84.041109} {\bibfield
   {journal} {\bibinfo  {journal} {Physical Review E}\ }\textbf {\bibinfo
  {volume} {84}} (\bibinfo {year} {2011})},\ \Eprint
  {https://arxiv.org/abs/1107.1044} {arXiv:1107.1044 [cond-mat.stat-mech]}
  \BibitemShut {NoStop}%
\bibitem [{\citenamefont {Clivaz}\ \emph
  {et~al.}(2019{\natexlab{a}})\citenamefont {Clivaz}, \citenamefont {Silva},
  \citenamefont {Haack}, \citenamefont {Brask}, \citenamefont {Brunner},\ and\
  \citenamefont {Huber}}]{Clivaz-2019}%
  \BibitemOpen
  \bibfield  {author} {\bibinfo {author} {\bibfnamefont {F.}~\bibnamefont
  {Clivaz}}, \bibinfo {author} {\bibfnamefont {R.}~\bibnamefont {Silva}},
  \bibinfo {author} {\bibfnamefont {G.}~\bibnamefont {Haack}}, \bibinfo
  {author} {\bibfnamefont {J.~B.}\ \bibnamefont {Brask}}, \bibinfo {author}
  {\bibfnamefont {N.}~\bibnamefont {Brunner}},\ and\ \bibinfo {author}
  {\bibfnamefont {M.}~\bibnamefont {Huber}},\ }\href
  {https://dx.doi.org/10.1103/PhysRevLett.123.170605} {\bibfield  {journal}
  {\bibinfo  {journal} {Physical Review Letters}\ }\textbf {\bibinfo {volume}
  {123}} (\bibinfo {year} {2019}{\natexlab{a}})},\ \Eprint
  {https://arxiv.org/abs/1903.04970} {arXiv:1903.04970 [quant-ph]} \BibitemShut
  {NoStop}%
\bibitem [{\citenamefont {Alhambra}\ \emph {et~al.}(2019)\citenamefont
  {Alhambra}, \citenamefont {Lostaglio},\ and\ \citenamefont
  {Perry}}]{Alhambra-2019}%
  \BibitemOpen
  \bibfield  {author} {\bibinfo {author} {\bibfnamefont {A.~M.}\ \bibnamefont
  {Alhambra}}, \bibinfo {author} {\bibfnamefont {M.}~\bibnamefont
  {Lostaglio}},\ and\ \bibinfo {author} {\bibfnamefont {C.}~\bibnamefont
  {Perry}},\ }\href {https://dx.doi.org/10.22331/q-2019-09-23-188} {\bibfield
  {journal} {\bibinfo  {journal} {Quantum}\ }\textbf {\bibinfo {volume} {3}},\
  \bibinfo {pages} {188} (\bibinfo {year} {2019})},\ \Eprint
  {https://arxiv.org/abs/1807.07974} {arXiv:1807.07974 [quant-ph]} \BibitemShut
  {NoStop}%
\bibitem [{\citenamefont {Boykin}\ \emph {et~al.}(2002)\citenamefont {Boykin},
  \citenamefont {Mor}, \citenamefont {Roychowdhury}, \citenamefont {Vatan},\
  and\ \citenamefont {Vrijen}}]{Boykin-2002}%
  \BibitemOpen
  \bibfield  {author} {\bibinfo {author} {\bibfnamefont {P.~O.}\ \bibnamefont
  {Boykin}}, \bibinfo {author} {\bibfnamefont {T.}~\bibnamefont {Mor}},
  \bibinfo {author} {\bibfnamefont {V.}~\bibnamefont {Roychowdhury}}, \bibinfo
  {author} {\bibfnamefont {F.}~\bibnamefont {Vatan}},\ and\ \bibinfo {author}
  {\bibfnamefont {R.}~\bibnamefont {Vrijen}},\ }\href
  {https://dx.doi.org/10.1073/pnas.241641898} {\bibfield  {journal} {\bibinfo
  {journal} {Proceedings of the National Academy of Sciences}\ }\textbf
  {\bibinfo {volume} {99}},\ \bibinfo {pages} {3388} (\bibinfo {year}
  {2002})},\ \Eprint {https://arxiv.org/abs/quant-ph/0106093}
  {arXiv:quant-ph/0106093 [quant-ph]} \BibitemShut {NoStop}%
\bibitem [{\citenamefont {Park}\ \emph {et~al.}(2016)\citenamefont {Park},
  \citenamefont {Rodriguez-Briones}, \citenamefont {Feng}, \citenamefont
  {Rahimi}, \citenamefont {Baugh},\ and\ \citenamefont {Laflamme}}]{Park-2016}%
  \BibitemOpen
  \bibfield  {author} {\bibinfo {author} {\bibfnamefont {D.~K.}\ \bibnamefont
  {Park}}, \bibinfo {author} {\bibfnamefont {N.~A.}\ \bibnamefont
  {Rodriguez-Briones}}, \bibinfo {author} {\bibfnamefont {G.}~\bibnamefont
  {Feng}}, \bibinfo {author} {\bibfnamefont {R.}~\bibnamefont {Rahimi}},
  \bibinfo {author} {\bibfnamefont {J.}~\bibnamefont {Baugh}},\ and\ \bibinfo
  {author} {\bibfnamefont {R.}~\bibnamefont {Laflamme}},\ }\bibinfo {title}
  {Heat bath algorithmic cooling with spins: Review and prospects},\ in\ \href
  {https://dx.doi.org/10.1007/978-1-4939-3658-8_8} {\emph {\bibinfo {booktitle}
  {Electron Spin Resonance (ESR) Based Quantum Computing}}},\ \bibinfo {editor}
  {edited by\ \bibinfo {editor} {\bibfnamefont {T.}~\bibnamefont {Takui}},
  \bibinfo {editor} {\bibfnamefont {L.}~\bibnamefont {Berliner}},\ and\
  \bibinfo {editor} {\bibfnamefont {G.}~\bibnamefont {Hanson}}}\ (\bibinfo
  {publisher} {Springer New York},\ \bibinfo {address} {New York, NY},\
  \bibinfo {year} {2016})\ pp.\ \bibinfo {pages} {227--255},\ \Eprint
  {https://arxiv.org/abs/1501.00952} {arXiv:1501.00952 [quant-ph]} \BibitemShut
  {NoStop}%
\bibitem [{\citenamefont {Rodríguez-Briones}\ \emph
  {et~al.}(2017)\citenamefont {Rodríguez-Briones}, \citenamefont
  {Martín-Martínez}, \citenamefont {Kempf},\ and\ \citenamefont
  {Laflamme}}]{Rodriguez-Briones-2017}%
  \BibitemOpen
  \bibfield  {author} {\bibinfo {author} {\bibfnamefont {N.~A.}\ \bibnamefont
  {Rodríguez-Briones}}, \bibinfo {author} {\bibfnamefont {E.}~\bibnamefont
  {Martín-Martínez}}, \bibinfo {author} {\bibfnamefont {A.}~\bibnamefont
  {Kempf}},\ and\ \bibinfo {author} {\bibfnamefont {R.}~\bibnamefont
  {Laflamme}},\ }\href {https://dx.doi.org/10.1103/PhysRevLett.119.050502}
  {\bibfield  {journal} {\bibinfo  {journal} {Phys. Rev. Lett.}\ }\textbf
  {\bibinfo {volume} {119}},\ \bibinfo {pages} {050502} (\bibinfo {year}
  {2017})},\ \Eprint {https://arxiv.org/abs/1703.03816} {arXiv:1703.03816
  [quant-ph]} \BibitemShut {NoStop}%
\bibitem [{\citenamefont {Rodríguez-Briones}\ and\ \citenamefont
  {Laflamme}(2016)}]{Rodriguez-Briones-2016}%
  \BibitemOpen
  \bibfield  {author} {\bibinfo {author} {\bibfnamefont {N.~A.}\ \bibnamefont
  {Rodríguez-Briones}}\ and\ \bibinfo {author} {\bibfnamefont
  {R.}~\bibnamefont {Laflamme}},\ }\href
  {https://dx.doi.org/10.1103/PhysRevLett.116.170501} {\bibfield  {journal}
  {\bibinfo  {journal} {Phys. Rev. Lett.}\ }\textbf {\bibinfo {volume} {116}},\
  \bibinfo {pages} {170501} (\bibinfo {year} {2016})},\ \Eprint
  {https://arxiv.org/abs/1412.6637} {arXiv:1412.6637 [quant-ph]} \BibitemShut
  {NoStop}%
\bibitem [{\citenamefont {Schulman}\ and\ \citenamefont
  {Vazirani}(1999)}]{Schulman-1999}%
  \BibitemOpen
  \bibfield  {author} {\bibinfo {author} {\bibfnamefont {L.~J.}\ \bibnamefont
  {Schulman}}\ and\ \bibinfo {author} {\bibfnamefont {U.~V.}\ \bibnamefont
  {Vazirani}},\ }\href {https://dx.doi.org/10.1145/301250.301332} {\bibfield
  {journal} {\bibinfo  {journal} {Proceedings of the Thirty-First Annual ACM
  Symposium on Theory of Computing}\ }\bibinfo {series} {STOC ’99},\ \bibinfo
  {pages} {322–329} (\bibinfo {year} {1999})},\ \Eprint
  {https://arxiv.org/abs/quant-ph/9804060} {arXiv:quant-ph/9804060 [quant-ph]}
  \BibitemShut {NoStop}%
\bibitem [{\citenamefont {Taranto}\ \emph {et~al.}(2020)\citenamefont
  {Taranto}, \citenamefont {Bakhshinezhad}, \citenamefont {Schüttelkopf},
  \citenamefont {Clivaz},\ and\ \citenamefont
  {Huber}}]{Taranto_2020ExponentialImprovement}%
  \BibitemOpen
  \bibfield  {author} {\bibinfo {author} {\bibfnamefont {P.}~\bibnamefont
  {Taranto}}, \bibinfo {author} {\bibfnamefont {F.}~\bibnamefont
  {Bakhshinezhad}}, \bibinfo {author} {\bibfnamefont {P.}~\bibnamefont
  {Schüttelkopf}}, \bibinfo {author} {\bibfnamefont {F.}~\bibnamefont
  {Clivaz}},\ and\ \bibinfo {author} {\bibfnamefont {M.}~\bibnamefont
  {Huber}},\ }\href {http://dx.doi.org/10.1103/PhysRevApplied.14.054005}
  {\bibfield  {journal} {\bibinfo  {journal} {Physical Review Applied}\
  }\textbf {\bibinfo {volume} {14}} (\bibinfo {year} {2020})}\BibitemShut
  {NoStop}%
\bibitem [{\citenamefont {Masanes}\ and\ \citenamefont
  {Oppenheim}(2017{\natexlab{b}})}]{Masanes-2017}%
  \BibitemOpen
  \bibfield  {author} {\bibinfo {author} {\bibfnamefont {L.}~\bibnamefont
  {Masanes}}\ and\ \bibinfo {author} {\bibfnamefont {J.}~\bibnamefont
  {Oppenheim}},\ }\href {https://dx.doi.org/10.1038/ncomms14538} {\bibfield
  {journal} {\bibinfo  {journal} {Nature Communications}\ }\textbf {\bibinfo
  {volume} {8}} (\bibinfo {year} {2017}{\natexlab{b}})},\ \Eprint
  {https://arxiv.org/abs/1412.3828} {arXiv:1412.3828 [quant-ph]} \BibitemShut
  {NoStop}%
\bibitem [{\citenamefont {Wilming}\ and\ \citenamefont
  {Gallego}(2017)}]{Wilming-2017}%
  \BibitemOpen
  \bibfield  {author} {\bibinfo {author} {\bibfnamefont {H.}~\bibnamefont
  {Wilming}}\ and\ \bibinfo {author} {\bibfnamefont {R.}~\bibnamefont
  {Gallego}},\ }\href {https://dx.doi.org/10.1103/PhysRevX.7.041033} {\bibfield
   {journal} {\bibinfo  {journal} {Physical Review X}\ }\textbf {\bibinfo
  {volume} {7}} (\bibinfo {year} {2017})},\ \Eprint
  {https://arxiv.org/abs/1701.07478} {arXiv:1701.07478 [quant-ph]} \BibitemShut
  {NoStop}%
\bibitem [{\citenamefont {Scharlau}\ and\ \citenamefont
  {Mueller}(2018)}]{Scharlau-2018}%
  \BibitemOpen
  \bibfield  {author} {\bibinfo {author} {\bibfnamefont {J.}~\bibnamefont
  {Scharlau}}\ and\ \bibinfo {author} {\bibfnamefont {M.~P.}\ \bibnamefont
  {Mueller}},\ }\href {https://dx.doi.org/10.22331/q-2018-02-22-54} {\bibfield
  {journal} {\bibinfo  {journal} {Quantum}\ }\textbf {\bibinfo {volume} {2}},\
  \bibinfo {pages} {54} (\bibinfo {year} {2018})},\ \Eprint
  {https://arxiv.org/abs/1605.06092} {arXiv:1605.06092 [quant-ph]} \BibitemShut
  {NoStop}%
\bibitem [{\citenamefont {Raeisi}\ and\ \citenamefont
  {Mosca}(2015)}]{Raeisi-2015}%
  \BibitemOpen
  \bibfield  {author} {\bibinfo {author} {\bibfnamefont {S.}~\bibnamefont
  {Raeisi}}\ and\ \bibinfo {author} {\bibfnamefont {M.}~\bibnamefont {Mosca}},\
  }\href {https://dx.doi.org/10.1103/PhysRevLett.114.100404} {\bibfield
  {journal} {\bibinfo  {journal} {Phys. Rev. Lett.}\ }\textbf {\bibinfo
  {volume} {114}},\ \bibinfo {pages} {100404} (\bibinfo {year} {2015})},\
  \Eprint {https://arxiv.org/abs/1407.3232} {arXiv:1407.3232 [quant-ph]}
  \BibitemShut {NoStop}%
\bibitem [{\citenamefont {Soldati}\ \emph {et~al.}(2024)\citenamefont
  {Soldati}, \citenamefont {Dasari}, \citenamefont {Wrachtrup},\ and\
  \citenamefont {Lutz}}]{soldati2024}%
  \BibitemOpen
  \bibfield  {author} {\bibinfo {author} {\bibfnamefont {R.~R.}\ \bibnamefont
  {Soldati}}, \bibinfo {author} {\bibfnamefont {D.~B.~R.}\ \bibnamefont
  {Dasari}}, \bibinfo {author} {\bibfnamefont {J.}~\bibnamefont {Wrachtrup}},\
  and\ \bibinfo {author} {\bibfnamefont {E.}~\bibnamefont {Lutz}},\ }\href
  {https://arxiv.org/abs/2410.18201} {\bibinfo {title} {Cooling limits of
  coherent refrigerators}} (\bibinfo {year} {2024}),\ \Eprint
  {https://arxiv.org/abs/2410.18201} {arXiv:2410.18201 [quant-ph]} \BibitemShut
  {NoStop}%
\bibitem [{\citenamefont {Brandão}\ \emph {et~al.}(2015)\citenamefont
  {Brandão}, \citenamefont {Horodecki}, \citenamefont {Ng}, \citenamefont
  {Oppenheim},\ and\ \citenamefont {Wehner}}]{Brandao-2015}%
  \BibitemOpen
  \bibfield  {author} {\bibinfo {author} {\bibfnamefont {F.}~\bibnamefont
  {Brandão}}, \bibinfo {author} {\bibfnamefont {M.}~\bibnamefont {Horodecki}},
  \bibinfo {author} {\bibfnamefont {N.}~\bibnamefont {Ng}}, \bibinfo {author}
  {\bibfnamefont {J.}~\bibnamefont {Oppenheim}},\ and\ \bibinfo {author}
  {\bibfnamefont {S.}~\bibnamefont {Wehner}},\ }\href
  {https://dx.doi.org/10.1073/pnas.1411728112} {\bibfield  {journal} {\bibinfo
  {journal} {Proceedings of the National Academy of Sciences}\ }\textbf
  {\bibinfo {volume} {112}},\ \bibinfo {pages} {3275–3279} (\bibinfo {year}
  {2015})},\ \Eprint {https://arxiv.org/abs/1305.5278} {arXiv:1305.5278
  [quant-ph]} \BibitemShut {NoStop}%
\bibitem [{\citenamefont {Reeb}\ and\ \citenamefont {Wolf}(2014)}]{Reeb-2014}%
  \BibitemOpen
  \bibfield  {author} {\bibinfo {author} {\bibfnamefont {D.}~\bibnamefont
  {Reeb}}\ and\ \bibinfo {author} {\bibfnamefont {M.~M.}\ \bibnamefont
  {Wolf}},\ }\href {https://dx.doi.org/10.1088/1367-2630/16/10/103011}
  {\bibfield  {journal} {\bibinfo  {journal} {New Journal of Physics}\ }\textbf
  {\bibinfo {volume} {16}},\ \bibinfo {pages} {103011} (\bibinfo {year}
  {2014})},\ \Eprint {https://arxiv.org/abs/1306.4352} {arXiv:1306.4352
  [quant-ph]} \BibitemShut {NoStop}%
\bibitem [{\citenamefont {M{\"u}ller}(2018)}]{Mueller-2018}%
  \BibitemOpen
  \bibfield  {author} {\bibinfo {author} {\bibfnamefont {M.~P.}\ \bibnamefont
  {M{\"u}ller}},\ }\href {https://dx.doi.org/10.1103/PhysRevX.8.041051}
  {\bibfield  {journal} {\bibinfo  {journal} {Phys. Rev. X}\ }\textbf {\bibinfo
  {volume} {8}},\ \bibinfo {pages} {041051} (\bibinfo {year} {2018})},\ \Eprint
  {https://arxiv.org/abs/1707.03451} {arXiv:1707.03451 [quant-ph]} \BibitemShut
  {NoStop}%
\bibitem [{\citenamefont {Ćwikliński}\ \emph {et~al.}(2015)\citenamefont
  {Ćwikliński}, \citenamefont {Studziński}, \citenamefont {Horodecki},\ and\
  \citenamefont {Oppenheim}}]{Cwiklinski-2015}%
  \BibitemOpen
  \bibfield  {author} {\bibinfo {author} {\bibfnamefont {P.}~\bibnamefont
  {Ćwikliński}}, \bibinfo {author} {\bibfnamefont {M.}~\bibnamefont
  {Studziński}}, \bibinfo {author} {\bibfnamefont {M.}~\bibnamefont
  {Horodecki}},\ and\ \bibinfo {author} {\bibfnamefont {J.}~\bibnamefont
  {Oppenheim}},\ }\href {https://dx.doi.org/10.1103/PhysRevLett.115.210403}
  {\bibfield  {journal} {\bibinfo  {journal} {Physical Review Letters}\
  }\textbf {\bibinfo {volume} {115}} (\bibinfo {year} {2015})},\ \Eprint
  {https://arxiv.org/abs/1405.5029} {arXiv:1405.5029 [quant-ph]} \BibitemShut
  {NoStop}%
\bibitem [{\citenamefont {Landauer}(1961)}]{Landauer_1961}%
  \BibitemOpen
  \bibfield  {author} {\bibinfo {author} {\bibfnamefont {R.}~\bibnamefont
  {Landauer}},\ }\href {https://doi.org/10.1147/rd.53.0183} {\bibfield
  {journal} {\bibinfo  {journal} {IBM J. Res. Dev.}\ }\textbf {\bibinfo
  {volume} {5}},\ \bibinfo {pages} {183} (\bibinfo {year} {1961})}\BibitemShut
  {NoStop}%
\bibitem [{\citenamefont {Skrzypczyk}\ \emph {et~al.}(2014)\citenamefont
  {Skrzypczyk}, \citenamefont {Short},\ and\ \citenamefont
  {Popescu}}]{Skrzypczyk-2014}%
  \BibitemOpen
  \bibfield  {author} {\bibinfo {author} {\bibfnamefont {P.}~\bibnamefont
  {Skrzypczyk}}, \bibinfo {author} {\bibfnamefont {A.~J.}\ \bibnamefont
  {Short}},\ and\ \bibinfo {author} {\bibfnamefont {S.}~\bibnamefont
  {Popescu}},\ }\href {https://dx.doi.org/10.1038/ncomms5185} {\bibfield
  {journal} {\bibinfo  {journal} {Nature Communications}\ }\textbf {\bibinfo
  {volume} {5}} (\bibinfo {year} {2014})},\ \Eprint
  {https://arxiv.org/abs/1307.1558} {arXiv:1307.1558 [quant-ph]} \BibitemShut
  {NoStop}%
\bibitem [{\citenamefont {Taranto}\ \emph {et~al.}(2021)\citenamefont
  {Taranto}, \citenamefont {Bakhshinezhad}, \citenamefont {Bluhm},
  \citenamefont {Silva}, \citenamefont {Friis}, \citenamefont {Lock},
  \citenamefont {Vitagliano}, \citenamefont {Binder}, \citenamefont {Debarba},
  \citenamefont {Schwarzhans}, \citenamefont {Clivaz},\ and\ \citenamefont
  {Huber}}]{Taranto-2021}%
  \BibitemOpen
  \bibfield  {author} {\bibinfo {author} {\bibfnamefont {P.}~\bibnamefont
  {Taranto}}, \bibinfo {author} {\bibfnamefont {F.}~\bibnamefont
  {Bakhshinezhad}}, \bibinfo {author} {\bibfnamefont {A.}~\bibnamefont
  {Bluhm}}, \bibinfo {author} {\bibfnamefont {R.}~\bibnamefont {Silva}},
  \bibinfo {author} {\bibfnamefont {N.}~\bibnamefont {Friis}}, \bibinfo
  {author} {\bibfnamefont {M.~P.~E.}\ \bibnamefont {Lock}}, \bibinfo {author}
  {\bibfnamefont {G.}~\bibnamefont {Vitagliano}}, \bibinfo {author}
  {\bibfnamefont {F.~C.}\ \bibnamefont {Binder}}, \bibinfo {author}
  {\bibfnamefont {T.}~\bibnamefont {Debarba}}, \bibinfo {author} {\bibfnamefont
  {E.}~\bibnamefont {Schwarzhans}}, \bibinfo {author} {\bibfnamefont
  {F.}~\bibnamefont {Clivaz}},\ and\ \bibinfo {author} {\bibfnamefont
  {M.}~\bibnamefont {Huber}},\ }\href@noop {} {\bibinfo {title} {Landauer vs.
  nernst: What is the true cost of cooling a quantum system?}} (\bibinfo {year}
  {2021}),\ \Eprint {https://arxiv.org/abs/2106.05151} {arXiv:2106.05151
  [quant-ph]} \BibitemShut {NoStop}%
\bibitem [{\citenamefont {Clivaz}\ \emph
  {et~al.}(2019{\natexlab{b}})\citenamefont {Clivaz}, \citenamefont {Silva},
  \citenamefont {Haack}, \citenamefont {Brask}, \citenamefont {Brunner},\ and\
  \citenamefont {Huber}}]{Clivaz-2019bis}%
  \BibitemOpen
  \bibfield  {author} {\bibinfo {author} {\bibfnamefont {F.}~\bibnamefont
  {Clivaz}}, \bibinfo {author} {\bibfnamefont {R.}~\bibnamefont {Silva}},
  \bibinfo {author} {\bibfnamefont {G.}~\bibnamefont {Haack}}, \bibinfo
  {author} {\bibfnamefont {J.~B.}\ \bibnamefont {Brask}}, \bibinfo {author}
  {\bibfnamefont {N.}~\bibnamefont {Brunner}},\ and\ \bibinfo {author}
  {\bibfnamefont {M.}~\bibnamefont {Huber}},\ }\href
  {https://dx.doi.org/10.1103/PhysRevE.100.042130} {\bibfield  {journal}
  {\bibinfo  {journal} {Physical Review E}\ }\textbf {\bibinfo {volume} {100}}
  (\bibinfo {year} {2019}{\natexlab{b}})},\ \Eprint
  {https://arxiv.org/abs/1710.11624} {arXiv:1710.11624 [quant-ph]} \BibitemShut
  {NoStop}%
\bibitem [{\citenamefont {Taranto}\ \emph {et~al.}(2024)\citenamefont
  {Taranto}, \citenamefont {Lipka-Bartosik}, \citenamefont
  {Rodríguez-Briones}, \citenamefont {Perarnau-Llobet}, \citenamefont {Friis},
  \citenamefont {Huber},\ and\ \citenamefont
  {Bakhshinezhad}}]{taranto_2024_FiniteResources}%
  \BibitemOpen
  \bibfield  {author} {\bibinfo {author} {\bibfnamefont {P.}~\bibnamefont
  {Taranto}}, \bibinfo {author} {\bibfnamefont {P.}~\bibnamefont
  {Lipka-Bartosik}}, \bibinfo {author} {\bibfnamefont {N.~A.}\ \bibnamefont
  {Rodríguez-Briones}}, \bibinfo {author} {\bibfnamefont {M.}~\bibnamefont
  {Perarnau-Llobet}}, \bibinfo {author} {\bibfnamefont {N.}~\bibnamefont
  {Friis}}, \bibinfo {author} {\bibfnamefont {M.}~\bibnamefont {Huber}},\ and\
  \bibinfo {author} {\bibfnamefont {P.}~\bibnamefont {Bakhshinezhad}},\ }\href
  {https://arxiv.org/abs/2404.06649} {\bibinfo {title} {Efficiently cooling
  quantum systems with finite resources: Insights from thermodynamic geometry}}
  (\bibinfo {year} {2024}),\ \Eprint {https://arxiv.org/abs/2404.06649}
  {arXiv:2404.06649 [quant-ph]} \BibitemShut {NoStop}%
\bibitem [{\citenamefont {Lipka-Bartosik}\ and\ \citenamefont
  {Perarnau-Llobet}(2024)}]{lipkabartosik2024MInDissipationinInteracting}%
  \BibitemOpen
  \bibfield  {author} {\bibinfo {author} {\bibfnamefont {P.}~\bibnamefont
  {Lipka-Bartosik}}\ and\ \bibinfo {author} {\bibfnamefont {M.}~\bibnamefont
  {Perarnau-Llobet}},\ }\href {https://arxiv.org/abs/2411.00944} {\bibinfo
  {title} {Minimizing dissipation via interacting environments: Quadratic
  convergence to landauer bound}} (\bibinfo {year} {2024}),\ \Eprint
  {https://arxiv.org/abs/2411.00944} {arXiv:2411.00944 [quant-ph]} \BibitemShut
  {NoStop}%
\bibitem [{\citenamefont {Rolandi}\ \emph {et~al.}(2023)\citenamefont
  {Rolandi}, \citenamefont {Abiuso},\ and\ \citenamefont
  {Perarnau-Llobet}}]{Rolandi_2023}%
  \BibitemOpen
  \bibfield  {author} {\bibinfo {author} {\bibfnamefont {A.}~\bibnamefont
  {Rolandi}}, \bibinfo {author} {\bibfnamefont {P.}~\bibnamefont {Abiuso}},\
  and\ \bibinfo {author} {\bibfnamefont {M.}~\bibnamefont {Perarnau-Llobet}},\
  }\href {http://dx.doi.org/10.1103/PhysRevLett.131.210401} {\bibfield
  {journal} {\bibinfo  {journal} {Physical Review Letters}\ }\textbf {\bibinfo
  {volume} {131}} (\bibinfo {year} {2023})}\BibitemShut {NoStop}%
\bibitem [{\citenamefont {Marshall}\ \emph {et~al.}(2011)\citenamefont
  {Marshall}, \citenamefont {Olkin},\ and\ \citenamefont
  {Arnold}}]{Marshall-2011}%
  \BibitemOpen
  \bibfield  {author} {\bibinfo {author} {\bibfnamefont {A.~W.}\ \bibnamefont
  {Marshall}}, \bibinfo {author} {\bibfnamefont {I.}~\bibnamefont {Olkin}},\
  and\ \bibinfo {author} {\bibfnamefont {B.~C.}\ \bibnamefont {Arnold}},\
  }\href {https://dx.doi.org/10.1007/978-0-387-68276-1} {\emph {\bibinfo
  {title} {Inequalities: Theory of Majorization and its Applications}}},\
  \bibinfo {edition} {2nd}\ ed.,\ Vol.\ \bibinfo {volume} {143}\ (\bibinfo
  {publisher} {Springer},\ \bibinfo {year} {2011})\BibitemShut {NoStop}%
\bibitem [{\citenamefont {Brøndsted}(1983)}]{Brondsted-1983}%
  \BibitemOpen
  \bibfield  {author} {\bibinfo {author} {\bibfnamefont {A.}~\bibnamefont
  {Brøndsted}},\ }\href
  {https://doi.org/https://doi.org/10.1007/978-1-4612-1148-8} {\emph {\bibinfo
  {title} {An Introduction to Convex Polytopes}}},\ \bibinfo {edition} {1st}\
  ed.\ (\bibinfo  {publisher} {Springer New York, NY},\ \bibinfo {year}
  {1983})\BibitemShut {NoStop}%
\bibitem [{\citenamefont {Ziegler}(1995)}]{Ziegler-1995}%
  \BibitemOpen
  \bibfield  {author} {\bibinfo {author} {\bibfnamefont {G.~M.}\ \bibnamefont
  {Ziegler}},\ }\href
  {https://doi.org/https://doi.org/10.1007/978-1-4613-8431-1} {\emph {\bibinfo
  {title} {Lectures on Polytopes}}},\ \bibinfo {edition} {1st}\ ed.\ (\bibinfo
  {publisher} {Springer New York, NY},\ \bibinfo {year} {1995})\BibitemShut
  {NoStop}%
\bibitem [{\citenamefont {Grünbaum}(2003)}]{Gruenbaum-2003}%
  \BibitemOpen
  \bibfield  {author} {\bibinfo {author} {\bibfnamefont {B.}~\bibnamefont
  {Grünbaum}},\ }\href
  {https://doi.org/https://doi.org/10.1007/978-1-4613-0019-9} {\emph {\bibinfo
  {title} {Convex Polytopes}}},\ \bibinfo {edition} {2nd}\ ed.\ (\bibinfo
  {publisher} {Springer New York, NY},\ \bibinfo {year} {2003})\BibitemShut
  {NoStop}%
\bibitem [{\citenamefont {Barvinok}(2002)}]{Barvinok-2002}%
  \BibitemOpen
  \bibfield  {author} {\bibinfo {author} {\bibfnamefont {A.}~\bibnamefont
  {Barvinok}},\ }\href {https://doi.org/https://doi.org/10.1090/gsm/054} {\emph
  {\bibinfo {title} {A course in convexity}}},\ \bibinfo {edition} {1st}\ ed.\
  (\bibinfo  {publisher} {American Mathematical Society},\ \bibinfo {year}
  {2002})\BibitemShut {NoStop}%
\bibitem [{\citenamefont {Lipka-Bartosik}\ \emph {et~al.}(2023)\citenamefont
  {Lipka-Bartosik}, \citenamefont {Diotallevi},\ and\ \citenamefont
  {Bakhshinezhad}}]{lipkabartosik-2023}%
  \BibitemOpen
  \bibfield  {author} {\bibinfo {author} {\bibfnamefont {P.}~\bibnamefont
  {Lipka-Bartosik}}, \bibinfo {author} {\bibfnamefont {G.~F.}\ \bibnamefont
  {Diotallevi}},\ and\ \bibinfo {author} {\bibfnamefont {P.}~\bibnamefont
  {Bakhshinezhad}},\ }\href@noop {} {\bibinfo {title} {Fundamental limits on
  anomalous energy flows in correlated quantum systems}} (\bibinfo {year}
  {2023}),\ \Eprint {https://arxiv.org/abs/2307.03828} {arXiv:2307.03828
  [quant-ph]} \BibitemShut {NoStop}%
\bibitem [{\citenamefont {Brunner}\ \emph {et~al.}(2012)\citenamefont
  {Brunner}, \citenamefont {Linden}, \citenamefont {Popescu},\ and\
  \citenamefont {Skrzypczyk}}]{Brunner2012}%
  \BibitemOpen
  \bibfield  {author} {\bibinfo {author} {\bibfnamefont {N.}~\bibnamefont
  {Brunner}}, \bibinfo {author} {\bibfnamefont {N.}~\bibnamefont {Linden}},
  \bibinfo {author} {\bibfnamefont {S.}~\bibnamefont {Popescu}},\ and\ \bibinfo
  {author} {\bibfnamefont {P.}~\bibnamefont {Skrzypczyk}},\ }\href
  {https://doi.org/10.1103/PhysRevE.85.051117} {\bibfield  {journal} {\bibinfo
  {journal} {Phys. Rev. E}\ }\textbf {\bibinfo {volume} {85}},\ \bibinfo
  {pages} {051117} (\bibinfo {year} {2012})}\BibitemShut {NoStop}%
\bibitem [{\citenamefont {Silva}\ \emph {et~al.}(2016)\citenamefont {Silva},
  \citenamefont {Manzano}, \citenamefont {Skrzypczyk},\ and\ \citenamefont
  {Brunner}}]{Silva2016}%
  \BibitemOpen
  \bibfield  {author} {\bibinfo {author} {\bibfnamefont {R.}~\bibnamefont
  {Silva}}, \bibinfo {author} {\bibfnamefont {G.}~\bibnamefont {Manzano}},
  \bibinfo {author} {\bibfnamefont {P.}~\bibnamefont {Skrzypczyk}},\ and\
  \bibinfo {author} {\bibfnamefont {N.}~\bibnamefont {Brunner}},\ }\href
  {https://doi.org/10.1103/PhysRevE.94.032120} {\bibfield  {journal} {\bibinfo
  {journal} {Phys. Rev. E}\ }\textbf {\bibinfo {volume} {94}},\ \bibinfo
  {pages} {032120} (\bibinfo {year} {2016})}\BibitemShut {NoStop}%
\end{thebibliography}%
